\documentclass[a4paper]{article}

\usepackage[margin=1in]{geometry}

\usepackage[T1]{fontenc}
\usepackage[utf8]{inputenc}
\usepackage{lmodern}

\usepackage{microtype}

\usepackage{amsmath,amssymb,amsthm}

\usepackage{tikz}
\usetikzlibrary{arrows,automata,positioning}

\tikzset{
  position/.style args={#1:#2 from #3}{
    at=(#3.#1),
    anchor=#1+180,
    shift=(#1:#2)
  }
}

\usepackage{enumitem}
\usepackage{tabularx}

\usepackage{lineno}

\usepackage{natbib}

\usepackage{xcolor}
\usepackage{hyperref}

\hypersetup{
  colorlinks=true,
  linkcolor=blue,
  citecolor=blue,
  urlcolor=blue
}

\theoremstyle{plain}

\newtheorem{theorem}{Theorem}
\newtheorem{proposition}{Proposition}
\newtheorem{lemma}{Lemma}
\newtheorem{corollary}{Corollary}
\newtheorem{fact}{Fact}

\theoremstyle{definition}

\newtheorem{definition}{Definition}
\newtheorem{example}{Example}

\theoremstyle{remark}

\newcommand{\FGCGF}{\mathtt{GCGF}}

\newcommand{\FXX}{\mathtt{X}}

\newcommand{\FCL}{\mathsf{CL}}

\newcommand{\FAG}{\mathtt{AG}}
\newcommand{\FAC}{\mathtt{AC}}
\newcommand{\FST}{\mathtt{ST}}

\newcommand{\Fav}{\mathtt{av}}
\newcommand{\Fout}{\mathsf{out}}

\newcommand{\FJA}{\mathtt{JA}}

\newcommand{\FCC}{\mathrm{C}}

\newcommand{\FDD}{\mathrm{D}}
\newcommand{\FEE}{\mathrm{E}}

\newcommand{\CAG}{\mathtt{CAG}}

\newcommand{\PAC}{\mathsf{AC}}
\newcommand{\PAL}{\alpha}

\newcommand{\FACNF}{\mathtt{NF}^\mathtt{ac}}
\newcommand{\FALNF}{\mathtt{NF}^\alpha}

\newcommand{\Facnei}{\mathrm{N}^\mathrm{ac}}
\newcommand{\Falnei}{\mathrm{N}^\alpha}
\newcommand{\Falneinc}{\mathrm{CoreN}^\alpha}

\newcommand{\FALEF}{\mathrm{EF}^\alpha}
\newcommand{\FACEF}{\mathrm{EF}^\mathrm{ac}}

\newcommand{\FAF}{\mathtt{AF}}

\newcommand{\FES}{\mathrm{ES}}

\newcommand{\Fdefs}[1]{\textbf{#1}}
\newcommand{\putaway}[1]{}
\newcommand{\ab}{\allowbreak}

\title{Representation theorems for actual and alpha powers \\ over general concurrent game frames \\ without assuming
independence of agents}

\author{
Zixuan Chen$^{1,2}$, Fengkui Ju$^{3,4}$\thanks{Corresponding author.}, and Thomas {\AA}gotnes$^{5,6,7}$\\[5pt]
{\small $^{1}$Institute for Logic, Language and Computation, University of Amsterdam, Amsterdam}\\
{\small The Netherlands}\\[2.5pt]
{\small $^{2}$\href{mailto:zixuan.chen21@outlook.com}{zixuan.chen21@outlook.com}}\\[2.5pt]
{\small $^{3}$School of Philosophy, Beijing Normal University, Beijing, China}\\[2.5pt]
{\small $^{4}$\href{mailto:fengkui.ju@bnu.edu.cn}{fengkui.ju@bnu.edu.cn}}\\[2.5pt]
{\small $^{5}$Department of Information Science and Media Studies, University of Bergen, Bergen, Norway}\\[2.5pt]
{\small $^{6}$School of Philosophy, Shanxi University, Taiyuan, China}\\[2.5pt]
{\small $^{7}$\href{mailto:thomas.agotnes@uib.no}{thomas.agotnes@uib.no}}
}

\date{}

\begin{document}

\maketitle

\setlist[enumerate]{itemsep=3pt, topsep=5pt, parsep=3pt, partopsep=0pt}
\setlist[itemize]{itemsep=3pt, topsep=5pt, parsep=3pt, partopsep=0pt}

%\linenumbers

\begin{abstract}

Concurrent game frames are a standard semantic framework for logics of strategic reasoning. Two notions of coalition power can be derived from such frames: \emph{alpha powers} and \emph{actual powers}. 
An alpha power of a coalition is a set of possible futures such that the coalition has an action that forces the resulting future to lie in that set.
An actual power of a coalition is a set of possible futures satisfying the following condition: the coalition has an action such that (1) the action forces the resulting future to lie in the set, and (2) every future in the set is compatible with that action.

Recent generalizations of concurrent game frames separate three structural
assumptions built into the standard model: seriality, independence of agents,
and determinism. This yields eight classes of general concurrent game frames.

In this paper, we prove that for actual powers, the four classes of general concurrent game frames, where independence of agents is not assumed, are representable by four corresponding classes of neighborhood frames. Building on this result, we show that for alpha powers, the same four classes of general concurrent game frames are likewise representable by four corresponding classes of neighborhood frames.

\end{abstract}

%%%%%%%%%%%%%%%
%%%%%%%%%%%%%%%
\section{Introduction}
\label{sec:01-introduction}

In this section, we begin by recalling concurrent game frames, alpha and actual powers,
and neighborhood frames. We then explain the difference between alpha powers
and actual powers, formulate the neighborhood-frame representation problem for
both notions of power, and discuss the significance of such representation
results. After that, we introduce the eight classes of general concurrent game
frames obtained by varying seriality, independence of agents, and determinism.
The section closes with an outline of the main results of the paper.

\emph{Much of
the background material in this section is drawn from
\cite{chen_representation_2026}; it is recalled here to make the present paper
self-contained.}

%%%%%%%%%%%%%%%
%%%%%%%%%%%%%%%
\subsection{Concurrent game frames, alpha and actual powers, and alpha and actual neighborhood frames}
\label{subsec:01-01-concurrent-game-frames-alpha-and-actual-powers-and-alpha-and-actual-neighborhood-frames}

\emph{Concurrent game frames} are a standard semantic framework for logics of strategic reasoning. Many influential logics in this area are interpreted over concurrent game frames, including Coalition Logic $\mathsf{CL}$ \cite{pauly_modal_2002} and Alternating-time Temporal Logic $\mathsf{ATL}$ \cite{alur_alternating-time_2002}. Roughly speaking, a concurrent game frame consists of states and agents (forming coalitions), where, at each state, every coalition has available joint actions, and each such action induces a set of possible outcome states.

Two notions of coalition power arise from concurrent game frames: \emph{alpha powers} and \emph{actual powers}. An \emph{alpha power} of a coalition is a set of possible futures such that the coalition has an action that forces every resulting outcome to lie in that set. An \emph{actual power} of a coalition is a set of possible futures such that the coalition has an action satisfying both of the following conditions: (1) the action forces every resulting outcome to lie in the set, and (2) every future in the set is realizable by the complement coalition (and the environment). Thus, unlike alpha powers, actual powers explicitly involve agents outside the coalition.

Correspondingly, one obtains two kinds of neighborhood frames: \emph{alpha neighborhood frames} and \emph{actual neighborhood frames}. In an alpha neighborhood frame, each coalition is assigned an \emph{alpha neighborhood function} that, at every state, specifies an \emph{alpha neighborhood}, i.e., a set of alpha powers. In an actual neighborhood frame, each coalition is assigned an \emph{actual neighborhood function} that, at every state, specifies an \emph{actual neighborhood}, i.e., a set of actual powers.

\emph{Alpha powers} are widely used in game theory; see, e.g.,~\cite{moulin_cores_1982}. 
By contrast, there appears to be comparatively little work on \emph{actual} powers. 
To the best of our knowledge, the main reference is van Benthem, Bezhanishvili, and Enqvist~\cite{benthem_new_2019}. 
Working primarily from the perspective of game equivalence, they investigate \emph{basic} powers of singleton coalitions in turn-based two-agent extensive games with imperfect information, defined in terms of (uniform) strategies. 
Conceptually, this notion of basic power coincides with what we call \emph{actual} power.

%%%%%%%%%%%%%%%
%%%%%%%%%%%%%%%
\subsection{Alpha powers versus actual powers}
\label{subsec:01-02-alpha-powers-versus-actual-powers}

Alpha and actual powers abstract from the underlying actions in different ways. Alpha powers record guarantees: if a coalition has an action whose possible outcomes are all contained in a set \(X\), then the same action also guarantees every superset of \(X\). Hence alpha powers are upward closed. Actual powers, by contrast, record the precise outcome set generated by an available action of the coalition, that is, the set of states left possible once that action is fixed and the remaining agents' choices are left open. They are therefore generally non-monotonic.

Actual powers provide a more precise account of coalition power in concurrent game frames than alpha powers.
Consider the following example.

%%%%%%%%%%%%%%%
%%%%%%%%%%%%%%%
\begin{example}[Same alpha powers, different actual powers]
\label{ex:01-same-alpha-powers-different-actual-powers}

Consider two AI-driven robots, $a$ and $b$, controlling a warehouse gate.
The two robots send their inputs to the gate controller simultaneously, and
robot $a$ has higher authority than robot $b$.

There are two relevant states, \(W=\{w_1,w_2\},\) where $w_1$ is the state in which the gate is closed, and $w_2$ is the state in
which the gate is open.
Robot $a$ has three actions: \(\mathtt{open}_a, \mathtt{close}_a, \mathtt{delegate}_a.\)
The first action opens the gate, the second closes the gate, and the third
delegates the decision to robot $b$. Robot $b$ has two actions: \(\mathtt{open}_b, \mathtt{close}_b.\)

We consider two scenarios. In both scenarios, if $a$ chooses
$\mathtt{open}_a$, the gate opens, and if $a$ chooses
$\mathtt{close}_a$, the gate closes, independently of $b$'s action.
The scenarios differ only in the effect of $\mathtt{delegate}_a$. In the first
scenario, delegation works as intended: if $a$ chooses
$\mathtt{delegate}_a$, then $b$'s action determines whether the gate opens or
closes. In the second scenario, delegation is ineffective: if $a$ chooses
$\mathtt{delegate}_a$, then the gate closes regardless of $b$'s action.

The two scenarios induce the same alpha powers for every coalition:
\[
\begin{array}{c|c}
\text{coalition} & \text{alpha powers} \\
\hline
\emptyset & \{W\} \\
\{a\} & \bigl\{\{w_1\},\{w_2\},W\bigr\} \\
\{b\} & \{W\} \\
\{a,b\} & \bigl\{\{w_1\},\{w_2\},W\bigr\}.
\end{array}
\]

However, the actual powers of $\{a\}$ differ. In the first scenario, $\{a\}$
has the following actual powers: \(\{w_1\}, \{w_2\}, W,\)
because $a$ can close, open, or delegate. In the second scenario, $\{a\}$ has
only the following actual powers: \(\{w_1\}, \{w_2\}.\)
Indeed, in the second scenario, choosing $\mathtt{delegate}_a$ no longer leaves
both outcomes open; it has the same exact outcome set as choosing
$\mathtt{close}_a$.

Thus the two scenarios are indistinguishable by alpha powers but are distinguished
by actual powers.

\end{example}

As noted above, actual powers explicitly keep track of the role of agents
outside the coalition. This makes them useful for modeling social scenarios
in which a coalition's ability is not exhausted by what it can guarantee, but
also depends on what its action leaves open to others. A central logical
example is Socially Friendly Coalition Logic~\cite{goranko_socially_2018}.
Its operator $[\FCC](\phi;\psi_1,\dots,\psi_k)$ says that coalition $\FCC$ has
a collective action $\sigma_\FCC$ that guarantees $\phi$, while still allowing
the complementary coalition $\overline{\FCC}$ to realize each of the
alternatives $\psi_1,\dots,\psi_k$ by a suitable response. Semantically, this
amounts to requiring an actual power \(X\) of $\FCC$ such that every state in
\(X\) satisfies \(\phi\), and, for each \(i \leq k\), some state in \(X\)
satisfies \(\psi_i\). The second requirement is precisely what cannot be captured by alpha powers
alone.

%%%%%%%%%%%%%%%
%%%%%%%%%%%%%%%
\subsection{Alpha and actual representation}
\label{subsec:01-03-alpha-and-actual-representation}

Every concurrent game frame induces an alpha neighborhood frame.
Pauly~\cite{pauly_modal_2002} and Goranko, Jamroga, and Turrini~\cite{goranko_strategic_2013} showed that the class of concurrent game frames is \emph{representable} by a class of alpha neighborhood frames characterized by certain \emph{good} properties of neighborhood functions, in the following sense:
(1) every alpha neighborhood frame induced by a concurrent game frame satisfies these properties, and
(2) every alpha neighborhood frame satisfying these properties is induced by some concurrent game frame.
This result is the \emph{alpha representation theorem} for concurrent game frames.

Every concurrent game frame also induces an actual neighborhood frame. Do there exist \emph{good} properties of neighborhood functions such that
(1) every actual neighborhood frame induced by a concurrent game frame satisfies these properties, and
(2) every actual neighborhood frame satisfying these properties is induced by some concurrent game frame?
If so, this would be an \emph{actual representation theorem} for concurrent game frames.

This question has not yet been fully settled in the literature. Among published results, the closest one is due to van Benthem, Bezhanishvili, and Enqvist~\cite{benthem_new_2019}. There, the authors establish a representation theorem for \emph{basic} (strategy-based) powers in turn-based two-agent extensive games with imperfect information, under the restriction to singleton coalitions. Via a standard power-invariance transformation between the class of two-agent concurrent game frames and an appropriate class of turn-based two-agent extensive games with imperfect information, their proof readily yields a representation theorem for actual powers in two-agent concurrent game frames, again restricted to singleton coalitions.
A recent preprint~\cite{chen_representation_2026} gives a systematic
two-agent treatment of this question, together with the corresponding alpha
representation question, in the broader setting of general concurrent game
frames, which will be discussed below.

%%%%%%%%%%%%%%%
%%%%%%%%%%%%%%%
\subsection{Significance of the representation theorems}
\label{subsec:01-04-significance-of-the-representation-theorems}

What are the powers of players, and coalitions of players, in games? Which sets of outcomes -- known as \emph{events} in probability theory or \emph{propositions} in logic -- is it in their power to make come about, in the sense that they can choose a (joint) action such that no matter what the other agents do, this is what will happen? This abstracts away the particulars of a game, and leaves an abstract representation of power which is exactly what we are often interested in when we analyse games. In order to understand power in games, the central question then is: what are the \emph{properties} (or \emph{axioms} or \emph{laws}) of individual and coalitional power induced from games in this way? That is exactly what Pauly's representation theorem \cite{pauly_modal_2002} answers, for alpha powers.
That is the main motivation, and the main significance of the representation theorems for actual power: giving us a deeper understanding of the power in games.

While characterizing the laws of power in games more precisely is in itself a
main motivation behind the representation theorems, they are also, as already
mentioned, highly relevant for the semantics of certain modal logics.
They help provide alternative neighborhood semantics. They also establish the
adequacy of a power-based abstraction of action frames. A concurrent game frame
explicitly records action names, availability functions, and outcome functions.
Yet many modal languages for strategic reasoning are invariant under changes to
these action-level components, provided that the induced powers remain
unchanged. Such languages refer only to the powers generated by actions: what
coalitions can force, in the case of alpha powers, or the exact outcome ranges
associated with coalition actions, in the case of actual powers. A
representation theorem identifies precisely which abstract neighborhood frames
can arise from action-based game frames. Thus, for languages whose truth
conditions depend only on the relevant notion of power, it justifies replacing
action-based models with neighborhood-based models without loss of semantic
information.

For alpha powers, this is the familiar role of Pauly's representation theorem.
When only enforceable sets matter, as in $\FCL$, semantics over concurrent game
frames can be transferred to semantics over alpha neighborhood frames. The
latter describe coalitional abilities directly, without retaining the
additional action-level structure of the underlying game. They therefore
provide a more economical setting, and often a more convenient one, for
metatheoretic arguments. For instance, Pauly~\cite{pauly_modal_2002} proved
completeness of $\FCL$ via alpha neighborhood frames, and {\AA}gotnes and
Alechina~\cite{agotnes_coalition_2019} introduced knowledge into $\FCL$ within
the same semantic framework.

Actual representation plays an analogous, though less developed, role in the
logical setting. While Coalition Logic yields the same logic
whether we interpret its language in neighborhood models induced from games
using alpha powers or in models induced using actual powers\footnote{Here we
assume that the modality $[\FCC]\phi$ is interpreted as the existence of a choice
$X$ of states available to $\FCC$ at $s$ such that $\phi$ is true at every state in
$X$. The semantics of these operators is usually given by requiring the
extension of $\phi$ (the set of states where $\phi$ is true) to be an available
choice for $\FCC$ at $s$, which is equivalent given monotonicity. See
\cite{pacuit2017neighborhood} for a discussion of the two definitions.}, the
same is not true for languages whose modalities are sensitive not only to what
a coalition can guarantee, but also to the residual range of outcomes left open
by a coalition action. A representation theorem then justifies
treating these outcome sets themselves as semantic primitives. This may provide
a compact basis for completeness proofs, canonical-model constructions,
definability arguments, and comparisons between models.

As noted above, Socially Friendly Coalition
Logic~\cite{goranko_socially_2018} is naturally understood as a logic of
coalitions' actual powers in concurrent game frames. The completeness of its
two-agent case has been established~\cite{goranko_complete_2026}, whereas the
general multi-agent case remains open. An actual representation theorem
therefore offers a natural route from the action-based semantics of such
operators to actual-neighborhood semantics. It does not by itself solve the
completeness problem, but it provides the semantic infrastructure needed to
study such questions without explicitly retaining action names and outcome
functions.

%%%%%%%%%%%%%%%
%%%%%%%%%%%%%%%
\subsection{Eight classes of general concurrent game frames}
\label{subsec:01-05-eight-classes-of-general-concurrent-game-frames}

Li and Ju~\cite{li_minimal_2025, li_completeness_2026} argued that standard concurrent game frames rely on three assumptions that can be too strong in practice: \emph{seriality}, \emph{independence of agents}, and \emph{determinism}.
We briefly recall their main motivations.

Games often terminate upon reaching designated states. For example, a rock--paper--scissors game ends once a winner is determined. Intuitively, at such terminal states, agents may have no available actions; hence seriality may fail.

There are also situations in which whether a coalition can perform an action depends on what other agents do at the same time. For example, suppose two agents, $a$ and $b$, are in a room with a single chair. Agent $a$ can sit, and agent $b$ can sit, but they cannot both sit simultaneously. In such cases, independence of agents fails.

Moreover, in many scenarios, a joint action of all participating agents may lead to more than one possible outcome state. The following example (from \cite{sergot_some_2014}) illustrates this point. A vase stands on a table, and an agent $a$ can raise or lower one end of the table. If the table tilts, the vase may fall; and if it falls, it may break. Thus, determinism may fail.

Motivated by such examples, Li and Ju~\cite{li_minimal_2025, li_completeness_2026} introduced eight classes of general concurrent game frames, determined by which of the three properties are assumed, and studied the corresponding eight generalized coalition logics.

%%%%%%%%%%%%%%%
%%%%%%%%%%%%%%%
\subsection{Our work}
\label{subsec:01-06-our-work}

As mentioned above, the recent preprint~\cite{chen_representation_2026}
gives a systematic two-agent picture. For each of the eight classes of
two-agent general concurrent game frames determined by seriality, independence
of agents, and determinism, it establishes two kinds of representation results. First, it defines a finite set
of natural properties of neighborhood frames and shows that the class of
neighborhood frames determined by these properties represents the class of
general concurrent game frames for actual powers. Second, it defines several
natural properties of neighborhood frames and shows that the class of
neighborhood frames determined by these properties represents the class of
general concurrent game frames for alpha powers.

For actual powers, however, that preprint also shows that the direct
finite-agent reading of the two-agent conditions is no longer sufficient once
there are at least three agents. The
reason is that they only regulate the behaviour of actual powers along
inclusions of coalitions. They do not record whether powers of overlapping coalitions can be generated as
restrictions of one coherent system of full-action witnesses.
Once at least three agents are present, actual powers assigned to overlapping
coalitions must also be globally compatible.

In this paper, we study how the two-agent representation picture extends to
arbitrary finite nonempty sets of agents.

We discuss properties of neighborhood frames at three levels: the index level,
the power-profile level, and the power level. The index-level conditions use
indices, which intuitively indicate abstract action profiles. The
power-profile-level conditions use power profiles, which consist of one actual
power for each coalition. The power-level conditions use actual powers. In
general, the index-level conditions are too abstract, whereas the power-level conditions are most transparent.

Our contributions are as follows.

First, for each of the eight classes of finite-agent general concurrent game
frames determined by seriality, independence of agents, and determinism, at the
index level, we give a finite set of properties of neighborhood frames and show
that the class of neighborhood frames determined by these properties represents
the class of general concurrent game frames for actual powers.

Second, for each of the four classes of general concurrent game frames
determined by seriality and determinism, at the power-profile level, we give a
finite set of properties of neighborhood frames. We show that these properties
are equivalent to the corresponding properties involving indices. It follows
that the four classes of neighborhood frames determined by these
power-profile-level properties represent the corresponding classes of general
concurrent game frames for actual powers.

Third, for alpha powers, we show that, for each of the four classes of general concurrent game
frames determined by seriality and determinism, the two-agent properties at the
level of powers generalize to arbitrary finite nonempty sets of agents.

Fourth, we analyze the obstruction caused by independence.

The remainder of the paper is organized as follows.

\begin{itemize}

\item Section~\ref{sec:02-action-frames-powers-neighborhood-frames-and-representability} fixes the basic framework. It recalls
action frames, alpha and actual powers, alpha and actual neighborhood frames,
and the corresponding notions of representability.

\item Section~\ref{sec:03-eight-classes-of-general-concurrent-game-frames}
reviews general concurrent game frames and the eight classes obtained by
imposing or omitting seriality, independence of agents, and determinism.

\item Section~\ref{sec:04-background-representation-theorems} recalls the relevant
two-agent representation results, the finite-agent alpha havingness direction,
and the three-agent counterexample showing that the direct finite-agent
analogue of the two-agent actual-power conditions is not sufficient.

\item Section~\ref{sec:05-representing-actual-powers-via-abstract-presentations}
isolates the abstract witness structure behind actual representation. It
defines abstract actual presentations and proves the abstract havingness and
enoughness theorems.

\item
Section~\ref{sec:06-coherent-profile-covers}
develops the finite-agent machinery at the level of power profiles: profile
agreement, the support operator, coherent profile covers, coherent
coverability, and the coherent core, and proves the translations between
coherent profile covers and abstract actual presentations.

\item
Section~\ref{sec:07-representing-actual-powers-without-assuming-independence-via-coherent-profile-covers}
formulates finite-agent \(\PAC\)-representativeness directly in terms of
coherent profile covers and proves the corresponding actual representation
theorems for the four classes of general concurrent game frames that impose no
independence condition.

\item
Section~\ref{sec:08-alpha-representation-without-assuming-independence-via-actual-representation}
establishes the corresponding alpha representation theorems for the same four classes of general concurrent game frames, by extracting a suitable actual basis from alpha neighborhoods.

\item Section~\ref{sec:09-amalgamation-of-power-profiles-is-insufficient-for-independence} discusses independence and
amalgamation of power profiles. It shows how abstract amalgamation induces a
natural profile-level condition, and explains why this condition is still not sufficient for independent actual representation in the finite-agent
setting.

\item Section~\ref{sec:10-concluding-remarks} concludes with final remarks and
directions for future work.

\end{itemize}

\emph{Sections~\ref{sec:02-action-frames-powers-neighborhood-frames-and-representability},
\ref{sec:03-eight-classes-of-general-concurrent-game-frames}, and
\ref{sec:04-background-representation-theorems} are adapted
from~\cite{chen_representation_2026}. We include this material to make the
paper self-contained.}

%%%%%%%%%%%%%%%
%%%%%%%%%%%%%%%
\section{Action frames, powers, neighborhood frames, and representability}
\label{sec:02-action-frames-powers-neighborhood-frames-and-representability}

In this section, we first introduce action frames in an abstract form; the additional
constraints defining general concurrent game frames will be imposed in
Section~\ref{sec:03-eight-classes-of-general-concurrent-game-frames}. We then
define alpha and actual powers, alpha and actual neighborhood frames, and the
corresponding notions of representability.

%%%%%%%%%%%%%%%
%%%%%%%%%%%%%%%
\subsection{Action frames}
\label{subsec:02-01-action-frames}

Let $\FAG$ be a finite nonempty set of agents. Subsets of $\FAG$ are called \Fdefs{coalitions}, and $\FAG$ itself is called the \Fdefs{grand coalition}.
In what follows, whenever no confusion arises, we write $a$ instead of $\{a\}$ for $a \in \FAG$.

Let $\FAC$ be a nonempty set of actions. For each $\FCC \subseteq \FAG$, define
\(\FJA_\FCC = \{\sigma_\FCC \mid \sigma_\FCC: \FCC \rightarrow \FAC\},\)
the set of \Fdefs{joint actions} of $\FCC$. Note $\FJA_\emptyset =\{\emptyset\}$. For $\FCC, \FDD \subseteq \FAG$ such that $\FCC \subseteq \FDD$, and for every $\sigma_\FDD \in \FJA_\FDD$, we denote by $\sigma_\FDD|_\FCC$ the \Fdefs{restriction} of $\sigma_\FDD$ to $\FCC$, which is an element of $\FJA_\FCC$.
If $\sigma_\FCC\in\FJA_\FCC$ and $\sigma_\FDD\in\FJA_\FDD$ agree on
$\FCC\cap\FDD$, we write $\sigma_\FCC\cup\sigma_\FDD$ for their union, which
is an element of $\FJA_{\FCC\cup\FDD}$. In particular, this notation is always
defined when $\FCC$ and $\FDD$ are disjoint.
In the sequel, we sometimes represent joint actions of coalitions as action sequences, implicitly assuming a fixed order of agents.

%%%%%%%%%%%%%%%
%%%%%%%%%%%%%%%
\begin{definition}[Action frames]
\label{def:01-action-frames}

An \Fdefs{action frame} is a tuple
\[
\FAF = (\FST, \FAC, \{\Fav_\FCC \mid \FCC \subseteq \FAG\}, \{\Fout_\FCC \mid \FCC \subseteq \FAG\}),
\]
where:

\begin{itemize}

%%%
\item
$\FST$ is a nonempty set of states.

%%%%
\item
$\FAC$ is a nonempty set of actions.

%%%%
\item
For each $\FCC \subseteq \FAG$, $\Fav_\FCC: \FST \rightarrow \mathcal{P}(\FJA_\FCC)$ is an \Fdefs{availability function} for $\FCC$.

\emph{Here, $\Fav_\FCC(s)$ is the set of all joint actions of coalition $\FCC$ that are available at state $s$.}

%%%%
\item
For each $\FCC \subseteq \FAG$, $\Fout_\FCC: \FST \times \FJA_\FCC \rightarrow \mathcal{P}(\FST)$ is an \Fdefs{outcome function} for $\FCC$.

\emph{Here, $\Fout_\FCC(s, \sigma_\FCC)$ is the set of possible outcome states when coalition $\FCC$ performs joint action $\sigma_\FCC$ at state $s$.}

\end{itemize}

\end{definition}

Note that the action set $\FAC$ may vary across action frames.

Different constraints may be imposed on action frames.

Intuitively, $\Fout_\emptyset(s,\emptyset)$ contains all states to which $s$ may evolve, i.e., all possible successors of $s$.

%%%%%%%%%%%%%%%
%%%%%%%%%%%%%%%
\subsection{Alpha and actual powers}
\label{subsec:02-02-alpha-and-actual-powers}

Fix an action frame
\[
\FAF = (\FST, \FAC, \{\Fav_\FCC \mid \FCC \subseteq \FAG\}, \{\Fout_\FCC \mid \FCC \subseteq \FAG\}),
\]
a state $s \in \FST$, a coalition $\FCC \subseteq \FAG$, and an available joint action $\sigma_\FCC \in \Fav_\FCC(s)$.
We say that a set of states $X \subseteq \FST$ is \Fdefs{safe} for $\sigma_\FCC$ at $s$ if \(\Fout_\FCC(s,\sigma_\FCC)\subseteq X\), that is, performing $\sigma_\FCC$ at $s$ guarantees that the next state lies in $X$.
We say that $X$ is \Fdefs{tight} for $\sigma_\FCC$ at $s$ if \(X\subseteq \Fout_\FCC(s,\sigma_\FCC)\), that is, every state in $X$ can occur as an outcome compatible with performing $\sigma_\FCC$ at $s$.

An \Fdefs{alpha power} of $\FCC$ at $s$ is a set of states $X \subseteq \FST$ that is safe for some available joint action of $\FCC$ at $s$.
An \Fdefs{actual power} of $\FCC$ at $s$ is a set of states $X \subseteq \FST$ that is both safe and tight for some available joint action of $\FCC$ at $s$. Equivalently, actual powers of $\FCC$ at $s$ are exactly the sets of the form $\Fout_\FCC(s,\sigma_\FCC)$ for $\sigma_\FCC \in \Fav_\FCC(s)$.
We say that an (alpha or actual) power $X$ \Fdefs{enables} a state $t \in \FST$ at $s$ if $t \in X$.

Alpha powers are monotonic: if $X$ is an alpha power, then every superset $X' \supseteq X$ is also an alpha power.
By contrast, actual powers are generally non-monotonic, since tightness rules out adding states that are not actually reachable.

%%%%%%%%%%%%%%%
%%%%%%%%%%%%%%%
\subsection{Alpha and actual neighborhood frames}
\label{subsec:02-03-alpha-and-actual-neighborhood-frames}

%%%%%%%%%%%%%%%
%%%%%%%%%%%%%%%
\begin{definition}[Alpha neighborhood frames]
\label{def:02-alpha-neighborhood-frames}

An \Fdefs{alpha neighborhood frame} is a tuple
\(\FALNF = (\FST, \{\Falnei_\FCC \mid \FCC \subseteq \FAG\}),\)
where:

\begin{itemize}

\item
$\FST$ is a nonempty set of states;

\item
For each $\FCC \subseteq \FAG$, $\Falnei_\FCC: \FST \to \mathcal{P}(\mathcal{P}(\FST))$ is an \Fdefs{alpha neighborhood function} of $\FCC$ such that, for every $s \in \FST$, the set $\Falnei_\FCC(s)$ is closed under supersets.

\end{itemize}

\end{definition}

Different constraints may be imposed on alpha neighborhood frames.

The set $\Falnei_\FCC(s)$ is called the \Fdefs{alpha neighborhood} of $\FCC$ at $s$. Intuitively, $\Falnei_\FCC(s)$ consists of the alpha powers of $\FCC$ at $s$.

%%%%%%%%%%%%%%%
%%%%%%%%%%%%%%%
\begin{definition}[Nonmonotonic cores]
\label{def:03-nonmonotonic-cores}

Let $\FALNF = (\FST, \{\Falnei_\FCC \mid \FCC \subseteq \FAG\})$ be an alpha neighborhood frame.
For each $\FCC \subseteq \FAG$, define $\Falneinc_\FCC$, called the \Fdefs{nonmonotonic core} of $\Falnei_\FCC$, by setting, for every $s \in \FST$,
\[
\Falneinc_\FCC(s)
=
\{\, X \in \Falnei_\FCC(s) \mid \text{there is no } Y \in \Falnei_\FCC(s) \text{ with } Y \subset X \,\}.
\]
That is, $\Falneinc_\FCC(s)$ is the set of $\subseteq$-minimal elements of $\Falnei_\FCC(s)$.
For each $s \in \FST$, the set $\Falneinc_\FCC(s)$ is also called the \Fdefs{nonmonotonic core} of $\Falnei_\FCC(s)$.

\end{definition}

Intuitively, the elements of $\Falneinc_\FCC(s)$ can be viewed as candidate actual powers of $\FCC$ at $s$.
However, one should not expect $\Falneinc_\FCC(s)$ to contain all actual powers of $\FCC$ at $s$.
Indeed, it may happen that one actual power properly contains another, whereas such proper containments cannot occur among elements of $\Falneinc_\FCC(s)$.

Note that $\Falneinc_\FCC(s)$ may be empty even when $\Falnei_\FCC(s)$ is nonempty. However, this cannot occur when $\FST$ is finite.

Later, $\bigcup \Falneinc_\emptyset(s)$ will be understood as the set of all states to which $s$ may evolve, i.e., the possible successors of $s$.

%%%%%%%%%%%%%%%
%%%%%%%%%%%%%%%
\begin{definition}[Actual neighborhood frames]
\label{def:04-actual-neighborhood-frames}

An \Fdefs{actual neighborhood frame} is a tuple
\(\FACNF = (\FST, \{\Facnei_\FCC \mid \FCC \subseteq \FAG\}),\)
where:

\begin{itemize}

\item
$\FST$ is a nonempty set of states;

\item
For each $\FCC \subseteq \FAG$, $\Facnei_\FCC: \FST \to \mathcal{P}(\mathcal{P}(\FST))$ is an \Fdefs{actual neighborhood function} of $\FCC$.

\end{itemize}

\end{definition}

Different constraints may be imposed on actual neighborhood frames. Once additional constraints are imposed on actual neighborhood frames, alpha neighborhood frames
may no longer be actual neighborhood frames.

The set $\Facnei_\FCC(s)$ is called the \Fdefs{actual neighborhood} of $\FCC$ at $s$. Intuitively, $\Facnei_\FCC(s)$ consists of the actual powers of coalition $\FCC$ at state $s$.

In the intended interpretation, $\bigcup \Facnei_\emptyset(s)$ consists of all states to which $s$ may evolve, i.e., the possible successors of $s$.

%%%%%%%%%%%%%%%
%%%%%%%%%%%%%%%
\subsection{Representation of action frames by alpha and actual neighborhood frames}
\label{subsec:02-04-representation-of-action-frames-by-alpha-and-actual-neighborhood-frames}

%%%%%%%%%%%%%%%
%%%%%%%%%%%%%%%
\begin{definition}[Alpha and actual effectivity functions of action frames]
\label{def:05-alpha-and-actual-effectivity-functions-of-action-frames}

Let $\FAF = (\FST, \FAC, \{\Fav_\FCC \mid \FCC \subseteq \FAG\}, \{\Fout_\FCC \mid \FCC \subseteq \FAG\})$ be an action frame.
For each $\FCC \subseteq \FAG$, define:

\begin{itemize}

\item
the \Fdefs{alpha effectivity function} $\FALEF_\FCC$ for $\FCC$ in $\FAF$ by setting, for every $s \in \FST$,
\[
\FALEF_\FCC(s)
=
\{\, Y \subseteq \FST \mid \Fout_\FCC(s,\sigma_\FCC) \subseteq Y \text{ for some } \sigma_\FCC \in \Fav_\FCC(s) \,\};
\]

\item
the \Fdefs{actual effectivity function} $\FACEF_\FCC$ for $\FCC$ in $\FAF$ by setting, for every $s \in \FST$,
\[
\FACEF_\FCC(s)
=
\{\, \Fout_\FCC(s,\sigma_\FCC) \mid \sigma_\FCC \in \Fav_\FCC(s) \,\}.
\]

\end{itemize}

\end{definition}

From the definition, it is immediate
that each $\FALEF_\FCC(s)$ is closed under supersets.

%%%%%%%%%%%%%%%
%%%%%%%%%%%%%%%
\begin{definition}[$\PAL$-representability and $\PAC$-representability of action frames]
\label{def:06-pal-representability-and-pac-representability-of-action-frames}

Let $\FAF = (\FST, \FAC, \ab \{\Fav_\FCC \mid \FCC \subseteq \FAG\}, \{\Fout_\FCC \mid \FCC \subseteq \FAG\})$ be an action frame.
We say that $\FAF$ is \Fdefs{$\PAL$-representable} by an alpha neighborhood frame \(\FALNF = (\FST, \{\Falnei_\FCC \mid \FCC \subseteq \FAG\})\)
if, for every $\FCC \subseteq \FAG$, we have $\FALEF_\FCC = \Falnei_\FCC$.

We say that $\FAF$ is \Fdefs{$\PAC$-representable} by an actual neighborhood frame \(\FACNF = (\FST, \{\Facnei_\FCC \mid \FCC \subseteq \FAG\})\)
if, for every $\FCC \subseteq \FAG$, we have $\FACEF_\FCC = \Facnei_\FCC$.

\end{definition}

Intuitively, $\PAL$-representability (resp.\ $\PAC$-representability) means that the corresponding alpha (resp.\ actual) neighborhood frame captures all information about the alpha (resp.\ actual) powers of coalitions in the action frame.

%%%%%%%%%%%%%%%
%%%%%%%%%%%%%%%
\begin{definition}[$\PAL$-representability and $\PAC$-representability of classes of action frames]
\label{def:07-pal-representability-and-pac-representability-of-classes-of-action-frames}

Let $\mathbf{AF}$ be a class of action frames.
We say that $\mathbf{AF}$ is \Fdefs{$\PAL$-representable} (resp.\ \Fdefs{$\PAC$-representable}) by a class of alpha (resp.\ actual) neighborhood frames $\alpha\text{-}\mathbf{NF}$ (resp.\ $\mathbf{AC}\text{-}\mathbf{NF}$) if:

\begin{itemize}

\item
Every action frame in $\mathbf{AF}$ is $\PAL$-representable (resp.\ $\PAC$-representable) by some alpha (resp.\ actual) neighborhood frame in $\alpha\text{-}\mathbf{NF}$ (resp.\ $\mathbf{AC}\text{-}\mathbf{NF}$);

\item
Every alpha (resp.\ actual) neighborhood frame in $\alpha\text{-}\mathbf{NF}$ (resp.\ $\mathbf{AC}\text{-}\mathbf{NF}$) $\PAL$-represents (resp.\ $\PAC$-represents) some action frame in $\mathbf{AF}$.

\end{itemize}

\end{definition}

Suppose that a class of action frames is $\PAL$-representable (resp.\ $\PAC$-representable) by a class of alpha (resp.\ actual) neighborhood frames. Intuitively, this means that, as far as alpha (resp.\ actual) powers are concerned, one can translate the information carried by action frames into the corresponding neighborhood frames.

Now fix a class of action frames $\mathbf{AF}$ and a collection of properties of alpha (resp.\ actual) neighborhood frames.
To show that $\mathbf{AF}$ is $\PAL$-representable (resp.\ $\PAC$-representable) by the class of alpha (resp.\ actual) neighborhood frames satisfying these properties, it suffices to establish the following two statements:

\begin{itemize}

\item
\Fdefs{Havingness} (necessity): every alpha (resp.\ actual) neighborhood frame that $\PAL$-represents (resp.\ $\PAC$-represents) an action frame in $\mathbf{AF}$ satisfies these properties;

\item
\Fdefs{Enoughness} (sufficiency): every alpha (resp.\ actual) neighborhood frame satisfying these properties $\PAL$-represents (resp.\ $\PAC$-represents) an action frame in $\mathbf{AF}$.

\end{itemize}

Later we will follow exactly this pattern.

%%%%%%%%%%%%%%%
%%%%%%%%%%%%%%%
\subsubsection*{On nontrivial representation}

Fix a class of action frames $\mathbf{AF}$. From each $\FAF \in \mathbf{AF}$, we obtain an alpha (resp.\ actual) neighborhood frame $\FALNF$ (resp.\ $\FACNF$) that $\PAL$-represents (resp.\ $\PAC$-represents) $\FAF$.
Let $\alpha\text{-}\mathbf{NF}$ (resp.\ $\mathbf{AC}\text{-}\mathbf{NF}$) be the class of all such alpha (resp.\ actual) neighborhood frames.
Then, trivially, $\mathbf{AF}$ is $\PAL$-representable (resp.\ $\PAC$-representable) by $\alpha\text{-}\mathbf{NF}$ (resp.\ $\mathbf{AC}\text{-}\mathbf{NF}$).
This observation is not very informative: to obtain a meaningful representation, the classes $\alpha\text{-}\mathbf{NF}$ and $\mathbf{AC}\text{-}\mathbf{NF}$ should be specified by some \emph{good} properties of alpha and actual neighborhood frames rather than by construction from $\mathbf{AF}$. Of course, \emph{good} is a vague notion.

%%%%%%%%%%%%%%%
%%%%%%%%%%%%%%%
\section{Eight classes of general concurrent game frames}
\label{sec:03-eight-classes-of-general-concurrent-game-frames}

In this section, we recall the eight classes of general concurrent game frames
introduced by Li and Ju~\cite{li_minimal_2025, li_completeness_2026}. The
definitions are stated for arbitrary finite nonempty sets of agents. The
two-agent representation theorem is recalled in the next section, and the
finite-agent results without the independence condition are developed later.

%%%%%%%%%%%%%%%
%%%%%%%%%%%%%%%
\begin{definition}[General concurrent game frames]
\label{def:08-general-concurrent-game-frames}

An action frame
\[
\FGCGF = (\FST, \FAC, \{\Fav_\FCC \mid \FCC \subseteq \FAG\}, \{\Fout_\FCC \mid \FCC \subseteq \FAG\})
\]
is a \Fdefs{general concurrent game frame} if the following conditions hold:

\begin{itemize}

\item

The \Fdefs{grand-coalition-induced outcome condition} (simply the \Fdefs{GCI-condition}):
for every $\FCC \subseteq \FAG$, every $s \in \FST$, and every $\sigma_\FCC \in \FJA_\FCC$,
\[
\Fout_\FCC(s,\sigma_\FCC)
=
\bigcup
\{\, \Fout_\FAG(s,\sigma_\FAG) \mid \sigma_\FAG \in \FJA_\FAG \text{ and } \sigma_\FAG|_\FCC = \sigma_\FCC \,\}.
\]

\item

The \Fdefs{outcome-driven availability condition} (simply the \Fdefs{ODA-condition}):
for every $\FCC \subseteq \FAG$ and every $s \in \FST$,
\[
\Fav_\FCC(s)
=
\{\, \sigma_\FCC \in \FJA_\FCC \mid \Fout_\FCC(s,\sigma_\FCC) \neq \emptyset \,\}.
\]

\end{itemize}

\end{definition}

Thus, under the GCI-condition and the ODA-condition, the functions $\Fout_\FCC$ and $\Fav_\FCC$ (for all $\FCC \subseteq \FAG$) are determined by the grand-coalition outcome function $\Fout_\FAG$.

%%%%%%%%%%%%%%%
%%%%%%%%%%%%%%%
\begin{definition}[Seriality, independence, and determinism]
\label{def:09-seriality-independence-and-determinism}

Let $\FGCGF = (\FST, \FAC, \{\Fav_\FCC \mid \FCC \subseteq \FAG\}, \{\Fout_\FCC \mid \FCC \subseteq \FAG\})$ be a general concurrent game frame. We say that:

\begin{itemize}

\item
$\FGCGF$ is \Fdefs{serial} if, for all $\FCC \subseteq \FAG$ and all $s \in \FST$, we have $\Fav_\FCC(s) \neq \emptyset$;

\item
$\FGCGF$ is \Fdefs{independent} if, for all $s \in \FST$ and all disjoint coalitions $\FCC,\FDD \subseteq \FAG$,
whenever $\sigma_\FCC \in \Fav_\FCC(s)$ and $\sigma_\FDD \in \Fav_\FDD(s)$, we have $\sigma_\FCC \cup \sigma_\FDD \in \Fav_{\FCC \cup \FDD}(s)$;

\item
$\FGCGF$ is \Fdefs{deterministic} if, for all $s \in \FST$ and all $\sigma_\FAG \in \Fav_\FAG(s)$, the set $\Fout_\FAG(s,\sigma_\FAG)$ is a singleton.

\end{itemize}

\end{definition}

Let $\mathtt{S}$, $\mathtt{I}$, and $\mathtt{D}$ denote seriality, independence, and determinism, respectively. Let
\[
\FES = \{\epsilon,\ \mathtt{S},\ \mathtt{I},\ \mathtt{D},\ \mathtt{SI},\ \mathtt{SD},\ \mathtt{ID},\ \mathtt{SID}\}
\]
be the set of the eight strings encoding all combinations of these three properties. Here $\epsilon$ denotes the empty combination.
For $\FXX \in \FES$, we say that a general concurrent game frame is an \Fdefs{$\FXX$-frame} if it satisfies the properties indicated by $\FXX$.

The following facts will be used implicitly throughout the sequel.

%%%%%%%%%%%%%%%
%%%%%%%%%%%%%%%
\begin{fact}[Basic consequences of general concurrent game frames]
\label{fact:01-basic-consequences-of-general-concurrent-game-frames}

Let
\[
\FGCGF = (\FST, \FAC, \{\Fav_\FCC \mid \FCC \subseteq \FAG\}, \{\Fout_\FCC \mid \FCC \subseteq \FAG\})
\]
be a general concurrent game frame. Then the following hold:

\begin{enumerate}

\item
\Fdefs{Outcome monotonicity:} for all $s \in \FST$, all $\FCC \subseteq \FDD \subseteq \FAG$, all $\sigma_\FCC \in \FJA_\FCC$, and all $\sigma_\FDD \in \FJA_\FDD$ such that $\sigma_\FDD|_\FCC = \sigma_\FCC$, we have
\(\Fout_\FDD(s,\sigma_\FDD)\subseteq \Fout_\FCC(s,\sigma_\FCC).\)

\item

\Fdefs{Alternative GCI-condition:}
for all $s \in \FST$, all $\FCC \subseteq \FAG$, and all $\sigma_\FCC \in \FJA_\FCC$,
\[
\Fout_\FCC(s,\sigma_\FCC)
=
\bigcup
\{\, \Fout_\FAG(s,\sigma_\FAG)\mid \sigma_\FAG \in \Fav_\FAG(s)\ \text{and}\ \sigma_\FAG|_\FCC=\sigma_\FCC \,\}.
\]

\end{enumerate}

\end{fact}

The first item is an immediate consequence of the GCI-condition, and the second is a reformulation of the GCI-condition using the ODA-condition (since unavailable grand-coalition actions have empty outcome sets).

%%%%%%%%%%%%%%%
%%%%%%%%%%%%%%%
\subsubsection*{On the eight frame classes}

The labels in $\FES$ record which frame conditions are imposed; they do not
record which conditions are required to fail. Thus the eight classes are not
intended to be disjoint. Accordingly, when we speak of classes not imposing the
independence condition, we mean the four labels not containing $\mathtt I$,
namely
$
\epsilon, \mathtt S, \mathtt D, \mathtt{SD},
$
not the class of frames in which independence fails.

The standard class of concurrent game frames corresponds to the
$\mathtt{SID}$ case. For this class, the corresponding representation theorem
for what we call alpha powers is Pauly's representation theorem for playable
effectivity functions~\cite{pauly_modal_2002}; see also
\cite{goranko_strategic_2013}. For weakened settings without the independence
condition, Shi and Wang~\cite{shi_representation_2024} give a related alpha
representation theorem for the serial and deterministic case, that is, for the
$\mathtt{SD}$ case.

A recent preprint~\cite{chen_representation_2026} treats the two-agent case
systematically. It proves, for each of the eight two-agent classes of general
concurrent game frames, both an $\PAC$-representation theorem for actual powers
and an $\PAL$-representation theorem for alpha powers.

%%%%%%%%%%%%%%%
%%%%%%%%%%%%%%%
\section{Background representation theorems}
\label{sec:04-background-representation-theorems}

This section recalls the background representation theorems from
\cite{chen_representation_2026}. These results serve as the two-agent
benchmark for the finite-agent representation theorems developed in the present
paper, and we state them without proof. We also recall the three-agent
counterexample from \cite{chen_representation_2026}, which shows that, for
actual powers, the direct finite-agent analogues of the two-agent
neighborhood-frame conditions are necessary but not sufficient.

%%%%%%%%%%%%%%%
%%%%%%%%%%%%%%%
\subsection{Actual powers}
\label{subsec:04-01-actual-powers}

%%%%%%%%%%%%%%%
%%%%%%%%%%%%%%%
\begin{definition}[Two-agent $\PAC$-representative actual neighborhood frames]
\label{def:10-two-agent-pac-representative-actual-neighborhood-frames}

Assume $|\FAG|=2$.

Let \(\FACNF=(\FST,\{\Facnei_\FCC\mid \FCC\subseteq\FAG\})\) be an actual neighborhood frame. We say that $\FACNF$ is
\Fdefs{two-agent $\PAC$-representative} if the following conditions are
satisfied.

\begin{enumerate}

\item \Fdefs{Actual triviality of the empty coalition.}
For all $s\in\FST$, if $\Facnei_\emptyset(s)$ is nonempty, then it is a
singleton.

\emph{Intuitively, the empty coalition has only one joint action, namely the
empty function. Hence, in an induced actual neighborhood frame, it can have at
most one actual power. It may have none in nonserial frames.}

\item \Fdefs{Liveness.}
For all $s\in\FST$ and all $\FCC\subseteq\FAG$, \(\emptyset\notin\Facnei_\FCC(s)\).

\emph{Intuitively, no available coalition action has the impossible exact
outcome $\emptyset$. This condition does not require coalitions to have powers;
it only rules out $\emptyset$ as a power when powers are present.}

\item \Fdefs{Actual power inclusion.}
For all $s\in\FST$ and all $\FCC,\FDD\subseteq\FAG$ with $\FCC\subseteq\FDD$,
for every $X\in\Facnei_\FDD(s)$ there exists $Y\in\Facnei_\FCC(s)$ such that \(X\subseteq Y\).

\emph{Intuitively, if a larger coalition $\FDD$ realizes $X$ by an action
$\sigma_\FDD$, then the restriction of $\sigma_\FDD$ to the smaller coalition
$\FCC$ realizes some outcome set $Y$ with $X\subseteq Y$. Larger coalitions can
narrow down the possible outcomes.}

\item \Fdefs{Actual power decomposition.}
For all $s\in\FST$ and all $\FCC,\FDD\subseteq\FAG$ with $\FCC\subseteq\FDD$,
for every $X\in\Facnei_\FCC(s)$ there exists $\Delta\subseteq\Facnei_\FDD(s)$
such that \(X=\bigcup\Delta\).

\emph{Intuitively, if a smaller coalition $\FCC$ realizes $X$, then the possible
ways of extending its action to the larger coalition $\FDD$ generate
$\FDD$-actual powers whose union is exactly $X$.}

\end{enumerate}

\end{definition}

We use ``$\PAC$-representative'' for actual neighborhood frames satisfying the
intrinsic conditions intended to characterize those that $\PAC$-represent
general concurrent game frames. This should not be confused with
``$\PAC$-representable'', which refers to being induced by a particular action
frame in the sense of Definition~\ref{def:07-pal-representability-and-pac-representability-of-classes-of-action-frames}.

%%%%%%%%%%%%%%%
%%%%%%%%%%%%%%%
\begin{definition}[Seriality, two-agent independence, and determinism of two-agent $\PAC$-representative actual neighborhood frames]
\label{def:11-seriality-two-agent-independence-and-determinism-of-two-agent-pac-representative-actual-neighborhood-frames}

Assume $|\FAG|=2$.

Let \(\FACNF=(\FST,\{\Facnei_\FCC\mid \FCC\subseteq\FAG\})\) be a two-agent $\PAC$-representative actual neighborhood frame. We say that:

\begin{itemize}

\item $\FACNF$ is \Fdefs{$\PAC$-serial} if, for all $s\in\FST$ and all
$\FCC\subseteq\FAG$, \(\Facnei_\FCC(s)\neq\emptyset\).

\item $\FACNF$ is \Fdefs{two-agent $\PAC$-independent} if, for all $s\in\FST$
and all disjoint coalitions $\FCC,\FDD\subseteq\FAG$, whenever
$X\in\Facnei_\FCC(s)$ and $Y\in\Facnei_\FDD(s)$, there exists
$Z\in\Facnei_{\FCC\cup\FDD}(s)$ such that \(Z\subseteq X\cap Y\).

\emph{Intuitively, for disjoint coalitions, any two actual powers can be jointly
refined to an actual power of their union that is compatible with both.}

\item $\FACNF$ is \Fdefs{$\PAC$-deterministic} if, for all $s\in\FST$ and all
$X\in\Facnei_\FAG(s)$, the set $X$ is a singleton.

\end{itemize}

\end{definition}

The qualifier ``two-agent'' in ``two-agent $\PAC$-independent'' is essential.
For two agents, the only nontrivial disjoint proper coalitions are the two
singletons, so the condition above captures the relevant independence
constraint. For three or more agents, however, the same pairwise set-level
condition does not control the compatibility of action witnesses across
overlapping coalitions. This is precisely the obstruction exhibited in
Section~\ref{subsec:04-03-weak-finite-agent-pac-representativeness-is-not-sufficient}.
By contrast, the notions of $\PAC$-seriality and $\PAC$-determinism are
unchanged for arbitrary finite nonempty sets of agents.

For each $\FXX\in\FES$, we say that a two-agent $\PAC$-representative actual
neighborhood frame is a
\Fdefs{two-agent $\PAC$-representative actual neighborhood $\FXX$-frame} if it
satisfies all $\PAC$-properties whose symbols occur in $\FXX$. No failure of
the omitted properties is required.

%%%%%%%%%%%%%%%
%%%%%%%%%%%%%%%
\begin{theorem}[Two-agent actual havingness theorem]
\label{thm:01-two-agent-actual-havingness-theorem}

Assume $|\FAG|=2$ and $\FXX\in\FES$.

For every two-agent actual neighborhood frame, if it $\PAC$-represents a
general concurrent game $\FXX$-frame, then it is a two-agent
$\PAC$-representative actual neighborhood $\FXX$-frame.

\end{theorem}

%%%%%%%%%%%%%%%
%%%%%%%%%%%%%%%
\begin{theorem}[Two-agent actual enoughness theorem]
\label{thm:02-two-agent-actual-enoughness-theorem}

Assume $|\FAG|=2$ and $\FXX\in\FES$.

Every two-agent $\PAC$-representative actual neighborhood $\FXX$-frame
$\PAC$-represents some general concurrent game $\FXX$-frame.

\end{theorem}

%%%%%%%%%%%%%%%
%%%%%%%%%%%%%%%
\subsection{Alpha powers}
\label{subsec:04-02-alpha-powers}

%%%%%%%%%%%%%%%
%%%%%%%%%%%%%%%
\begin{definition}[Finite-agent $\PAL$-representative alpha neighborhood frames]
\label{def:12-finite-agent-pal-representative-alpha-neighborhood-frames}

Assume $\FAG$ is finite.

Let \(\FALNF=(\FST,\{\Falnei_\FCC\mid \FCC\subseteq\FAG\})\) be an alpha neighborhood frame. We say that $\FALNF$ is
\Fdefs{$\PAL$-representative} if the following conditions are satisfied.

\begin{enumerate}

\item \Fdefs{Alpha triviality of the empty coalition.}
For all $s\in\FST$, if $\Falnei_\emptyset(s)\neq\emptyset$, then
$\Falneinc_\emptyset(s)$ is a singleton.

\emph{Intuitively, when the empty-coalition alpha neighborhood is nonempty,
alpha triviality makes its nonmonotonic core a singleton. Its unique element,
equivalently $\bigcup\Falneinc_\emptyset(s)$, plays the role of the successor
set of $s$.}

\item \Fdefs{Liveness.}
For all $s\in\FST$ and all $\FCC\subseteq\FAG$, \(\emptyset\notin\Falnei_\FCC(s)\).

\emph{Intuitively, no coalition has the absurd alpha power $\emptyset$.}

\item \Fdefs{Groundedness of alpha powers.}
For all $s\in\FST$, all $\FCC\subseteq\FAG$, and all
$X\in\Falnei_\FCC(s)$, there exists $Y\in\Falnei_\FCC(s)$ such that
\(Y\subseteq \bigcup\Falneinc_\emptyset(s) \text{ and } Y\subseteq X.\)

\emph{This condition says that every alpha power has a smaller alpha power
grounded in the successor set of $s$. The original set $X$ may contain states
outside the successor set, but only because alpha neighborhoods are upward
closed.}

\item \Fdefs{Coalition monotonicity of alpha powers.}
For all $s\in\FST$ and all $\FCC,\FDD\subseteq\FAG$ with $\FCC\subseteq\FDD$, \(\Falnei_\FCC(s)\subseteq\Falnei_\FDD(s)\).

\emph{Intuitively, enlarging a coalition cannot destroy an alpha power: if
$\FCC$ can guarantee $X$, then any larger coalition $\FDD$ can also guarantee
$X$.}

\end{enumerate}

\end{definition}

%%%%%%%%%%%%%%%
%%%%%%%%%%%%%%%
\begin{definition}[Seriality, independence, and determinism of $\PAL$-representative alpha neighborhood frames]
\label{def:13-seriality-independence-and-determinism-of-pal-representative-alpha-neighborhood-frames}

Let \(\FALNF=(\FST,\{\Falnei_\FCC\mid \FCC\subseteq\FAG\})\) be an $\PAL$-representative alpha neighborhood frame. We say that:

\begin{itemize}

\item $\FALNF$ is \Fdefs{$\PAL$-serial} if, for all $s\in\FST$ and all
$\FCC\subseteq\FAG$, \(\Falnei_\FCC(s)\neq\emptyset\).

\item $\FALNF$ is \Fdefs{$\PAL$-independent} if, for all $s\in\FST$ and all
disjoint coalitions $\FCC,\FDD\subseteq\FAG$, whenever
$X\in\Falnei_\FCC(s)$ and $Y\in\Falnei_\FDD(s)$, we have \(X\cap Y\in\Falnei_{\FCC\cup\FDD}(s)\).

\emph{Intuitively, if two disjoint coalitions can guarantee $X$ and $Y$,
respectively, then their union can guarantee the joint requirement
$X\cap Y$.}

\item $\FALNF$ is \Fdefs{$\PAL$-deterministic} if, for all $s\in\FST$:
\begin{itemize}
    \item[(a)] every element of $\Falneinc_\FAG(s)$ is a singleton; and
    \item[(b)] $\bigcup\Falneinc_\emptyset(s)\subseteq
    \bigcup\Falneinc_\FAG(s)$.
\end{itemize}

\emph{Intuitively, condition (a) says that every minimal grand-coalition alpha
power corresponds to a singleton outcome, as expected under determinism.
Condition (b) says that every possible successor state is realized by some such
grand-coalition outcome. Thus the grand coalition accounts for all successors
of $s$ by singleton outcomes.}

\end{itemize}

\end{definition}

We continue to use the symbols $\mathtt{S}$, $\mathtt{I}$, and $\mathtt{D}$ for
seriality, independence, and determinism, respectively, and $\FES$ for the set
of strings encoding their eight possible combinations. For each $\FXX\in\FES$,
we say that an $\PAL$-representative alpha neighborhood frame is a
\Fdefs{$\PAL$-representative alpha neighborhood $\FXX$-frame} if it satisfies
all $\PAL$-properties whose symbols occur in $\FXX$. No failure of the omitted
properties is required.

%%%%%%%%%%%%%%%
%%%%%%%%%%%%%%%
\begin{theorem}[Finite-agent alpha havingness theorem]
\label{thm:03-finite-agent-alpha-havingness-theorem}

Assume $\FAG$ is finite and $\FXX\in\FES$.

For every alpha neighborhood frame, if it $\PAL$-represents a general
concurrent game $\FXX$-frame, then it is an $\PAL$-representative alpha
neighborhood $\FXX$-frame.

\end{theorem}

%%%%%%%%%%%%%%%
%%%%%%%%%%%%%%%
\begin{theorem}[Two-agent alpha enoughness theorem]
\label{thm:04-two-agent-alpha-enoughness-theorem}

Assume $|\FAG|=2$ and $\FXX\in\FES$.

Every two-agent $\PAL$-representative alpha neighborhood $\FXX$-frame
$\PAL$-represents some general concurrent game $\FXX$-frame.

\end{theorem}

The actual representation results recalled above are essentially two-agent
results. The next section shows that the direct finite-agent reading of the
two-agent actual conditions is not sufficient for $\PAC$-representability once
three agents are allowed. This motivates the stronger finite-agent notion of
$\PAC$-representativeness introduced later.

%%%%%%%%%%%%%%%
%%%%%%%%%%%%%%%
\subsection{Weak finite-agent \texorpdfstring{$\PAC$}{PAC}-representativeness is not sufficient}
\label{subsec:04-03-weak-finite-agent-pac-representativeness-is-not-sufficient}

The two-agent characterization of actual powers uses four elementary
conditions: actual triviality of the empty coalition, liveness, actual power
inclusion, and actual power decomposition. These conditions have direct
finite-agent readings.

%%%%%%%%%%%%%%%
%%%%%%%%%%%%%%%
\begin{definition}[Weak finite-agent \texorpdfstring{$\PAC$}{PAC}-representativeness]
\label{def:14-weak-finite-agent-pac-representativeness}

Let \(\FACNF=(\FST,\{\Facnei_\FCC\mid \FCC\subseteq \FAG\})\) be an actual neighborhood frame, where \(\FAG\) is finite and nonempty. We say
that \(\FACNF\) is \Fdefs{weak finite-agent \(\PAC\)-representative} if the
following conditions are satisfied:

\begin{enumerate}

    \item \Fdefs{Actual triviality of the empty coalition:}
    for every \(s\in\FST\), if \(\Facnei_\emptyset(s)\) is nonempty, then
    \(\Facnei_\emptyset(s)\) is a singleton.

    \item \Fdefs{Liveness:}
    for every \(s\in\FST\) and every \(\FCC\subseteq\FAG\), \(\emptyset\notin\Facnei_\FCC(s)\).

    \item \Fdefs{Actual power inclusion:}
    for every \(s\in\FST\), all coalitions
    \(\FCC\subseteq\FDD\subseteq\FAG\), and every \(X\in\Facnei_\FDD(s)\),
    there exists \(Y\in\Facnei_\FCC(s)\) such that \(X\subseteq Y\).

    \item \Fdefs{Actual power decomposition:}
    for every \(s\in\FST\), all coalitions
    \(\FCC\subseteq\FDD\subseteq\FAG\), and every \(X\in\Facnei_\FCC(s)\),
    there exists a set \(\Delta\subseteq\Facnei_\FDD(s)\) such that \(X=\bigcup\Delta\).

\end{enumerate}

\end{definition}

The following proposition records the necessity of these weak finite-agent
conditions for \(\PAC\)-representability by general concurrent game frames.

%%%%%%%%%%%%%%%
%%%%%%%%%%%%%%%
\begin{proposition}[Necessity of the weak finite-agent conditions]
\label{prop:01-necessity-of-the-weak-finite-agent-conditions}

Let \(\FACNF=(\FST,\{\Facnei_\FCC\mid \FCC\subseteq\FAG\})\) be an actual neighborhood frame, where \(\FAG\) is finite and nonempty. If
\(\FACNF\) \(\PAC\)-represents some general concurrent game frame, then
\(\FACNF\) is weak finite-agent \(\PAC\)-representative.

\end{proposition}

The proof is given in \cite{chen_representation_2026}. 
The converse already fails with three agents. We recall the counterexample from
\cite{chen_representation_2026}.

%%%%%%%%%%%%%%%
%%%%%%%%%%%%%%%
\begin{example}
\label{ex:02-example}

Let \(\FAG=\{a,b,c\}\) and \(\FST=\{s,t_1,t_2,u,v\}\). Put
\[
Z=\{t_1,t_2\},
\qquad
U=\{u\},
\qquad
V=\{v\},
\]
and
\[
P=Z\cup U=\{t_1,t_2,u\},
\qquad
Q=Z\cup V=\{t_1,t_2,v\},
\qquad
O=Z\cup U\cup V=\{t_1,t_2,u,v\}.
\]

When no confusion can arise, we write \(a\) for \(\{a\}\), \(ab\) for
\(\{a,b\}\), and similarly for other coalitions. Define an actual neighborhood
frame \(\FACNF=(\FST,\{\Facnei_\FCC\mid \FCC\subseteq\FAG\})\) as follows. At the distinguished state \(s\), set
\[
\begin{array}{rcl@{\qquad}rcl}
\Facnei_\emptyset(s)&=&\{O\},
&
\Facnei_a(s)&=&\{P,Q\},
\\
\Facnei_b(s)&=&\{O\},
&
\Facnei_c(s)&=&\{O\},
\\
\Facnei_{ab}(s)&=&\{P,\{t_1\},\{t_2\},V\},
&
\Facnei_{ac}(s)&=&\{Q,\{t_1\},\{t_2\},U\},
\\
\Facnei_{bc}(s)&=&\{O\},
&
\Facnei_\FAG(s)&=&\{Z,\{t_1\},\{t_2\},U,V\}.
\end{array}
\]
At every other state \(x\in O\), set \(\Facnei_\FCC(x)=\{\FST\}\) for every coalition \(\FCC\subseteq\FAG\).

\end{example}

It is shown in \cite{chen_representation_2026} that \(\FACNF\) is weak
finite-agent \(\PAC\)-representative, but that it does not
\(\PAC\)-represent any general concurrent game frame. We recall only the
argument establishing the latter claim.

Suppose, for contradiction, that \(\FACNF\) \(\PAC\)-represents a general
concurrent game frame
\[
\FGCGF=
(\FST,\FAC,\{\Fav_\FCC\mid\FCC\subseteq\FAG\},
\{\Fout_\FCC\mid\FCC\subseteq\FAG\}).
\]
By \(\PAC\)-representation, since \(Z\in\Facnei_\FAG(s)\), choose
\(\sigma_\FAG\in\Fav_\FAG(s)\) such that \(\Fout_\FAG(s,\sigma_\FAG)=Z\).
Let \(\sigma_{ab}=\sigma_\FAG|_{ab},  \sigma_{ac}=\sigma_\FAG|_{ac},  \sigma_a=\sigma_\FAG|_a.\)

By outcome monotonicity, \(Z=\Fout_\FAG(s,\sigma_\FAG)  \subseteq \Fout_{ab}(s,\sigma_{ab}).\)
Thus \(\Fout_{ab}(s,\sigma_{ab})\neq\emptyset\), and the ODA-condition gives
\(\sigma_{ab}\in\Fav_{ab}(s)\). By \(\PAC\)-representation,
\(\Fout_{ab}(s,\sigma_{ab})\in\Facnei_{ab}(s) = \{P,\{t_1\},\{t_2\},V\}.\)
Among these sets, only \(P\) contains \(Z\). Hence \(\Fout_{ab}(s,\sigma_{ab})=P\).
The same argument applied to \(ac\) gives \(\Fout_{ac}(s,\sigma_{ac})=Q\).

Since \(\sigma_a\) is the common restriction of \(\sigma_{ab}\) and
\(\sigma_{ac}\), outcome monotonicity yields
\[
P=\Fout_{ab}(s,\sigma_{ab})\subseteq\Fout_a(s,\sigma_a)
\text{ and }
Q=\Fout_{ac}(s,\sigma_{ac})\subseteq\Fout_a(s,\sigma_a).
\]
Therefore \(P\cup Q\subseteq\Fout_a(s,\sigma_a)\).
In particular, \(\Fout_a(s,\sigma_a)\neq\emptyset\), so the ODA-condition gives
\(\sigma_a\in\Fav_a(s)\). By \(\PAC\)-representation, \(\Fout_a(s,\sigma_a)\in\Facnei_a(s)=\{P,Q\}\).
This is impossible: \(\Fout_a(s,\sigma_a)\) must contain \(P\cup Q=O\), while
both possible values \(P\) and \(Q\) are proper subsets of \(O\).

The obstruction comes from incomparable overlapping coalitions. With three
agents, coalitions such as \(\{a,b\}\) and \(\{a,c\}\) are incomparable but
overlap in \(a\). If their powers are induced by a common game frame, then the powers
arising from restrictions of the same full joint action must remain compatible
on the shared \(a\)-component.

The four conditions of weak finite-agent \(\PAC\)-representativeness control
how powers behave along coalition inclusions, but they do not ensure that
powers assigned to overlapping coalitions can be generated by one coherent
system of action restrictions. This is why weak finite-agent
\(\PAC\)-representativeness is not sufficient.

%%%%%%%%%%%%%%%
%%%%%%%%%%%%%%%
\section{Representing actual powers via abstract presentations}
\label{sec:05-representing-actual-powers-via-abstract-presentations}

The counterexample in the previous section shows that weak finite-agent
\(\PAC\)-representativeness misses a compatibility requirement. The missing
requirement is not only about how powers behave along coalition inclusions. It
also concerns how powers assigned to incomparable overlapping coalitions are
coordinated. If an actual neighborhood frame arises from a game frame, then
powers of different coalitions that arise from restrictions of the same full
joint action must be compatible through their shared components.

This section introduces abstract actual presentations to record this
coordination. At each state, a set of indices is used; intuitively, an index
stands for an abstract full choice of actions by all agents. For each agent
\(a\), an equivalence relation on indices records when two indices have the same
\(a\)-component. The relation for a coalition is obtained by intersecting the
corresponding agent-equivalence relations, and hence records agreement on all
agents in that coalition. Each index is also assigned a set of grand-coalition
outcomes. Given an index and a coalition, the corresponding coalition power is
recovered as the union of the grand-coalition outcome sets assigned to all
indices that agree with the original index on that coalition.

Using these presentations, together with liveness, we define \(\PAC\)-abstract-representative actual
neighborhood frames. We also formulate the additional conditions corresponding
to seriality, independence, and determinism. Seriality and determinism are
stated directly on the neighborhood frame, while independence is expressed as an
amalgamation property of abstract actual presentations. We then prove the
corresponding havingness and enoughness theorems. Consequently, for actual
powers, the eight finite-agent classes of general concurrent game frames are
characterized, up to \(\PAC\)-representability, by the corresponding
\(\PAC\)-abstract-representative actual neighborhood frames.

The notion is deliberately abstract: it quantifies over indices rather than
speaking only in terms of powers or neighborhoods. Its purpose is to isolate the
coordination structure that is absent from weak finite-agent
\(\PAC\)-representativeness. In later sections, for the classes that do not
impose independence of agents, abstract presentability is replaced by an
equivalent condition formulated in terms of simultaneous compatible choices of
powers for all coalitions. This equivalence then yields a more concrete,
index-free representation theorem.

%%%%%%%%%%%%%%%
%%%%%%%%%%%%%%%
\subsection{Abstract actual presentations}
\label{subsec:05-01-abstract-actual-presentations}

We now introduce abstract actual presentations. These are indexed structures
that record the abstract witnesses from which actual coalition powers can be
recovered. Their existence will later be used to formulate the corresponding
notion of representativeness.

%%%%%%%%%%%%%%%
%%%%%%%%%%%%%%%
\begin{definition}[Abstract actual presentations]
\label{def:15-abstract-actual-presentations}

Let \(\FACNF=(\FST,\{\Facnei_\FCC\mid \FCC\subseteq\FAG\})\) be an actual neighborhood frame, where \(\FAG\) is finite and nonempty, and fix
\(s\in\FST\). Write \(N^s_\FCC=\Facnei_\FCC(s)\) for each coalition \(\FCC\subseteq\FAG\).

An \Fdefs{abstract actual presentation} of \(\FACNF\) at \(s\) is a tuple \(\mathbb I_s=(I_s,o_s,\{\equiv^s_a\}_{a\in\FAG})\), where \(I_s\) is a set whose elements are called \Fdefs{abstract full-action
indices} at \(s\), \(o_s:I_s\to\mathcal P(\FST)\) is a
\Fdefs{grand-outcome map}, and, for every agent \(a\in\FAG\),
\(\equiv^s_a\) is an equivalence relation on \(I_s\).

From this tuple we derive the following objects. For every coalition
\(\FCC\subseteq\FAG\), the \Fdefs{coalition-equivalence relation}
\(\equiv^s_\FCC\) on \(I_s\) is defined by
\(i\equiv^s_\FCC j \text{iff} i\equiv^s_a j \text{ for every } a\in\FCC.\)
For \(\FCC=\emptyset\), this is the universal relation on \(I_s\).

For every coalition \(\FCC\subseteq\FAG\), the \Fdefs{presented-power map} \(\rho^s_\FCC:I_s\to\mathcal P(\FST)\) is defined by
\[
\rho^s_\FCC(i)
=
\bigcup
\{\,o_s(j)
\mid
j\in I_s \text{ and } i\equiv^s_\FCC j\,\}.
\]

The tuple \(\mathbb I_s\) is required to satisfy the following conditions.

\begin{itemize}

\item \Fdefs{Grand compatibility.} For all \(i,j\in I_s\),
\(i\equiv^s_\FAG j \Longrightarrow o_s(i)=o_s(j).\)

\item \Fdefs{Coverage.} For every coalition \(\FCC\subseteq\FAG\),
\[
N^s_\FCC
=
\{\,\rho^s_\FCC(i)\mid i\in I_s\,\}.
\]

\end{itemize}

A family \(\mathbb I=(\mathbb I_s)_{s\in\FST}\) is an \Fdefs{abstract actual presentation} of \(\FACNF\) if
\(\mathbb I_s\) is an abstract actual presentation of \(\FACNF\) at \(s\), for
every \(s\in\FST\).

\end{definition}

The intended reading is as follows. An index \(i\in I_s\) represents an
abstract full joint action at \(s\). The relation \(i\equiv^s_a j\)
records that the indices \(i\) and \(j\) have the same abstract
\(a\)-component. Thus \(i\equiv^s_\FCC j\) means that the two indices agree on
the \(\FCC\)-component. The set \(o_s(i)\) is the exact outcome set of the full
index, while \(\rho^s_\FCC(i)\) is the exact outcome range left by the
\(\FCC\)-component of that index. The definition allows \(I_s=\emptyset\); in
that case, coverage forces all neighborhoods at \(s\) to be empty.

%%%%%%%%%%%%%%%
%%%%%%%%%%%%%%%
\begin{lemma}[Basic properties of presented powers]
\label{lem:01-basic-properties-of-presented-powers}

Let \(\mathbb I_s=(I_s,o_s,\{\equiv^s_a\}_{a\in\FAG})\) be an abstract actual
presentation at \(s\). Then the following hold.

\begin{enumerate}

\item For every \(i\in I_s\), \(\rho^s_\FAG(i)=o_s(i)\).

\item For every \(i\in I_s\) and all coalitions
\(\FCC\subseteq\FDD\subseteq\FAG\), \(\rho^s_\FDD(i)\subseteq \rho^s_\FCC(i)\).

\item For all \(i,j\in I_s\) and all coalitions \(\FDD\subseteq\FCC\subseteq\FAG\),
if \(i\equiv^s_\FCC j\), then \(\rho^s_\FDD(i)=\rho^s_\FDD(j)\).

\end{enumerate}

\end{lemma}

%%%%%%%%%%%%%%%
%%%%%%%%%%%%%%%
\begin{proof}

For item 1, fix \(i\in I_s\). By definition, \(\rho^s_\FAG(i) = \bigcup\{\,o_s(j)\mid i\equiv^s_\FAG j\,\}\).
This union contains \(o_s(i)\). Conversely, if \(o_s(j)\) occurs in the union,
then \(i\equiv^s_\FAG j\), and hence \(o_s(j)=o_s(i)\) by grand compatibility.
Therefore \(\rho^s_\FAG(i)=o_s(i)\).

For item 2, fix \(i\in I_s\) and assume
\(\FCC\subseteq\FDD\). If \(i\equiv^s_\FDD j\), then
\(i\equiv^s_\FCC j\). Thus the union defining \(\rho^s_\FDD(i)\) is taken over
a subset of the indices used in the union defining \(\rho^s_\FCC(i)\). Hence \(\rho^s_\FDD(i)\subseteq\rho^s_\FCC(i)\).

For item 3, assume \(i\equiv^s_\FCC j\) and
\(\FDD\subseteq\FCC\). Then \(i\equiv^s_\FDD j\). Since
\(\equiv^s_\FDD\) is an equivalence relation, for every \(k\in I_s\), \(i\equiv^s_\FDD k \text{iff} j\equiv^s_\FDD k\).
Therefore the two unions defining \(\rho^s_\FDD(i)\) and
\(\rho^s_\FDD(j)\) range over the same indices. Hence they are equal.

\end{proof}

%%%%%%%%%%%%%%%
%%%%%%%%%%%%%%%
\subsection{\texorpdfstring{$\PAC$}{PAC}-abstract-representativeness}
\label{subsec:05-02-pac-abstract-representativeness}

We now use abstract actual presentations to define a witness-based notion of
representativeness for actual neighborhood frames.

%%%%%%%%%%%%%%%
%%%%%%%%%%%%%%%
\begin{definition}[\texorpdfstring{$\PAC$}{PAC}-abstract-representative actual neighborhood frames]
\label{def:16-pac-abstract-representative-actual-neighborhood-frames}

Let \(\FACNF=(\FST,\{\Facnei_\FCC\mid \FCC\subseteq\FAG\})\) be an actual neighborhood frame, where \(\FAG\) is finite and nonempty. We say
that \(\FACNF\) is \Fdefs{\(\PAC\)-abstract-representative} if the following
conditions hold.

\begin{enumerate}

\item \Fdefs{Abstract presentability.} There exists an abstract actual
presentation \(\mathbb I=(\mathbb I_s)_{s\in\FST}\) of \(\FACNF\).

\item \Fdefs{Liveness.} For every \(s\in\FST\) and every
\(\FCC\subseteq\FAG\), \(\emptyset\notin\Facnei_\FCC(s)\).

\end{enumerate}

\end{definition}

The following result will be used later.

%%%%%%%%%%%%%%%
%%%%%%%%%%%%%%%
\begin{lemma}[Empty-coalition triviality from abstract presentability]
\label{lem:02-empty-coalition-triviality-from-abstract-presentability}

Let \(\FACNF=(\FST,\{\Facnei_\FCC\mid \FCC\subseteq\FAG\})\) be an actual neighborhood frame, where \(\FAG\) is finite and nonempty. If
\(\FACNF\) has an abstract actual presentation, then \(\FACNF\) is trivial for
the empty coalition. That is, for every \(s\in\FST\), if
\(\Facnei_\emptyset(s)\) is nonempty, then it is a singleton.

\end{lemma}

%%%%%%%%%%%%%%%
%%%%%%%%%%%%%%%
\begin{proof}

Let \(\mathbb I=(\mathbb I_s)_{s\in\FST}\) be an abstract actual presentation of \(\FACNF\). Fix \(s\in\FST\), and write \(\mathbb I_s=(I_s,o_s,\{\equiv^s_a\}_{a\in\FAG})\).
We show that \(\Facnei_\emptyset(s)\) has at most one member.

If \(I_s=\emptyset\), then by coverage
\[
\Facnei_\emptyset(s)
=
\{\rho^s_\emptyset(i)\mid i\in I_s\}
=
\emptyset,
\]
so there is nothing to prove.

Suppose \(I_s\neq\emptyset\). Since \(\equiv^s_\emptyset\) is the universal
relation on \(I_s\), for every \(i\in I_s\),
\[
\rho^s_\emptyset(i)
=
\bigcup\{o_s(j)\mid j\in I_s \text{ and } i\equiv^s_\emptyset j\}
=
\bigcup\{o_s(j)\mid j\in I_s\}.
\]
Thus \(\rho^s_\emptyset(i)\) is independent of \(i\). Hence \(\{\rho^s_\emptyset(i)\mid i\in I_s\}\)
has exactly one member. By coverage, \(\Facnei_\emptyset(s) = \{\rho^s_\emptyset(i)\mid i\in I_s\}\).
Therefore \(\Facnei_\emptyset(s)\) is either empty or a singleton.

Since \(s\) was arbitrary, \(\FACNF\) is trivial for the empty coalition.

\end{proof}

%%%%%%%%%%%%%%%
%%%%%%%%%%%%%%%
\begin{definition}[Amalgamating abstract actual presentations]
\label{def:17-amalgamating-abstract-actual-presentations}

Let \(\mathbb I=(\mathbb I_s)_{s\in\FST}\) be an abstract actual presentation of an actual neighborhood frame over
\(\FAG\), where \(\mathbb I_s=(I_s,o_s,\{\equiv^s_a\}_{a\in\FAG})\) for each \(s\in\FST\).

For \(s\in\FST\), we say that the local presentation \(\mathbb I_s\) is
\Fdefs{amalgamating} if, for all \(i,j\in I_s\) and all disjoint coalitions
\(\FCC,\FDD\subseteq\FAG\), there exists \(k\in I_s\) such that \(i\equiv^s_\FCC k \text{ and } j\equiv^s_\FDD k\).

We say that the presentation \(\mathbb I\) is \Fdefs{amalgamating} if
\(\mathbb I_s\) is amalgamating for every \(s\in\FST\).

\end{definition}

%%%%%%%%%%%%%%%
%%%%%%%%%%%%%%%
\begin{definition}[Seriality, independence, and determinism at the abstract level]
\label{def:18-seriality-independence-and-determinism-at-the-abstract-level}

Let \(\FACNF\) be an \(\PAC\)-abstract-representative actual neighborhood frame.

\begin{itemize}

\item \(\FACNF\) is \Fdefs{\(\PAC\)-abstract-serial} if, for every
\(s\in\FST\) and every \(\FCC\subseteq\FAG\), \(\Facnei_\FCC(s)\neq\emptyset\).

\item \(\FACNF\) is \Fdefs{\(\PAC\)-abstract-independent} if it has an
amalgamating abstract actual presentation.

\item \(\FACNF\) is \Fdefs{\(\PAC\)-abstract-deterministic} if, for every
\(s\in\FST\) and every \(X\in\Facnei_\FAG(s)\), the set \(X\) is a singleton.

\end{itemize}

\end{definition}

Seriality and determinism are stated directly in terms of neighborhood powers.
Independence is stated at the index level, through the existence of an
amalgamating abstract actual presentation.

For every \(\FXX\in\FES\), an \(\PAC\)-abstract-representative actual
neighborhood frame is a \Fdefs{\(\PAC\)-abstract-representative actual
neighborhood \(\FXX\)-frame} if it satisfies the corresponding additional
conditions: it is \(\PAC\)-abstract-serial when \(\mathtt S\) occurs in
\(\FXX\), \(\PAC\)-abstract-independent when \(\mathtt I\) occurs in \(\FXX\),
and \(\PAC\)-abstract-deterministic when \(\mathtt D\) occurs in \(\FXX\).

%%%%%%%%%%%%%%%
%%%%%%%%%%%%%%%
\subsection{Abstract actual havingness}
\label{subsec:05-03-abstract-actual-havingness}

The havingness direction starts from a genuine game frame and constructs an
abstract presentation of its induced actual neighborhood frame. At each state,
the indices are the available full joint actions. Equivalence for an agent is
equality of that agent's component.

%%%%%%%%%%%%%%%
%%%%%%%%%%%%%%%
\begin{theorem}[Abstract actual havingness theorem]
\label{thm:05-abstract-actual-havingness-theorem}

Let \(\FAG\) be finite and nonempty, and let \(\FXX\in\FES\). Let
\[
\FGCGF=(\FST,\FAC,\{\Fav_\FCC\mid \FCC\subseteq\FAG\},
\{\Fout_\FCC\mid \FCC\subseteq\FAG\})
\]
be a general concurrent game \(\FXX\)-frame. Let
\[
\FACNF=(\FST,\{\Facnei_\FCC\mid \FCC\subseteq\FAG\})
\]
be the actual neighborhood frame induced by \(\FGCGF\), that is,
\[
\Facnei_\FCC(s)
=
\{\,\Fout_\FCC(s,\sigma_\FCC)
\mid
\sigma_\FCC\in\Fav_\FCC(s)\,\}
\]
for every \(s\in\FST\) and every \(\FCC\subseteq\FAG\). Then \(\FACNF\) is a
\(\PAC\)-abstract-representative actual neighborhood \(\FXX\)-frame.

\end{theorem}

%%%%%%%%%%%%%%%
%%%%%%%%%%%%%%%
\begin{proof}

We first prove that \(\FACNF\) is \(\PAC\)-abstract-representative. We then
verify the additional conditions required by \(\FXX\).

\medskip
\noindent
\textbf{Abstract presentability.}
Fix \(s\in\FST\). We construct an abstract actual presentation at \(s\). Put
\(I_s=\Fav_\FAG(s).\)
For \(\sigma\in I_s\), define
\(o_s(\sigma)=\Fout_\FAG(s,\sigma).\)
For every \(a\in\FAG\), define an equivalence relation on \(I_s\) by
\(\sigma\equiv^s_a\tau \text{iff} \sigma(a)=\tau(a).\)
Let
\(\mathbb I_s=(I_s,o_s,\{\equiv^s_a\}_{a\in\FAG}).\)
We show that \(\mathbb I_s\) is an abstract actual presentation of
\(\FACNF\) at \(s\).

Grand compatibility is immediate. If \(\sigma\equiv^s_\FAG\tau\), then
\(\sigma=\tau\), and hence
\(o_s(\sigma)=o_s(\tau).\)

We next record the key computation. For every \(\sigma\in I_s\) and every
\(\FCC\subseteq\FAG\),
\[
\begin{aligned}
\rho^s_\FCC(\sigma)
&=
\bigcup
\{\,o_s(\tau)
\mid
\tau\in I_s \text{ and } \sigma\equiv^s_\FCC\tau\,\} \\
&=
\bigcup
\{\,\Fout_\FAG(s,\tau)
\mid
\tau\in\Fav_\FAG(s) \text{ and } \tau|_\FCC=\sigma|_\FCC\,\} \\
&=
\Fout_\FCC(s,\sigma|_\FCC).
\end{aligned}
\tag{AHP}
\label{eq:01-abstract-having-rho-is-outcome}
\]
The second equality uses the definition of \(\equiv^s_\FCC\), and the last
equality is the alternative GCI-condition.

It remains to verify coverage, that is, for every coalition
\(\FCC\subseteq\FAG\),
\(\Facnei_\FCC(s) = \{\,\rho^s_\FCC(\sigma)\mid \sigma\in I_s\,\}.\)

First, let \(\sigma\in I_s\). Since \(\sigma\in\Fav_\FAG(s)\), the
ODA-condition gives
\(\Fout_\FAG(s,\sigma)\neq\emptyset.\)
By outcome monotonicity,
\(\Fout_\FAG(s,\sigma) \subseteq \Fout_\FCC(s,\sigma|_\FCC).\)
Hence
\(\Fout_\FCC(s,\sigma|_\FCC)\neq\emptyset.\)
By the ODA-condition again,
\(\sigma|_\FCC\in\Fav_\FCC(s).\)
Using \eqref{eq:01-abstract-having-rho-is-outcome}, we get
\(\rho^s_\FCC(\sigma) = \Fout_\FCC(s,\sigma|_\FCC) \in\Facnei_\FCC(s).\)
Thus
\(\{\,\rho^s_\FCC(\sigma)\mid \sigma\in I_s\,\} \subseteq \Facnei_\FCC(s).\)

Conversely, let \(X\in\Facnei_\FCC(s)\). Choose
\(\eta_\FCC\in\Fav_\FCC(s)\) such that
\(X=\Fout_\FCC(s,\eta_\FCC).\)
By the ODA-condition,
\(X\neq\emptyset.\)
By the alternative GCI-condition,
\[
X
=
\bigcup
\{\,\Fout_\FAG(s,\tau)
\mid
\tau\in\Fav_\FAG(s) \text{ and } \tau|_\FCC=\eta_\FCC\,\}.
\]
Since \(X\neq\emptyset\), there exists
\(\tau\in\Fav_\FAG(s)=I_s\) such that
\(\tau|_\FCC=\eta_\FCC.\)
Therefore, by \eqref{eq:01-abstract-having-rho-is-outcome},
\(\rho^s_\FCC(\tau) = \Fout_\FCC(s,\tau|_\FCC) = \Fout_\FCC(s,\eta_\FCC) = X.\)
This proves the reverse inclusion, and hence coverage holds.

\medskip
\noindent
\textbf{Liveness.}
Let \(s\in\FST\), let \(\FCC\subseteq\FAG\), and let
\(X\in\Facnei_\FCC(s)\). Then
\(X=\Fout_\FCC(s,\sigma_\FCC)\)
for some \(\sigma_\FCC\in\Fav_\FCC(s)\). By the ODA-condition,
\(\Fout_\FCC(s,\sigma_\FCC)\neq\emptyset.\)
Hence \(X\neq\emptyset\). Therefore
\(\emptyset\notin\Facnei_\FCC(s)\)
for every \(s\in\FST\) and every \(\FCC\subseteq\FAG\).

Since \(s\) was arbitrary, the family
\(\mathbb I=(\mathbb I_s)_{s\in\FST}\)
is an abstract actual presentation of \(\FACNF\). Therefore \(\FACNF\) is
\(\PAC\)-abstract-representative.

\medskip
\noindent
\textbf{Seriality.}
Assume that \(\mathtt S\) occurs in \(\FXX\). Then \(\FGCGF\) is serial. Hence
\(\Fav_\FCC(s)\neq\emptyset\)
for every \(s\in\FST\) and every \(\FCC\subseteq\FAG\). By the definition of
the induced actual neighborhood frame, it follows that
\(\Facnei_\FCC(s)\neq\emptyset\)
for every \(s\in\FST\) and every \(\FCC\subseteq\FAG\). Thus \(\FACNF\) is
\(\PAC\)-abstract-serial.

\medskip
\noindent
\textbf{Independence.}
Assume that \(\mathtt I\) occurs in \(\FXX\). Then \(\FGCGF\) is independent.
We show that the abstract presentation constructed above is amalgamating.

Fix \(s\in\FST\), let \(\sigma,\tau\in I_s=\Fav_\FAG(s)\), and let
\(\FCC,\FDD\subseteq\FAG\) be disjoint coalitions. As in the coverage argument,
outcome monotonicity and the ODA-condition give
\(\sigma|_\FCC\in\Fav_\FCC(s) \text{ and } \tau|_\FDD\in\Fav_\FDD(s).\)
By independence of \(\FGCGF\),
\(\sigma|_\FCC\cup\tau|_\FDD \in \Fav_{\FCC\cup\FDD}(s).\)
By the ODA-condition,
\(\Fout_{\FCC\cup\FDD} (s,\sigma|_\FCC\cup\tau|_\FDD) \neq\emptyset.\)
Applying the alternative GCI-condition to the partial action
\(\sigma|_\FCC\cup\tau|_\FDD\), we obtain some
\(\kappa\in\Fav_\FAG(s)=I_s\) such that
\(\kappa|_{\FCC\cup\FDD} = \sigma|_\FCC\cup\tau|_\FDD.\)
Since \(\FCC\) and \(\FDD\) are disjoint, this implies
\(\kappa|_\FCC=\sigma|_\FCC \text{ and } \kappa|_\FDD=\tau|_\FDD.\)
Equivalently,
\(\sigma\equiv^s_\FCC\kappa \text{ and } \tau\equiv^s_\FDD\kappa.\)
Thus the presentation \(\mathbb I\) is amalgamating, and \(\FACNF\) is
\(\PAC\)-abstract-independent.

\medskip
\noindent
\textbf{Determinism.}
Assume that \(\mathtt D\) occurs in \(\FXX\). Then \(\FGCGF\) is deterministic.
Let \(s\in\FST\) and let \(X\in\Facnei_\FAG(s)\). Then
\(X=\Fout_\FAG(s,\sigma)\)
for some \(\sigma\in\Fav_\FAG(s)\). By determinism of \(\FGCGF\), the set
\(\Fout_\FAG(s,\sigma)\) is a singleton. Hence \(X\) is a singleton. Thus
\(\FACNF\) is \(\PAC\)-abstract-deterministic.

We have verified all conditions required for \(\FACNF\) to be a
\(\PAC\)-abstract-representative actual neighborhood \(\FXX\)-frame.

\end{proof}

%%%%%%%%%%%%%%%
%%%%%%%%%%%%%%%
\subsection{Abstract actual enoughness}
\label{subsec:05-04-abstract-actual-enoughness}

We now prove the converse direction of
Theorem~\ref{thm:05-abstract-actual-havingness-theorem}. Starting with an abstract actual
presentation of an actual neighborhood frame, we construct a general concurrent
game frame whose actual effectivity functions recover exactly the given
neighborhood functions.

The proof is organized as follows. We first define the generated game frame.
We then prove two technical facts. These facts yield exact recovery of the
neighborhood frame. Finally, we check preservation of seriality, independence,
and determinism.

%%%%%%%%%%%%%%%
%%%%%%%%%%%%%%%
\begin{definition}[Game frame generated by an abstract actual presentation]
\label{def:19-game-frame-generated-by-an-abstract-actual-presentation}

Let
\[
\FACNF=(\FST,\{\Facnei_\FCC\mid \FCC\subseteq\FAG\})
\]
be an actual neighborhood frame, and let
\[
\mathbb I=(\mathbb I_s)_{s\in\FST}
\]
be an abstract actual presentation of \(\FACNF\), where
\[
\mathbb I_s=(I_s,o_s,\{\equiv^s_a\}_{a\in\FAG}).
\]

For \(i\in I_s\) and \(a\in\FAG\), write
\([i]^s_a = \{\,j\in I_s\mid j\equiv^s_a i\,\}.\)

Define an action frame
\[
\FGCGF^{\mathbb I}
=
(\FST,\FAC^{\mathbb I},
\{\Fav_\FCC\mid \FCC\subseteq\FAG\},
\{\Fout_\FCC\mid \FCC\subseteq\FAG\})
\]
as follows.

\begin{itemize}

\item
Choose a symbol \(*\) distinct from all triples below, and define the global
action set by
\[
\FAC^{\mathbb I}
=
\{*\}
\cup
\{\,(s,a,[i]^s_a)
\mid
s\in\FST,
\ a\in\FAG,
\ i\in I_s\,\}.
\]
The dummy action \(*\) ensures that the global action set is nonempty.

\item
For every \(s\in\FST\) and every \(i\in I_s\), define the full joint action
\(\sigma^{s,i}_\FAG\in\FJA_\FAG\)
by setting, for every \(a\in\FAG\),
\(\sigma^{s,i}_\FAG(a)=(s,a,[i]^s_a).\)
We call \(\sigma^{s,i}_\FAG\) the \Fdefs{distinguished full joint action}
generated by \(i\) at \(s\).

Note that not every
element of \(\FJA_\FAG\) is distinguished: a full joint action may use the dummy
action \(*\), or may use actions tagged by states other than \(s\).

Also note that different indices $i$ and $j$ induce the same distinguished full joint action at $s$ if for all $a$, $i \equiv^s_a j$.

\item
Define the grand-coalition outcome function by
\begin{equation}
\label{eq:02-generated-from-abstract-grand-outcome}
\Fout_\FAG(s,\sigma_\FAG)
=
\begin{cases}
 o_s(i),
 &\text{if }\sigma_\FAG=\sigma^{s,i}_\FAG
 \text{ for some }i\in I_s,\\[1mm]
 \emptyset,
 &\text{otherwise.}
\end{cases}
\end{equation}
This is well-defined. Indeed, suppose that
\(\sigma^{s,i}_\FAG=\sigma^{s,j}_\FAG\). Then
\([i]^s_a=[j]^s_a\) for every \(a\in\FAG\). Hence
\(i\equiv^s_a j\) for every \(a\in\FAG\), so
\(i\equiv^s_\FAG j\). By grand compatibility, \(o_s(i)=o_s(j)\).

\item
For every proper coalition \(\FCC\subsetneq\FAG\), define
\begin{equation}
\label{eq:03-generated-from-abstract-lower-outcome}
\Fout_\FCC(s,\sigma_\FCC)
=
\bigcup
\{\,\Fout_\FAG(s,\sigma_\FAG)
\mid
\sigma_\FAG\in\FJA_\FAG
\text{ and }
\sigma_\FAG|_\FCC=\sigma_\FCC\,\}.
\end{equation}

\item
Finally, for every \(\FCC\subseteq\FAG\), define availability by
\begin{equation}
\label{eq:04-generated-from-abstract-availability}
\Fav_\FCC(s)
=
\{\,\sigma_\FCC\in\FJA_\FCC
\mid
\Fout_\FCC(s,\sigma_\FCC)\neq\emptyset\,\}.
\end{equation}

\end{itemize}

\end{definition}

By construction, \(\FGCGF^{\mathbb I}\) satisfies the GCI-condition: this is
immediate from \eqref{eq:03-generated-from-abstract-lower-outcome} for proper
coalitions and trivial for the grand coalition. It also satisfies the
ODA-condition by \eqref{eq:04-generated-from-abstract-availability}. Hence
\(\FGCGF^{\mathbb I}\) is a general concurrent game frame.

The next lemma records two consequences of the way distinguished full joint
actions are generated.

%%%%%%%%%%%%%%%
%%%%%%%%%%%%%%%
\begin{lemma}[Local components and realization of presented powers]
\label{lem:03-local-components-and-realization-of-presented-powers}

Let \(\FGCGF^{\mathbb I}\) be the game frame generated by an abstract actual
presentation \(\mathbb I\).

\begin{itemize}

\item 

For every \(s\in\FST\), all
\(i,j\in I_s\), and every \(\FCC\subseteq\FAG\),
\begin{equation}
\label{eq:05-restriction-equivalence-abstract}
\sigma^{s,i}_\FAG|_\FCC
=
\sigma^{s,j}_\FAG|_\FCC
\Longleftrightarrow
i\equiv^s_\FCC j.
\end{equation}

\item For every \(s\in\FST\), every \(i\in I_s\), and every
\(\FCC\subseteq\FAG\),
\begin{equation}
\label{eq:06-local-realization-abstract}
\Fout_\FCC(s,\sigma^{s,i}_\FAG|_\FCC)
=
\rho^s_\FCC(i).
\end{equation}

\end{itemize}

\end{lemma}

%%%%%%%%%%%%%%%
%%%%%%%%%%%%%%%
\begin{proof}

We first prove \eqref{eq:05-restriction-equivalence-abstract}.
Fix \(s\in\FST\), \(i,j\in I_s\), and \(\FCC\subseteq\FAG\). Since
\(\sigma^{s,i}_\FAG|_\FCC\) and \(\sigma^{s,j}_\FAG|_\FCC\) have the same
domain \(\FCC\), equality of these restrictions is pointwise equality on
\(\FCC\). Hence
\[
\begin{aligned}
\sigma^{s,i}_\FAG|_\FCC
=
\sigma^{s,j}_\FAG|_\FCC
&\Longleftrightarrow
\sigma^{s,i}_\FAG(a)=\sigma^{s,j}_\FAG(a)
\text{ for every }a\in\FCC                                      \\
&\Longleftrightarrow
(s,a,[i]^s_a)=(s,a,[j]^s_a)
\text{ for every }a\in\FCC                                      \\
&\Longleftrightarrow
[i]^s_a=[j]^s_a
\text{ for every }a\in\FCC                                      \\
&\Longleftrightarrow
i\equiv^s_a j
\text{ for every }a\in\FCC                                      \\
&\Longleftrightarrow
i\equiv^s_\FCC j.
\end{aligned}
\]
Here the fourth equivalence uses the fact that \(\equiv^s_a\) is an
equivalence relation, so two \(\equiv^s_a\)-classes are equal exactly when
their representatives are \(\equiv^s_a\)-equivalent. This proves
\eqref{eq:05-restriction-equivalence-abstract}.

We now prove \eqref{eq:06-local-realization-abstract}. If \(\FCC=\FAG\), then
\(\Fout_\FAG(s,\sigma^{s,i}_\FAG) = o_s(i) = \rho^s_\FAG(i),\)
where the last equality follows from
Lemma~\ref{lem:01-basic-properties-of-presented-powers}.

Suppose now that \(\FCC\subsetneq\FAG\). By
\eqref{eq:03-generated-from-abstract-lower-outcome},
\eqref{eq:02-generated-from-abstract-grand-outcome}, and
\eqref{eq:05-restriction-equivalence-abstract}, we have
\[
\begin{aligned}
\Fout_\FCC(s,\sigma^{s,i}_\FAG|_\FCC)
&=
\bigcup
\{\,\Fout_\FAG(s,\sigma_\FAG)
\mid
\sigma_\FAG\in\FJA_\FAG
\text{ and }
\sigma_\FAG|_\FCC=\sigma^{s,i}_\FAG|_\FCC\,\} \\
&=
\bigcup
\{\,o_s(j)
\mid
j\in I_s
\text{ and }
\sigma^{s,j}_\FAG|_\FCC=\sigma^{s,i}_\FAG|_\FCC\,\} \\
&=
\bigcup
\{\,o_s(j)
\mid
j\in I_s
\text{ and }
i\equiv^s_\FCC j\,\} \\
&=
\rho^s_\FCC(i).
\end{aligned}
\]
This proves \eqref{eq:06-local-realization-abstract}.

\end{proof}

%%%%%%%%%%%%%%%
%%%%%%%%%%%%%%%
\begin{lemma}[Availability in generated game frames]
\label{lem:04-availability-in-generated-abstract-frames}

Let \(\FACNF\) be an \(\PAC\)-abstract-representative actual neighborhood frame,
let \(\mathbb I\) be an abstract actual presentation of \(\FACNF\), and let
\(\FGCGF^{\mathbb I}\) be the generated general concurrent game frame. Then,
for every \(s\in\FST\), every \(\FCC\subseteq\FAG\), and every
\(\sigma_\FCC\in\FJA_\FCC\),
\[
\sigma_\FCC\in\Fav_\FCC(s)
\Longleftrightarrow
\text{there exists }i\in I_s\text{ such that }
\sigma_\FCC=\sigma^{s,i}_\FAG|_\FCC.
\]

\end{lemma}

%%%%%%%%%%%%%%%
%%%%%%%%%%%%%%%
\begin{proof}

Suppose first that \(\sigma_\FCC\in\Fav_\FCC(s)\). By
\eqref{eq:04-generated-from-abstract-availability},
\(\Fout_\FCC(s,\sigma_\FCC)\neq\emptyset.\)
If \(\FCC=\FAG\), then
\eqref{eq:02-generated-from-abstract-grand-outcome} implies that
\(\sigma_\FAG=\sigma^{s,i}_\FAG\) for some \(i\in I_s\).

Now suppose that \(\FCC\subsetneq\FAG\). By
\eqref{eq:03-generated-from-abstract-lower-outcome}, the nonemptiness of
\(\Fout_\FCC(s,\sigma_\FCC)\) implies that there is a full joint action
\(\tau_\FAG\in\FJA_\FAG\) such that
\(\tau_\FAG|_\FCC=\sigma_\FCC \text{ and } \Fout_\FAG(s,\tau_\FAG)\neq\emptyset.\)
By \eqref{eq:02-generated-from-abstract-grand-outcome}, such a \(\tau_\FAG\)
must be of the form \(\sigma^{s,i}_\FAG\) for some \(i\in I_s\). Hence
\(\sigma_\FCC=\sigma^{s,i}_\FAG|_\FCC\).

Conversely, suppose that
\(\sigma_\FCC=\sigma^{s,i}_\FAG|_\FCC\)
for some \(i\in I_s\). By
Lemma~\ref{lem:03-local-components-and-realization-of-presented-powers},
\(\Fout_\FCC(s,\sigma_\FCC) = \Fout_\FCC(s,\sigma^{s,i}_\FAG|_\FCC) = \rho^s_\FCC(i).\)
By coverage of the abstract presentation,
\(\rho^s_\FCC(i)\in\Facnei_\FCC(s)\). Since \(\FACNF\) satisfies liveness,
\(\rho^s_\FCC(i)\neq\emptyset\). Therefore
\(\Fout_\FCC(s,\sigma_\FCC)\neq\emptyset\), and so
\(\sigma_\FCC\in\Fav_\FCC(s)\) by
\eqref{eq:04-generated-from-abstract-availability}.

\end{proof}

%%%%%%%%%%%%%%%
%%%%%%%%%%%%%%%
\begin{lemma}[Exact recovery from abstract presentations]
\label{lem:05-exact-recovery-from-abstract-presentations}

Let \(\FACNF\) be an \(\PAC\)-abstract-representative actual neighborhood frame,
let \(\mathbb I\) be an abstract actual presentation of \(\FACNF\), and let
\(\FGCGF^{\mathbb I}\) be the generated general concurrent game frame. Let
\(\FACEF_\FCC\) be the actual effectivity function of \(\FCC\) in
\(\FGCGF^{\mathbb I}\). Then, for every \(s\in\FST\) and every
\(\FCC\subseteq\FAG\),
\(\FACEF_\FCC(s)=\Facnei_\FCC(s).\)

\end{lemma}

%%%%%%%%%%%%%%%
%%%%%%%%%%%%%%%
\begin{proof}

Fix \(s\in\FST\) and \(\FCC\subseteq\FAG\).

\begin{itemize}
    \item

\medskip
\noindent
\emph{First inclusion.}
Let \(X\in\FACEF_\FCC(s)\). Then there exists
\(\sigma_\FCC\in\Fav_\FCC(s)\) such that
\(X=\Fout_\FCC(s,\sigma_\FCC).\)
By Lemma~\ref{lem:04-availability-in-generated-abstract-frames}, choose
\(i\in I_s\) such that
\(\sigma_\FCC=\sigma^{s,i}_\FAG|_\FCC.\)
By Lemma~\ref{lem:03-local-components-and-realization-of-presented-powers},
\(X = \Fout_\FCC(s,\sigma^{s,i}_\FAG|_\FCC) = \rho^s_\FCC(i).\)
By coverage, \(\rho^s_\FCC(i)\in\Facnei_\FCC(s)\). Hence
\(\FACEF_\FCC(s)\subseteq\Facnei_\FCC(s).\)

\item 
\emph{Second inclusion.}
Let \(X\in\Facnei_\FCC(s)\). By coverage, choose \(i\in I_s\) such that
\(X=\rho^s_\FCC(i).\)
By Lemma~\ref{lem:04-availability-in-generated-abstract-frames},
\(\sigma^{s,i}_\FAG|_\FCC\in\Fav_\FCC(s).\)
By Lemma~\ref{lem:03-local-components-and-realization-of-presented-powers},
\(\Fout_\FCC(s,\sigma^{s,i}_\FAG|_\FCC) = \rho^s_\FCC(i) = X.\)
Thus \(X\in\FACEF_\FCC(s)\), and therefore
\(\Facnei_\FCC(s)\subseteq\FACEF_\FCC(s).\)

\end{itemize}

\end{proof}

%%%%%%%%%%%%%%%
%%%%%%%%%%%%%%%
\begin{lemma}[Preservation of seriality, independence, and determinism]
\label{lem:06-preservation-of-seriality-independence-and-determinism}

Let \(\FACNF\) be an \(\PAC\)-abstract-representative actual neighborhood frame,
let \(\mathbb I\) be an abstract actual presentation of \(\FACNF\), and let
\(\FGCGF^{\mathbb I}\) be the generated general concurrent game frame. Then the
following hold.

\begin{enumerate}

\item If \(\FACNF\) is \(\PAC\)-abstract-serial, then
\(\FGCGF^{\mathbb I}\) is serial.

\item If \(\mathbb I\) is amalgamating, then \(\FGCGF^{\mathbb I}\) is
independent.

\item If \(\FACNF\) is \(\PAC\)-abstract-deterministic, then
\(\FGCGF^{\mathbb I}\) is deterministic.

\end{enumerate}

\end{lemma}

%%%%%%%%%%%%%%%
%%%%%%%%%%%%%%%
\begin{proof}

Let \(\FACEF_\FCC\) be the actual effectivity function of \(\FCC\) in
\(\FGCGF^{\mathbb I}\). By
Lemma~\ref{lem:05-exact-recovery-from-abstract-presentations},
\(\FACEF_\FCC(s)=\Facnei_\FCC(s)\)
for every \(s\in\FST\) and every \(\FCC\subseteq\FAG\).

\begin{itemize}
    \item 

\emph{Seriality.}
Assume that \(\FACNF\) is \(\PAC\)-abstract-serial. Let \(s\in\FST\) and \(\FCC\subseteq\FAG\). Then
\(\Facnei_\FCC(s)\neq\emptyset\). Exact recovery gives
\(\FACEF_\FCC(s)\neq\emptyset\). Hence there is some
\(\sigma_\FCC\in\Fav_\FCC(s)\). Therefore \(\FGCGF^{\mathbb I}\) is serial.

\item 

\emph{Independence.}
Assume that \(\mathbb I\) is amalgamating. Let \(s\in\FST\), let
\(\FCC,\FDD\subseteq\FAG\) be disjoint, and suppose that
\(\sigma_\FCC\in\Fav_\FCC(s) \text{ and } \tau_\FDD\in\Fav_\FDD(s).\)
By Lemma~\ref{lem:04-availability-in-generated-abstract-frames}, choose
\(i,j\in I_s\) such that
\(\sigma_\FCC=\sigma^{s,i}_\FAG|_\FCC \text{ and } \tau_\FDD=\sigma^{s,j}_\FAG|_\FDD.\)
Since \(\mathbb I\) is amalgamating, there exists \(k\in I_s\) such that
\(i\equiv^s_\FCC k \text{ and } j\equiv^s_\FDD k.\)
By Lemma~\ref{lem:03-local-components-and-realization-of-presented-powers},
\(\sigma^{s,k}_\FAG|_\FCC = \sigma^{s,i}_\FAG|_\FCC = \sigma_\FCC\)
and
\(\sigma^{s,k}_\FAG|_\FDD = \sigma^{s,j}_\FAG|_\FDD = \tau_\FDD.\)
Since \(\FCC\) and \(\FDD\) are disjoint,
\(\sigma_\FCC\cup\tau_\FDD = \sigma^{s,k}_\FAG|_{\FCC\cup\FDD}.\)
By Lemma~\ref{lem:04-availability-in-generated-abstract-frames},
\(\sigma_\FCC\cup\tau_\FDD \in \Fav_{\FCC\cup\FDD}(s).\)
Thus \(\FGCGF^{\mathbb I}\) is independent.

\item 

\emph{Determinism.}
Assume that \(\FACNF\) is \(\PAC\)-abstract-deterministic. Let \(s\in\FST\)
and let \(\sigma_\FAG\in\Fav_\FAG(s)\). Then
\(\Fout_\FAG(s,\sigma_\FAG) \in \FACEF_\FAG(s) = \Facnei_\FAG(s).\)
By \(\PAC\)-abstract-determinism, every member of \(\Facnei_\FAG(s)\) is a
singleton. Hence \(\Fout_\FAG(s,\sigma_\FAG)\) is a singleton. Therefore
\(\FGCGF^{\mathbb I}\) is deterministic.

\end{itemize}

\end{proof}

%%%%%%%%%%%%%%%
%%%%%%%%%%%%%%%
\begin{theorem}[Abstract actual enoughness theorem]
\label{thm:06-abstract-actual-enoughness-theorem}

Let \(\FAG\) be finite and nonempty, and let \(\FXX\in\FES\). Let
\[
\FACNF=(\FST,\{\Facnei_\FCC\mid \FCC\subseteq\FAG\})
\]
be an \(\PAC\)-abstract-representative actual neighborhood \(\FXX\)-frame.
Then there exists a general concurrent game \(\FXX\)-frame
\[
\FGCGF=(\FST,\FAC,\{\Fav_\FCC\mid \FCC\subseteq\FAG\},
\{\Fout_\FCC\mid \FCC\subseteq\FAG\})
\]
such that \(\FGCGF\) is \(\PAC\)-representable by \(\FACNF\).

\end{theorem}

%%%%%%%%%%%%%%%
%%%%%%%%%%%%%%%
\begin{proof}
~

\begin{itemize}

\item 

\emph{Choosing a presentation.}
If
\(\FXX\in\{\mathtt{I},\mathtt{SI},\mathtt{ID},\mathtt{SID}\},\)
choose an amalgamating abstract actual presentation \(\mathbb I\) of
\(\FACNF\), which exists because \(\FACNF\) is \(\PAC\)-abstract-independent.
If
\(\FXX\in\{\epsilon,\mathtt{S},\mathtt{D},\mathtt{SD}\},\)
choose any abstract actual presentation \(\mathbb I\) of \(\FACNF\), which
exists by \(\PAC\)-abstract-representativeness.

\item
\emph{Exact representation.}
Let \(\FGCGF^{\mathbb I}\) be the game frame generated from \(\mathbb I\) as in
Definition~\ref{def:19-game-frame-generated-by-an-abstract-actual-presentation}. As
observed after that definition, \(\FGCGF^{\mathbb I}\) is a general concurrent
game frame.

Let \(\FACEF_\FCC\) be the actual effectivity function of \(\FCC\) in
\(\FGCGF^{\mathbb I}\). By
Lemma~\ref{lem:05-exact-recovery-from-abstract-presentations}, for every
\(s\in\FST\) and every \(\FCC\subseteq\FAG\),
\(\FACEF_\FCC(s)=\Facnei_\FCC(s).\)
Thus \(\FGCGF^{\mathbb I}\) is \(\PAC\)-representable by \(\FACNF\).

\item
\emph{Preservation of the properties encoded by \(\FXX\).} This follows from Lemma~\ref{lem:06-preservation-of-seriality-independence-and-determinism}.

\end{itemize}

Taking \(\FGCGF=\FGCGF^{\mathbb I}\) completes
the proof.

\end{proof}

%%%%%%%%%%%%%%%
%%%%%%%%%%%%%%%
\subsubsection*{On finite actions}

The preceding result should not be read as a finite-action representation
result. It makes no claim that the representing action set \(\FAC\) can always
be chosen finite, even when the state set \(\FST\) is finite.

This is already visible in the abstract enoughness construction. There, the
action set is built from the equivalence classes \([i]^s_a\) of indices in an
abstract actual presentation:
\[
\FAC^{\mathbb I}
=
\{*\}
\cup
\{\,(s,a,[i]^s_a)
\mid
s\in\FST,\ a\in\FAG,\ i\in I_s\,\}.
\]
Thus \(\FAC^{\mathbb I}\) is finite only under additional finiteness
assumptions, for example if \(\FST\) is finite and, for each state \(s\) and
agent \(a\), only finitely many equivalence classes \([i]^s_a\) occur.

%%%%%%%%%%%%%%%
%%%%%%%%%%%%%%%
\section{Coherent profile covers}
\label{sec:06-coherent-profile-covers}

The abstract notion of \(\PAC\)-representativeness introduced above is
formulated in terms of abstract actual presentations. Such presentations are
useful for representation: their indices play the role of abstract full-joint
witnesses, and the equivalence relations on indices record agreement on agent
components. However, this indexed description is not an intrinsic condition on
the actual neighborhood frame itself. It introduces auxiliary objects that are
not part of the given family of actual powers.

The aim of this section is to replace these indexed presentations by an
index-free neighborhood-level formulation, expressed through coherent profile
covers. A power profile records, at a fixed state, a simultaneous choice of
actual power for every coalition. A coherent profile cover is then a family of
such profiles that is large enough to cover all actual powers and coherent
enough to recover coalition powers from the grand-coalition components of
compatible profiles.

In the next section, coherent coverability, formulated through coherent profile
covers, will be used to define the finite-agent notion of
\(\PAC\)-representativeness. Thus the present section serves as a bridge from
the abstract, witness-indexed formulation to a neighborhood-intrinsic
representation condition.

The section proceeds as follows. We first define power profiles and profile
agreement. We then introduce the support operator and coherent profile sets.
Next, we define coherent profile covers and coherent coverability. Next, we
introduce the coherent core, which gives a canonical characterization of
coherent coverability. Finally, we show that coherent profile covers and abstract actual presentations can be translated to each other.

%%%%%%%%%%%%%%%
%%%%%%%%%%%%%%%
\subsection{Power profiles and profile agreement}
\label{subsec:06-01-power-profiles-and-profile-agreement}

Power profiles are the basic index-free objects used in this section. At a
fixed state, a power profile selects one actual power for each coalition, in a
way that respects coalition inclusion.

%%%%%%%%%%%%%%%
%%%%%%%%%%%%%%%
\begin{definition}[Power profiles]
\label{def:20-power-profiles}

Let
\[
\FACNF = (\FST, \{\Facnei_\FCC \mid \FCC \subseteq \FAG\})
\]
be an actual neighborhood frame, where \(\FAG\) is finite and nonempty, and fix
a state \(s \in \FST\).

A \Fdefs{power profile} at \(s\) is a family
\[
\gamma = (\gamma_\FCC)_{\FCC \subseteq \FAG}
\]
satisfying the following conditions.

\begin{enumerate}

\item \Fdefs{Membership.} For every coalition \(\FCC \subseteq \FAG\),
\[
\gamma_\FCC \in \Facnei_\FCC(s).
\]

\item \Fdefs{Inclusion monotonicity.} For all coalitions
\(\FCC \subseteq \FDD \subseteq \FAG\),
\[
\gamma_\FDD \subseteq \gamma_\FCC.
\]

\end{enumerate}

Let \(\Gamma_s\) denote the set of all power profiles at \(s\). No existence is
presupposed: depending on the actual neighborhoods at \(s\), the set
\(\Gamma_s\) may be empty.

\end{definition}

Intuitively, a power profile can be understood as an action profile.

In what follows, we use $\gamma \upharpoonright_\FCC$ to denote $(\gamma_\FDD)_{\FDD \subseteq \FCC}$, called the \Fdefs{restriction} of $\gamma$ to $\FCC$.

%%%%%%%%%%%%%%%
%%%%%%%%%%%%%%%
\begin{definition}[Agreement of power profiles below a coalition]
\label{def:21-agreement-of-power-profiles-below-a-coalition}

In the setting of Definition~\ref{def:20-power-profiles}, let
\(\gamma,\delta\in\Gamma_s\) and let \(\FCC\subseteq\FAG\). We write
\[
\gamma \equiv^s_\FCC \delta
\]
and say that \(\gamma\) and \(\delta\) \Fdefs{agree below \(\FCC\)} if
\[
\gamma_\FDD = \delta_\FDD
\text{ for every } \FDD \subseteq \FCC.
\]

\end{definition}

Agreement below \(\FCC\) means agreement on all components indexed by
subcoalitions of \(\FCC\), not merely agreement on the \(\FCC\)-component
itself. This choice mirrors the abstract-index setting: if two abstract full
indices agree on the agents in \(\FCC\), then the powers determined by all
subcoalitions of \(\FCC\) are the same.

We record several basic observations about profile agreement that will be used
below.

\begin{itemize}

\item 

For each coalition \(\FCC\subseteq\FAG\), the relation \(\equiv^s_\FCC\) is an
equivalence relation on \(\Gamma_s\). 

\item 

If
\(\FCC\subseteq\FDD\subseteq\FAG\), then
$
\gamma\equiv^s_\FDD\delta
\Longrightarrow
\gamma\equiv^s_\FCC\delta.
$

\item 
$
\gamma\equiv^s_\FAG\delta
\Longleftrightarrow
\gamma=\delta,
$
while
$
\gamma\equiv^s_\emptyset\delta
\Longleftrightarrow
\gamma_\emptyset=\delta_\emptyset.
$

\end{itemize}

We use the same notation \(\equiv^s_\FCC\) for agreement of power profiles and
for agreement of abstract indices. This is only a notational convention. The
domain will always make clear which relation is meant; when both kinds of
objects occur together, we explicitly distinguish profile agreement from
abstract-index agreement.

%%%%%%%%%%%%%%%
%%%%%%%%%%%%%%%
\subsection{Support and coherence}
\label{subsec:06-02-support-and-coherence}

In this section, we introduce the support operator, which sends each set of power profiles to the set of power profiles supported by it. This operator provides the basis for our definition of coherent profile sets. We then record several elementary facts about the support operator and coherent profile sets that will be used in the sequel.

%%%%%%%%%%%%%%%
%%%%%%%%%%%%%%%
\begin{definition}[Support operator]
\label{def:22-support-operator}

Let
\[
\FACNF=(\FST,\{\Facnei_\FCC\mid \FCC\subseteq\FAG\})
\]
be an actual neighborhood frame, where \(\FAG\) is finite and nonempty. Fix
\(s\in\FST\), and let \(\Gamma_s\) be the set of power profiles at \(s\).

For \(X\subseteq\Gamma_s\), \(\gamma\in\Gamma_s\), and
\(\FCC\subseteq\FAG\), write
\[
\CAG_X^\FCC(\gamma)
=
\{\delta_\FAG\mid
\delta\in X
\text{ and }
\delta\equiv^s_\FCC\gamma\}.
\]

Define the \Fdefs{support operator}
\[
\Phi_s:\mathcal P(\Gamma_s)\to\mathcal P(\Gamma_s)
\]
by
\[
\Phi_s(X)
=
\{\gamma\in\Gamma_s
\mid
\text{for every }\FCC\subseteq\FAG,\ 
\gamma_\FCC=\bigcup \CAG_X^\FCC(\gamma)\}.
\]

\end{definition}

To unpack the definition, fix \(X\), \(\gamma\), and \(\FCC\).

\begin{itemize}

\item

The set
\(\CAG_X^\FCC(\gamma)\) consists of the grand-coalition components
\(\delta_\FAG\) of those profiles \(\delta\in X\) whose restriction below
\(\FCC\) agrees with that of \(\gamma\). Thus
\(\bigcup \CAG_X^\FCC(\gamma)\) is the set of states appearing in the
grand-coalition powers of all such compatible profiles.

\item

The equation
\[
\gamma_\FCC=\bigcup \CAG_X^\FCC(\gamma)
\]
says that the \(\FCC\)-component of \(\gamma\) is exactly recovered from these
compatible grand-coalition powers. When this equation holds, we say that
\(\gamma_\FCC\) is \Fdefs{supported by} \(X\).

\item

\(\gamma\in\Phi_s(X)\) means that, for every \(\FCC\), the component
\(\gamma_\FCC\) is supported by \(X\). In this case, we say that \(\gamma\) is
\Fdefs{supported by \(X\)}.

\end{itemize}

%%%%%%%%%%%%%%%
%%%%%%%%%%%%%%%
\begin{definition}[Coherent profile sets]
\label{def:23-coherent-profile-sets}

Let
\[
\FACNF=(\FST,\{\Facnei_\FCC\mid \FCC\subseteq\FAG\})
\]
be an actual neighborhood frame, where \(\FAG\) is finite and nonempty, and fix
\(s\in\FST\). Let \(\Gamma_s\) be the set of power profiles at \(s\).

A set \(X\subseteq\Gamma_s\) is \Fdefs{coherent} if
\(X\subseteq\Phi_s(X).\)
Equivalently, every profile \(\gamma\) in \(X\) is supported by \(X\): for every
\(\FCC\subseteq\FAG\),
\(\gamma_\FCC = \bigcup \CAG_X^\FCC(\gamma).\)

\end{definition}

Intuitively, a coherent set of power profiles is internally self-supporting.
No profile in such a set carries a coalition component that cannot be recovered
from compatible grand-coalition components of profiles in the same set.

Actually, coherent sets are so-called \emph{post-fixed points} of $\Phi_s$.

The next two lemmas record the structural properties of the support operator.

The first result indicates that under a mild liveness assumption for the grand coalition, coherent
sets are not merely post-fixed points of \(\Phi_s\); they are fixed points.

%%%%%%%%%%%%%%%
%%%%%%%%%%%%%%%
\begin{lemma}[Coherent sets are fixed points under liveness]
\label{lem:07-coherent-sets-are-fixed-points-under-liveness}

Let
\[
\FACNF=(\FST,\{\Facnei_\FCC\mid \FCC\subseteq\FAG\})
\]
be an actual neighborhood frame, where \(\FAG\) is finite and nonempty, and fix
\(s\in\FST\). Assume that the grand-coalition neighborhood is live at \(s\),
that is,
\(\emptyset\notin\Facnei_\FAG(s).\)
If \(X\subseteq\Gamma_s\) is coherent, then
\(X=\Phi_s(X).\)

\end{lemma}

%%%%%%%%%%%%%%%
%%%%%%%%%%%%%%%
\begin{proof}
By coherence,
\(X\subseteq\Phi_s(X).\)
It remains to prove the reverse inclusion. Let \(\gamma\in\Phi_s(X)\). Since
\(\gamma\in\Gamma_s\), we have
\(\gamma_\FAG\in\Facnei_\FAG(s).\)
By liveness,
\(\gamma_\FAG\neq\emptyset.\)

Suppose, for contradiction, that \(\gamma\notin X\). Agreement below the grand
coalition \(\FAG\) is equality of power profiles. Hence no profile in \(X\)
agrees with \(\gamma\) below \(\FAG\), and therefore
\(\CAG_X^\FAG(\gamma)=\emptyset.\)
Since \(\gamma\in\Phi_s(X)\), the support condition for \(\FAG\) gives
\(\gamma_\FAG = \bigcup \CAG_X^\FAG(\gamma) = \emptyset,\)
contradicting \(\gamma_\FAG\neq\emptyset\). Thus \(\gamma\in X\), and so
\(\Phi_s(X)\subseteq X.\)
Therefore
\(X=\Phi_s(X).\)
\end{proof}

The second is monotonicity.

%%%%%%%%%%%%%%%
%%%%%%%%%%%%%%%
\begin{lemma}[Monotonicity of support]
\label{lem:08-monotonicity-of-support}

Let
\[
\FACNF=(\FST,\{\Facnei_\FCC\mid\FCC\subseteq\FAG\})
\]
be an actual neighborhood frame, where \(\FAG\) is finite and nonempty. Fix
\(s\in\FST\), let \(\Gamma_s\) be the set of power profiles at \(s\), and let
\[
\Phi_s:\mathcal P(\Gamma_s)\to\mathcal P(\Gamma_s)
\]
be the support operator at \(s\).

Then \(\Phi_s\) is monotone. That is, for all
\(X,Y\subseteq\Gamma_s\), if \(X\subseteq Y\), then
\(\Phi_s(X)\subseteq\Phi_s(Y).\)

\end{lemma}

%%%%%%%%%%%%%%%
%%%%%%%%%%%%%%%
\begin{proof}
Assume \(X\subseteq Y\subseteq\Gamma_s\), and let
\(\gamma\in\Phi_s(X)\). We show that \(\gamma\in\Phi_s(Y)\).

Fix \(\FCC\subseteq\FAG\). Since \(X\subseteq Y\), we have
\(\CAG_X^\FCC(\gamma)\subseteq \CAG_Y^\FCC(\gamma).\)
Thus, using \(\gamma\in\Phi_s(X)\),
\(\gamma_\FCC = \bigcup \CAG_X^\FCC(\gamma) \subseteq \bigcup \CAG_Y^\FCC(\gamma).\)

For the converse inclusion, let
\(y\in\bigcup \CAG_Y^\FCC(\gamma)\). Then
\(y\in\delta_\FAG\) for some \(\delta\in Y\) with
\(\delta\equiv^s_\FCC\gamma\). By inclusion monotonicity,
\(\delta_\FAG\subseteq\delta_\FCC,\)
and by agreement below \(\FCC\),
\(\delta_\FCC=\gamma_\FCC.\)
Hence \(y\in\gamma_\FCC\). Therefore
\(\bigcup \CAG_Y^\FCC(\gamma)\subseteq\gamma_\FCC.\)

So
\(\gamma_\FCC=\bigcup \CAG_Y^\FCC(\gamma).\)
Since \(\FCC\subseteq\FAG\) was arbitrary, \(\gamma\in\Phi_s(Y)\). Hence
\(\Phi_s(X)\subseteq\Phi_s(Y).\)
\end{proof}

%%%%%%%%%%%%%%%
%%%%%%%%%%%%%%%
\subsection{Coherent profile covers}
\label{subsec:06-03-coherent-profile-covers}

%%%%%%%%%%%%%%%
%%%%%%%%%%%%%%%
\begin{definition}[Coherent profile covers]
\label{def:24-coherent-profile-covers}

Let
\[
\FACNF=(\FST,\{\Facnei_\FCC \mid \FCC \subseteq \FAG\})
\]
be an actual neighborhood frame, where \(\FAG\) is finite and nonempty, and fix
\(s\in\FST\).

A \Fdefs{coherent profile cover} at \(s\) is a coherent set
\(\Omega_s\subseteq\Gamma_s\) such that, for every coalition
\(\FCC\subseteq\FAG\),
\(\Facnei_\FCC(s) = \{\gamma_\FCC\mid \gamma\in\Omega_s\}.\)

\end{definition}

Intuitively, a coherent profile cover \(\Omega_s\) gives a global organization
of the local actual powers at \(s\). Each profile in \(\Omega_s\) records one
coordinated assignment of powers to all coalitions, while the projection of
\(\Omega_s\) to each coalition \(\FCC\) recovers exactly the actual powers
\(\Facnei_\FCC(s)\) available to that coalition at \(s\).

%%%%%%%%%%%%%%%
%%%%%%%%%%%%%%%
\begin{definition}[Coherent coverability]
\label{def:25-coherent-coverability}

Let
\[
\FACNF=(\FST,\{\Facnei_\FCC \mid \FCC \subseteq \FAG\})
\]
be an actual neighborhood frame, where \(\FAG\) is finite and nonempty, and fix
\(s\in\FST\).

We say that \(s\) is \Fdefs{coherently coverable} if it admits a coherent
profile cover
\(\Omega_s\subseteq\Gamma_s\)
at \(s\).

We say that \(\FACNF\) is \Fdefs{coherently coverable} if every state
\(s\in\FST\) is coherently coverable.

\end{definition}

Intuitively, coherent coverability says that the local actual powers at each
state can be organized into a coherent system of power profiles.
Such a system is exhaustive in that, at each coalition coordinate, it recovers
all and only the actual powers available there. It is also disciplined by
coherence: its profiles must mutually support their coalition components
through compatible grand-coalition powers.

%%%%%%%%%%%%%%%
%%%%%%%%%%%%%%%
\begin{example}[Coherent covers]
\label{ex:03-coherent-covers}

Let
\[
\FAG=\{a,b,c\},
\qquad
W=\{x,y\},
\qquad
\FST=\{s\}\cup W.
\]
We define the actual neighborhoods at the distinguished state \(s\) by
\[
\Facnei_\FCC(s)=\{W\}
\qquad
\text{for every } \FCC\subsetneq\FAG,
\]
and
\[
\Facnei_\FAG(s)=\bigl\{\{x\},\{y\}\bigr\}.
\]
To make the actual neighborhood frame fully defined, set, for
every \(t\in W\) and every \(\FCC\subseteq\FAG\),
\[
\Facnei_\FCC(t)=\{\FST\}.
\]

Define two profiles \(\gamma^x\) and \(\gamma^y\) at \(s\) by
\[
(\gamma^x)_\FCC
=
\begin{cases}
W, & \text{if } \FCC\subsetneq\FAG,\\
\{x\}, & \text{if } \FCC=\FAG,
\end{cases}
\qquad
(\gamma^y)_\FCC
=
\begin{cases}
W, & \text{if } \FCC\subsetneq\FAG,\\
\{y\}, & \text{if } \FCC=\FAG.
\end{cases}
\]

Let
\[
\Omega_s=\{\gamma^x,\gamma^y\}.
\]
We check that \(\Omega_s\) is a coherent profile cover at \(s\).

\begin{itemize}

\item \(\gamma^x\) and \(\gamma^y\) are power profiles at \(s\).

Membership is immediate from the definition of the actual neighborhoods at
\(s\). For inclusion monotonicity, let
\(\FCC\subseteq\FDD\subseteq\FAG\). If \(\FDD\subsetneq\FAG\), then both
components are \(W\). If \(\FDD=\FAG\) and \(\FCC\subsetneq\FAG\), then the
grand-coalition component is either \(\{x\}\) or \(\{y\}\), and hence is
contained in \(W\). If \(\FDD=\FAG\) and \(\FCC = \FAG\), then inclusion monotonicity clearly holds in this case. Thus both
profiles satisfy inclusion monotonicity.

\item It is easy to see that \(\Omega_s\) covers all actual powers at \(s\).

\item \(\Omega_s\) is coherent.
We show that both profiles in \(\Omega_s\) are supported by \(\Omega_s\).

We first check that \(\gamma^x\) is supported by \(\Omega_s\). Let \(\FCC\subseteq\FAG\). It suffices to show $\bigcup \CAG_{\Omega_s}^\FCC(\gamma^x)
=
(\gamma^x)_\FCC$.

\begin{itemize}

\item 

Suppose first that \(\FCC\subsetneq\FAG\).
Then every subcoalition of \(\FCC\) is proper, and hence \(\gamma^x\) and
\(\gamma^y\) agree below \(\FCC\). Therefore
\[
\CAG_{\Omega_s}^\FCC(\gamma^x)
=
\bigl\{(\gamma^x)_\FAG,(\gamma^y)_\FAG\bigr\}
=
\bigl\{\{x\},\{y\}\bigr\}.
\]
It follows that
\[
\bigcup \CAG_{\Omega_s}^\FCC(\gamma^x)
=
\{x\}\cup\{y\}
=
W
=
(\gamma^x)_\FCC.
\]

\item 
Now suppose that \(\FCC=\FAG\). Agreement below \(\FAG\) is equality of
profiles, and \(\gamma^y\) does not agree with \(\gamma^x\) at the
grand-coalition component. Hence the only profile in \(\Omega_s\) that agrees
with \(\gamma^x\) below \(\FAG\) is \(\gamma^x\) itself. Thus
\[
\CAG_{\Omega_s}^\FAG(\gamma^x)
=
\bigl\{(\gamma^x)_\FAG\bigr\}
=
\bigl\{\{x\}\bigr\},
\]
and so
\[
\bigcup \CAG_{\Omega_s}^\FAG(\gamma^x)
=
\{x\}
=
(\gamma^x)_\FAG.
\]

\end{itemize}

We next check that \(\gamma^y\) is supported by \(\Omega_s\). Let \(\FCC\subseteq\FAG\). It suffices to show $\bigcup \CAG_{\Omega_s}^\FCC(\gamma^y)
=
(\gamma^y)_\FCC$.

\begin{itemize}

\item 

 Suppose first that \(\FCC\subsetneq\FAG\).
Again, \(\gamma^x\) and \(\gamma^y\) agree below \(\FCC\), because all
components below \(\FCC\) are proper-coalition components. Therefore
\[
\CAG_{\Omega_s}^\FCC(\gamma^y)
=
\bigl\{(\gamma^x)_\FAG,(\gamma^y)_\FAG\bigr\}
=
\bigl\{\{x\},\{y\}\bigr\}.
\]
Hence
\[
\bigcup \CAG_{\Omega_s}^\FCC(\gamma^y)
=
\{x\}\cup\{y\}
=
W
=
(\gamma^y)_\FCC.
\]

\item
Now suppose that \(\FCC=\FAG\). Agreement below \(\FAG\) is equality of
profiles, and \(\gamma^x\) does not agree with \(\gamma^y\) at the
grand-coalition component. Hence the only profile in \(\Omega_s\) that agrees
with \(\gamma^y\) below \(\FAG\) is \(\gamma^y\) itself. Thus
\[
\CAG_{\Omega_s}^\FAG(\gamma^y)
=
\bigl\{(\gamma^y)_\FAG\bigr\}
=
\bigl\{\{y\}\bigr\},
\]
and so
\[
\bigcup \CAG_{\Omega_s}^\FAG(\gamma^y)
=
\{y\}
=
(\gamma^y)_\FAG.
\]

\end{itemize}

Therefore both profiles in \(\Omega_s\) are supported by \(\Omega_s\). Hence
\[
\Omega_s\subseteq\Phi_s(\Omega_s),
\]
so \(\Omega_s\) is coherent.

\end{itemize}

\end{example}

%%%%%%%%%%%%%%%
%%%%%%%%%%%%%%%
\subsection{The coherent core}
\label{subsec:06-04-the-coherent-core}

Having defined coherent coverability, we now turn from existence to
canonicity. By definition, a state is coherently coverable when there exists
some coherent profile cover at that state. This formulation is inherently
existential: it quantifies over possible coherent systems of power profiles,
but does not yet identify a canonical one.

The aim of this subsection is to remove this existential quantification. Using
the support operator, we define a distinguished set of profiles,
\(\operatorname{Core}_s\), namely the greatest coherent set of power profiles
at \(s\). The main result then shows that coherent coverability can be tested
entirely on this canonical object: a coherent profile cover exists at \(s\) if
and only if \(\operatorname{Core}_s\) itself is a coherent profile cover at
\(s\).

%%%%%%%%%%%%%%%
%%%%%%%%%%%%%%%
\begin{definition}[Coherent core]
\label{def:26-coherent-core}

Let
\[
\FACNF=(\FST,\{\Facnei_\FCC\mid \FCC\subseteq\FAG\})
\]
be an actual neighborhood frame, where \(\FAG\) is finite and nonempty, and fix
\(s\in\FST\). By
Lemma~\ref{lem:08-monotonicity-of-support}, the support operator
\[
\Phi_s:\mathcal P(\Gamma_s)\to\mathcal P(\Gamma_s)
\]
is monotone. By the
Knaster--Tarski theorem \cite{bradfield_modal_2007}, \(\Phi_s\) has a greatest post-fixed point, that is also a greatest fixed point. We denote this fixed point by
\[
\operatorname{Core}_s.
\]

Equivalently,
\[
\operatorname{Core}_s
=
\bigcup
\{X\subseteq\Gamma_s\mid X\subseteq\Phi_s(X)\}.
\]

Note
\[
\operatorname{Core}_s
=
\Phi_s(\operatorname{Core}_s),
\]
and, for every \(X\subseteq\Gamma_s\),
\[
X\subseteq\Phi_s(X)
\Longrightarrow
X\subseteq\operatorname{Core}_s.
\]

Since coherent sets of power profiles are precisely the post-fixed points of
\(\Phi_s\), the set \(\operatorname{Core}_s\) is the greatest coherent set of
power profiles at \(s\). We call it the \Fdefs{coherent core} at \(s\).

\end{definition}

Note that the coherent core is always defined, whether or not \(s\) is coherently
coverable.

By the following result, coherent coverability is
not merely an existential condition: whenever any coherent profile cover exists,
the greatest coherent set of profiles is already a coherent cover.

%%%%%%%%%%%%%%%
%%%%%%%%%%%%%%%
\begin{proposition}[Coherent coverability via the coherent core]
\label{prop:02-coherent-coverability-via-the-coherent-core}

Let
\[
\FACNF=(\FST,\{\Facnei_\FCC \mid \FCC \subseteq \FAG\})
\]
be an actual neighborhood frame, where \(\FAG\) is finite and nonempty, and fix
\(s\in\FST\). The following are equivalent:

\begin{enumerate}

\item There is a coherent profile cover
at \(s\).

\item The coherent core \(\operatorname{Core}_s\) is a coherent profile cover
at \(s\).

\end{enumerate}

\end{proposition}

%%%%%%%%%%%%%%%
%%%%%%%%%%%%%%%
\begin{proof}

We prove the two implications.

\noindent
\((1 \Rightarrow 2)\)
Assume that there is a coherent profile cover
\(\Omega_s\subseteq\Gamma_s\)
at \(s\). Since \(\Omega_s\) is coherent, and
\(\operatorname{Core}_s\) is the greatest coherent set of power profiles at
\(s\), we have
\(\Omega_s\subseteq\operatorname{Core}_s.\)

By construction, \(\operatorname{Core}_s\) is coherent. Thus, to show that it
is a coherent profile cover, it remains only to verify the covering condition.
Fix a coalition \(\FCC\subseteq\FAG\). Since \(\Omega_s\) is a coherent profile
cover at \(s\),
\[
\Facnei_\FCC(s)
=
\{\gamma_\FCC\mid \gamma\in\Omega_s\}.
\]
Together with \(\Omega_s\subseteq\operatorname{Core}_s\), this gives
\(\Facnei_\FCC(s) \subseteq \{\gamma_\FCC\mid \gamma\in\operatorname{Core}_s\}.\)

Conversely, every profile in \(\operatorname{Core}_s\) belongs to
\(\Gamma_s\). Hence, for every \(\gamma\in\operatorname{Core}_s\), the
membership condition for power profiles gives
\(\gamma_\FCC\in\Facnei_\FCC(s).\)
Therefore
\(\{\gamma_\FCC\mid \gamma\in\operatorname{Core}_s\} \subseteq \Facnei_\FCC(s).\)
Combining the two inclusions, we obtain
\(\Facnei_\FCC(s) = \{\gamma_\FCC\mid \gamma\in\operatorname{Core}_s\}.\)
Since \(\FCC\subseteq\FAG\) was arbitrary,
\(\operatorname{Core}_s\) covers all actual coalition powers at \(s\). Together
with coherence, this shows that \(\operatorname{Core}_s\) is a coherent
profile cover at \(s\).

\medskip

\noindent
\((2 \Rightarrow 1)\)
If \(\operatorname{Core}_s\) is a coherent profile cover at \(s\), then there
exists a coherent profile cover at \(s\), namely
\[
\operatorname{Core}_s
\]
itself.

\end{proof}

\subsection{Translations between coherent profile covers and abstract presentations}
\label{subsec:06-05-translations-between-coherent-profile-covers-and-abstract-presentations}

We now connect the index-free formulation developed in this section with the
indexed presentations introduced in
Section~\ref{sec:05-representing-actual-powers-via-abstract-presentations}.
We show that for every actual neighborhood frame $\FACNF$ and every state $s$ of it, (1) from an abstract actual presentation at \(s\), we can get a coherent profile cover at $s$, and (2) under actual triviality of the
empty coalition, from a nonempty coherent profile cover at $s$, we can get an abstract actual presentation at $s$.

%%%%%%%%%%%%%%%
%%%%%%%%%%%%%%%
\subsubsection{From abstract actual presentations to coherent profile covers}
\label{subsubsec:06-05-01-from-abstract-actual-presentations-to-coherent-profile-covers}

%%%%%%%%%%%%%%%
%%%%%%%%%%%%%%%
\begin{definition}[Profile sets induced by abstract actual presentations]
\label{def:27-profile-set-induced-by-an-abstract-actual-presentation}

Let
\[
\FACNF=(\FST,\{\Facnei_\FCC\mid \FCC\subseteq\FAG\})
\]
be an actual neighborhood frame, where \(\FAG\) is finite and nonempty. Fix
\(s\in\FST\), and let
\[
\mathbb I_s=(I_s,o_s,\{\equiv^s_a\}_{a\in\FAG})
\]
be an abstract actual presentation of \(\FACNF\) at \(s\).

For each \(i\in I_s\), define
\(\gamma^i=(\gamma^i_\FCC)_{\FCC\subseteq\FAG}\)
by
\(\gamma^i_\FCC=\rho^s_\FCC(i)\)
for every \(\FCC\subseteq\FAG\), where \(\rho^s_\FCC\) is the presented-power
map associated with \(\mathbb I_s\).

The \Fdefs{profile set induced by \(\mathbb I_s\)} is
\(\Omega_s = \{\,\gamma^i\mid i\in I_s\,\}.\)

\end{definition}

%%%%%%%%%%%%%%%
%%%%%%%%%%%%%%%
\begin{lemma}[Abstract presentations induce coherent profile covers]
\label{lem:09-abstract-presentations-induce-coherent-profile-covers}

Let
\[
\FACNF=(\FST,\{\Facnei_\FCC\mid \FCC\subseteq\FAG\})
\]
be an actual neighborhood frame, where \(\FAG\) is finite and nonempty. Fix
\(s\in\FST\), and let \(\mathbb I_s\) be an abstract actual presentation of
\(\FACNF\) at \(s\).

Then the induced profile set
$
\Omega_s
$
is a coherent profile cover at \(s\).

\end{lemma}

%%%%%%%%%%%%%%%
%%%%%%%%%%%%%%%
\begin{proof}

Let
\(\Omega_s=\{\,\gamma^i\mid i\in I_s\,\}.\)

\medskip
\noindent
\textbf{Step 1: The induced profiles are power profiles.}
Let \(i\in I_s\). We show that
\(\gamma^i=(\gamma^i_\FCC)_{\FCC\subseteq\FAG}\)
is a power profile at \(s\).

First, fix a coalition \(\FCC\subseteq\FAG\). By the coverage clause of the
abstract actual presentation,
\(\Facnei_\FCC(s) = \{\,\rho^s_\FCC(j)\mid j\in I_s\,\}.\)
Since \(i\in I_s\), it follows that
\(\gamma^i_\FCC = \rho^s_\FCC(i) \in \Facnei_\FCC(s).\)
Thus the membership condition for power profiles holds.

Second, let
\(\FCC\subseteq\FDD\subseteq\FAG.\)
By Lemma~\ref{lem:01-basic-properties-of-presented-powers},
\(\rho^s_\FDD(i)\subseteq\rho^s_\FCC(i).\)
Hence
\(\gamma^i_\FDD = \rho^s_\FDD(i) \subseteq \rho^s_\FCC(i) = \gamma^i_\FCC.\)
Thus inclusion monotonicity holds.

\medskip
\noindent
\textbf{Step 2: The induced profile set covers the actual neighborhoods.}
Let \(\FCC\subseteq\FAG\). By the definition of \(\Omega_s\), the definition of
\(\gamma^i_\FCC\), and the coverage clause of the abstract actual presentation,
\[
\{\,\gamma_\FCC\mid \gamma\in\Omega_s\,\}
=
\{\,\gamma^i_\FCC\mid i\in I_s\,\}
=
\{\,\rho^s_\FCC(i)\mid i\in I_s\,\}
=
\Facnei_\FCC(s).
\]
Thus \(\Omega_s\) satisfies the covering condition for every coalition
\(\FCC\subseteq\FAG\).

\medskip
\noindent
\textbf{Step 3: The induced profile set is coherent.}
It remains to show that every profile in \(\Omega_s\) is supported by
\(\Omega_s\). Let \(i\in I_s\) and let \(\FCC\subseteq\FAG\). We prove
\[
\gamma^i_\FCC
=
\bigcup
\{\,\delta_\FAG
\mid
\delta\in\Omega_s
\text{ and }
\delta\equiv^s_\FCC\gamma^i\,\},
\tag{APC}
\label{eq:07-abstract-to-coherent-support}
\]
where \(\equiv^s_\FCC\) on the right-hand side is agreement of power profiles
below \(\FCC\).

First, we prove the inclusion from left to right. Let
\(x\in\gamma^i_\FCC.\)
By the definition of \(\gamma^i_\FCC\), we have
\(x\in\rho^s_\FCC(i).\)
By the definition of \(\rho^s_\FCC(i)\), there exists \(j\in I_s\) such that
\(i\equiv^s_\FCC j \text{ and } x\in o_s(j),\)
where \(i\equiv^s_\FCC j\) is the coalition-equivalence relation in the
abstract actual presentation.

We claim that \(\gamma^j\equiv^s_\FCC\gamma^i\) as power profiles. Let
\(\FDD\subseteq\FCC\). Since \(i\equiv^s_\FCC j\), we have
\(i\equiv^s_\FDD j\). Hence, by
Lemma~\ref{lem:01-basic-properties-of-presented-powers},
\(\rho^s_\FDD(i)=\rho^s_\FDD(j).\)
Therefore
\(\gamma^i_\FDD = \rho^s_\FDD(i) = \rho^s_\FDD(j) = \gamma^j_\FDD.\)
Since this holds for every \(\FDD\subseteq\FCC\), we obtain
\(\gamma^j\equiv^s_\FCC\gamma^i.\)

Moreover, by Lemma~\ref{lem:01-basic-properties-of-presented-powers},
\(\gamma^j_\FAG = \rho^s_\FAG(j) = o_s(j).\)
Since \(x\in o_s(j)\), it follows that
\(x\in\gamma^j_\FAG.\)
Also \(\gamma^j\in\Omega_s\). Thus \(x\) belongs to the right-hand side of
\eqref{eq:07-abstract-to-coherent-support}. Hence
\(\gamma^i_\FCC \subseteq \bigcup \{\,\delta_\FAG \mid \delta\in\Omega_s \text{ and } \delta\equiv^s_\FCC\gamma^i\,\}.\)

Conversely, let
\(x\in \bigcup \{\,\delta_\FAG \mid \delta\in\Omega_s \text{ and } \delta\equiv^s_\FCC\gamma^i\,\}.\)
Then there exists \(\delta\in\Omega_s\) such that
\(\delta\equiv^s_\FCC\gamma^i \text{ and } x\in\delta_\FAG.\)
Since \(\delta\in\Omega_s\subseteq\Gamma_s\), the profile \(\delta\) is a power
profile. Applying inclusion monotonicity to
\(\FCC\subseteq\FAG\), we get
\(\delta_\FAG\subseteq\delta_\FCC.\)
Since \(\delta\equiv^s_\FCC\gamma^i\), agreement below \(\FCC\) gives
\(\delta_\FCC=\gamma^i_\FCC.\)
Therefore
\(x\in\delta_\FAG\subseteq\delta_\FCC=\gamma^i_\FCC.\)
This proves the reverse inclusion in
\eqref{eq:07-abstract-to-coherent-support}.

Thus \eqref{eq:07-abstract-to-coherent-support} holds for every \(i\in I_s\) and
every \(\FCC\subseteq\FAG\). Hence every profile in \(\Omega_s\) is supported by
\(\Omega_s\), so
\(\Omega_s\subseteq\Phi_s(\Omega_s).\)
Therefore \(\Omega_s\) is coherent.

Combining coherence with the covering condition proved in Step~2, we conclude
that \(\Omega_s\) is a coherent profile cover at \(s\).

\end{proof}

%%%%%%%%%%%%%%%
%%%%%%%%%%%%%%%
\subsubsection{From coherent profile covers to abstract actual presentations}
\label{subsubsec:06-05-02-from-coherent-profile-covers-to-abstract-actual-presentations}

The translation from coherent profile covers to abstract actual presentations
is less immediate. We therefore introduce an intermediate construction, the
unfolding forest of a coherent profile cover, which turns coalition-level
profile agreement into an agent-indexed abstract index structure.

Let
\[
\FACNF=(\FST,\{\Facnei_\FCC\mid \FCC\subseteq\FAG\})
\]
be an actual neighborhood frame, where \(\FAG\) is finite and nonempty. Fix
\(s\in\FST\).
Let \(\Omega_s\subseteq\Gamma_s\) be a nonempty coherent profile cover at
\(s\).

For \(\gamma\in\Omega_s\) and \(\FCC\subseteq\FAG\), put
\[
S^s_\FCC(\gamma)
=
\{\delta\in\Omega_s\mid \delta\equiv^s_\FCC\gamma\},
\]
where \(\equiv^s_\FCC\) is power-profile agreement below \(\FCC\).

The forest constructed below unravels the coalition-indexed profile-agreement structure
\[
(\Omega_s,\{S^s_\FCC\}_{\emptyset\neq\FCC\subseteq\FAG}),
\]
where \(S^s_\FCC(\gamma)\) consists of the profiles in \(\Omega_s\) that agree
with \(\gamma\) below \(\FCC\).

%%%%%%%%%%%%%%%
%%%%%%%%%%%%%%%
\begin{definition}[Unfolding forest of a coherent profile cover]
\label{def:28-unfolding-labelled-rooted-forest-of-a-coherent-profile-cover}

Let
\[
\FACNF=(\FST,\{\Facnei_\FCC\mid \FCC\subseteq\FAG\})
\]
be an actual neighborhood frame, where \(\FAG\) is finite and nonempty. Fix
\(s\in\FST\).

Let \(\Omega_s\subseteq\Gamma_s\) be a nonempty coherent profile cover at
\(s\).

We construct a \Fdefs{labelled rooted forest}
\[
\mathcal{F}_s = (I_s, E_s, \pi_s, \lambda_s)
\]
as follows.

\begin{itemize}

\item \textbf{Nodes and roots.}

Let \(I_s\) be the set of all finite sequences
\[
p=
\langle
\gamma_0,
(\FCC_1,\gamma_1),
\dots,
(\FCC_n,\gamma_n)
\rangle
\]
such that \(\gamma_0\in\Omega_s\), each
\(\emptyset\neq\FCC_m\subseteq\FAG\), and
\[
\gamma_m\in S^s_{\FCC_m}(\gamma_{m-1})
\qquad(1\leq m\leq n).
\]
The roots are the one-element sequences \(\langle\gamma\rangle\), where
\(\gamma\in\Omega_s\).

\item \textbf{Edges.}

For \(p,q\in I_s\), there is a directed edge \(p \to q\) iff \(q\) is obtained
from \(p\) by appending a pair \((\FEE,\delta)\), where
\(\emptyset\neq\FEE\subseteq\FAG\) and
\[
\delta \in S^s_\FEE(\gamma_n)
\text{for } 
p=\langle \gamma_0,(\FCC_1,\gamma_1),\dots,(\FCC_n,\gamma_n)\rangle.
\]
We denote the resulting edge set by \(E_s\).

\item \textbf{Node labeling.}

Define the node labeling function \(\pi_s : I_s \to \Omega_s\) by
\[
\pi_s(p)=\gamma_n
\]
for
\[
p=
\langle
\gamma_0,
(\FCC_1,\gamma_1),
\dots,
(\FCC_n,\gamma_n)
\rangle.
\]

\item \textbf{Edge labeling.}

Define the edge labeling function \(\lambda_s : E_s \to \mathcal{P}(\FAG)\setminus\{\emptyset\}\) by:
if \(p \to q\) is obtained by appending \((\FEE,\delta)\), then
\[
\lambda_s(p,q)=\FEE.
\]

\end{itemize}

\end{definition}

Note the following points:

\begin{itemize}

\item 

By construction, whenever two adjacent nodes \(p,q\) are connected by an edge
labelled \(\FEE\), their profiles agree below \(\FEE\):
\[
\pi_s(p)\equiv^s_\FEE \pi_s(q).
\tag{E}
\label{eq:08-edge-agreement-abstract-unfolding}
\]

\item 

For a root \(\langle\gamma\rangle\), we have \(\pi_s(\langle\gamma\rangle)=\gamma\).
Hence \(\pi_s\) is surjective onto \(\Omega_s\).

\end{itemize}

%%%%%%%%%%%%%%%
%%%%%%%%%%%%%%%
\begin{definition}[Index structure of a coherent profile cover]
\label{def:29-index-structure-of-a-coherent-profile-cover}

Let
\[
\FACNF=(\FST,\{\Facnei_\FCC\mid \FCC\subseteq\FAG\})
\]
be an actual neighborhood frame, where \(\FAG\) is finite and nonempty. Fix
\(s\in\FST\).

Let \(\Omega_s\subseteq\Gamma_s\) be a nonempty coherent profile cover at
\(s\).

Let
\[
\mathcal{F}_s = (I_s, E_s, \pi_s, \lambda_s)
\]
be the labelled rooted forest of \(\Omega_s\).

We define the structure
\[
\mathbb I_s=(I_s,o_s,\{\approx^s_a\}_{a\in\FAG}),
\]
called the \Fdefs{index structure induced by \(\Omega_s\)}, as follows.

\begin{itemize}

\item \textbf{Outcome map.}

Define the outcome function \(o_s : I_s \to \mathcal{P}(\FST)\) by
\[
o_s(p)=\pi_s(p)_\FAG,
\]
where \(\pi_s : I_s \to \Omega_s\) is the node labeling map of
\(\mathcal{F}_s\), and \(\pi_s(p)_\FAG\) denotes the \(\FAG\)-component of the profile \(\pi_s(p)\).

\item \textbf{Agent equivalence relations.}

For each agent \(a\in\FAG\), define a relation \(\approx^s_a\) on \(I_s\) by
\[
p \approx^s_a q
\]
iff \(p\) and \(q\) lie in the same tree of \(\mathcal{F}_s\), and the unique
path connecting \(p\) and \(q\) uses only edges whose labels contain \(a\).

\end{itemize}

\end{definition}

For \(\FCC\subseteq\FAG\), define
\[
p\approx^s_\FCC q
\text{iff}
p\approx^s_a q \text{ for every } a\in\FCC.
\]

If \(\FCC\neq\emptyset\), this is equivalent to saying that \(p\) and \(q\)
lie in the same tree and the unique path between them uses only edges whose
labels contain \(\FCC\). When \(\FCC=\emptyset\), the relation
\(\approx^s_\emptyset\) is universal on \(I_s\).

%%%%%%%%%%%%%%%
%%%%%%%%%%%%%%%
\begin{lemma}[Coherent profile covers unfold into abstract presentations]
\label{lem:10-coherent-profile-covers-unfold-into-abstract-presentations}

Let
\[
\FACNF=(\FST,\{\Facnei_\FCC\mid \FCC\subseteq\FAG\})
\]
be an actual neighborhood frame, where \(\FAG\) is finite and nonempty. Fix
\(s\in\FST\). Assume that actual triviality of the empty coalition holds at
\(s\), that is, if \(\Facnei_\emptyset(s)\) is nonempty, then it is a
singleton.

Let \(\Omega_s\subseteq\Gamma_s\) be a nonempty coherent profile cover at
\(s\). Then the index structure
\[
\mathbb I_s=(I_s,o_s,\{\approx^s_a\}_{a\in\FAG})
\]
induced by \(\Omega_s\) is an abstract actual presentation of \(\FACNF\) at
\(s\).

\end{lemma}

%%%%%%%%%%%%%%%
%%%%%%%%%%%%%%%
\begin{proof}

Let
\[
\mathcal{F}_s=(I_s,E_s,\pi_s,\lambda_s)
\]
be the unfolding labelled rooted forest of \(\Omega_s\). Recall that
\(\pi_s:I_s\to\Omega_s\) is the node labeling function and that
\[
o_s(p)=\pi_s(p)_\FAG
\]
for every \(p\in I_s\).

By construction, every node label belongs to \(\Omega_s\), and every profile
in \(\Omega_s\) appears as the label of a root. Hence
\[
\pi_s[I_s]=\Omega_s.
\tag{R}
\label{eq:09-surjectivity-of-pi-unfolding}
\]

Since \(\Omega_s\) is coherent, for every \(\gamma\in\Omega_s\) and every
coalition \(\FCC\subseteq\FAG\),
\[
\gamma_\FCC
=
\bigcup
\{\delta_\FAG
\mid
\delta\in S^s_\FCC(\gamma)\}.
\tag{SC}
\label{eq:10-coherent-support-for-unfolding}
\]

We verify the two requirements in the definition of an abstract actual
presentation.

\medskip
\noindent
\textbf{Step 1: The agent relations are equivalence relations.}

For each \(a\in\FAG\), the relation \(\approx^s_a\) is an equivalence relation
on \(I_s\). Reflexivity and symmetry are immediate from the definition using
the unique path in the underlying undirected forest. For transitivity, suppose
\(p\approx^s_a q\) and \(q\approx^s_a r\). Then the path from \(p\) to \(q\)
and the path from \(q\) to \(r\) use only edges whose labels contain \(a\).
The unique reduced path from \(p\) to \(r\) is contained in the union of these
two paths, and therefore also uses only edges whose labels contain \(a\). Thus
\(p\approx^s_a r\).

Hence \(\mathbb I_s\) has the required underlying form.

\medskip
\noindent
\textbf{Step 2: Local recovery of profile components.}

Let \(\rho^s_\FCC\) be the presented-power map computed from
\(\mathbb I_s\), that is,
\[
\rho^s_\FCC(p)
=
\bigcup
\{\,o_s(q)
\mid
q\in I_s \text{ and } p\approx^s_\FCC q\,\}.
\]
We first prove that, for every \(p\in I_s\) and every
\(\FCC\subseteq\FAG\),
\[
\rho^s_\FCC(p)=\pi_s(p)_\FCC.
\tag{LSA}
\label{eq:11-local-support-abstract-unfolding}
\]

First consider the case \(\FCC=\emptyset\). Since \(\Omega_s\) is nonempty and
covers the empty-coalition neighborhood,
\[
\Facnei_\emptyset(s)
=
\{\gamma_\emptyset\mid \gamma\in\Omega_s\}
\]
is nonempty. By actual triviality of the empty coalition at \(s\), this set is
a singleton. Hence all profiles in \(\Omega_s\) have the same
\(\emptyset\)-component, and therefore
\[
S^s_\emptyset(\pi_s(p))=\Omega_s.
\]
Using coherence, we obtain
\[
\pi_s(p)_\emptyset
=
\bigcup
\{\delta_\FAG\mid \delta\in\Omega_s\}.
\tag{1}
\]
On the other hand, \(\approx^s_\emptyset\) is the universal relation on
\(I_s\). Therefore, by the definition of \(o_s\) and the surjectivity of
\(\pi_s\),
\[
\rho^s_\emptyset(p)
=
\bigcup
\{\,o_s(q)\mid q\in I_s\,\}
=
\bigcup
\{\,\pi_s(q)_\FAG\mid q\in I_s\,\}
=
\bigcup
\{\delta_\FAG\mid \delta\in\Omega_s\}.
\tag{2}
\]
By (1) and (2),
\[
\rho^s_\emptyset(p)=\pi_s(p)_\emptyset.
\]

Now let \(\FCC\neq\emptyset\). We prove both inclusions.

First, suppose
\(x\in\rho^s_\FCC(p).\)
Then there is \(q\in I_s\) such that
\(p\approx^s_\FCC q \text{ and } x\in o_s(q)=\pi_s(q)_\FAG.\)
Since \(p\approx^s_\FCC q\), the unique path between \(p\) and \(q\) uses only
edges whose labels contain \(\FCC\). Along each edge on this path, the two
adjacent node labels agree below the edge label. Hence adjacent labels agree
below \(\FCC\). By transitivity of power-profile agreement below \(\FCC\),
\(\pi_s(q)\equiv^s_\FCC \pi_s(p).\)
Thus
\(\pi_s(q)_\FCC=\pi_s(p)_\FCC.\)
Since \(\pi_s(q)\in\Omega_s\subseteq\Gamma_s\), it is a power profile, and
inclusion monotonicity gives
\(\pi_s(q)_\FAG\subseteq\pi_s(q)_\FCC.\)
Hence
\(x\in \pi_s(q)_\FAG \subseteq \pi_s(q)_\FCC = \pi_s(p)_\FCC.\)
Therefore
\(\rho^s_\FCC(p)\subseteq\pi_s(p)_\FCC.\)

Conversely, suppose
\(x\in\pi_s(p)_\FCC.\)
By coherence, using \eqref{eq:10-coherent-support-for-unfolding} with
\(\gamma=\pi_s(p)\), there exists
\(\delta\in S^s_\FCC(\pi_s(p))\)
such that
\(x\in\delta_\FAG.\)
Since \(\FCC\neq\emptyset\), the sequence
\(q=p^\frown\langle(\FCC,\delta)\rangle\)
is a node of \(I_s\). The edge from \(p\) to \(q\) is labelled by \(\FCC\), so
\(p\approx^s_\FCC q.\)
Moreover,
\(\pi_s(q)=\delta.\)
Therefore
\(x\in\delta_\FAG=\pi_s(q)_\FAG=o_s(q),\)
and hence
\(x\in\rho^s_\FCC(p).\)
This proves
\(\pi_s(p)_\FCC\subseteq\rho^s_\FCC(p).\)

Thus \eqref{eq:11-local-support-abstract-unfolding} holds for every
\(p\in I_s\) and every coalition \(\FCC\subseteq\FAG\).

\medskip
\noindent
\textbf{Step 3: Grand compatibility.}

Suppose
\(p\approx^s_\FAG q.\)
Then the unique path between \(p\) and \(q\) uses only edges whose labels
contain \(\FAG\). Hence every edge on this path is labelled by \(\FAG\).
Along each such edge, adjacent node labels agree below \(\FAG\). Since
agreement below \(\FAG\) is equality of power profiles, following the path
gives
\(\pi_s(p)=\pi_s(q).\)
Therefore
\(o_s(p)=\pi_s(p)_\FAG=\pi_s(q)_\FAG=o_s(q).\)
Thus grand compatibility holds.

\medskip
\noindent
\textbf{Step 4: Coverage.}

Let \(\FCC\subseteq\FAG\). By \eqref{eq:11-local-support-abstract-unfolding},
\(\{\rho^s_\FCC(p)\mid p\in I_s\} = \{\pi_s(p)_\FCC\mid p\in I_s\}.\)
By \eqref{eq:09-surjectivity-of-pi-unfolding},
\[
\{\pi_s(p)_\FCC\mid p\in I_s\}
=
\{\gamma_\FCC\mid \gamma\in\Omega_s\}.
\]
Since \(\Omega_s\) is a coherent profile cover at \(s\),
\[
\{\gamma_\FCC\mid \gamma\in\Omega_s\}
=
\Facnei_\FCC(s).
\]
Therefore
\[
\Facnei_\FCC(s)
=
\{\rho^s_\FCC(p)\mid p\in I_s\}.
\]
This is exactly the coverage condition for coalition \(\FCC\).

Since \(\FCC\subseteq\FAG\) was arbitrary, coverage holds for every coalition.
Together with Step~3, this shows that \(\mathbb I_s\) is an abstract actual
presentation of \(\FACNF\) at \(s\).

\end{proof}

%%%%%%%%%%%%%%%
%%%%%%%%%%%%%%%
\section{Representing actual powers without assuming independence via coherent profile covers}
\label{sec:07-representing-actual-powers-without-assuming-independence-via-coherent-profile-covers}

In this section, we first define \(\PAC\)-representativeness for actual
neighborhood frames in terms of coherent profile covers, together with the
relevant seriality and determinism conditions. We then show that
\(\PAC\)-representativeness is equivalent to
\(\PAC\)-abstract-representativeness. The corresponding actual representation
theorems then follow from the abstract havingness and enoughness theorems.

%%%%%%%%%%%%%%%
%%%%%%%%%%%%%%%
\subsection{\texorpdfstring{$\PAC$}{PAC}-representativeness via coherent profile covers}
\label{subsec:07-01-finite-agent-pac-representativeness-via-coherent-profile-covers}

%%%%%%%%%%%%%%%
%%%%%%%%%%%%%%%
\begin{definition}[\(\PAC\)-representative actual neighborhood frames]
\label{def:30-finite-agent-pac-representative-actual-neighborhood-frames}

Let
\[
\FACNF = (\FST, \{\Facnei_\FCC \mid \FCC \subseteq \FAG\})
\]
be an actual neighborhood frame, where \(\FAG\) is finite and nonempty. We say
that \(\FACNF\) is \Fdefs{finite-agent \(\PAC\)-representative}, or simply
\Fdefs{\(\PAC\)-representative} in the finite-agent setting, if it satisfies
the following conditions.

\begin{enumerate}

\item \Fdefs{Actual triviality of the empty coalition.}
For every \(s \in \FST\), if \(\Facnei_\emptyset(s)\) is nonempty, then it is a
singleton.

\item \Fdefs{Liveness.}
For every \(s \in \FST\) and every \(\FCC \subseteq \FAG\),
\[
\emptyset \notin \Facnei_\FCC(s).
\]

\item \Fdefs{Coherent coverability.}
For every \(s \in \FST\), there exists a (possibly empty) coherent profile cover
\[
\Omega_s \subseteq \Gamma_s
\]
at \(s\).

Equivalently, by
Proposition~\ref{prop:02-coherent-coverability-via-the-coherent-core},
the coherent core \(\operatorname{Core}_s\) at \(s\) is a coherent profile
cover.

\end{enumerate}

\end{definition}

The liveness condition in the preceding definition can be reduced to a
condition on the grand-coalition neighborhood.

%%%%%%%%%%%%%%%
%%%%%%%%%%%%%%%
\begin{lemma}[Liveness reduces to the grand-coalition neighborhood]
\label{lem:liveness-reduces-to-the-grand-coalition-neighborhood}

Let
\[
\FACNF=(\FST,\{\Facnei_\FCC\mid\FCC\subseteq\FAG\})
\]
be a coherently coverable actual neighborhood frame, where \(\FAG\) is finite
and nonempty. Then the following are equivalent:

\begin{enumerate}

\item For every \(s\in\FST\) and every \(\FCC\subseteq\FAG\),
\[
\emptyset\notin\Facnei_\FCC(s).
\]

\item For every \(s\in\FST\),
\[
\emptyset\notin\Facnei_\FAG(s).
\]

\end{enumerate}

\end{lemma}

%%%%%%%%%%%%%%%
%%%%%%%%%%%%%%%
\begin{proof}

The implication from \(1\) to \(2\) is immediate.

Conversely, assume \(2\). Fix \(s\in\FST\) and
\(\FCC\subseteq\FAG\). Since \(\FACNF\) is coherently coverable, let
\(\Omega_s\subseteq\Gamma_s\) be a coherent profile cover at \(s\).

Suppose, toward a contradiction, that
\[
\emptyset\in\Facnei_\FCC(s).
\]
By the covering condition, there exists \(\gamma\in\Omega_s\) such that
\[
\gamma_\FCC=\emptyset.
\]
Since \(\gamma\) is a power profile and
\(\FCC\subseteq\FAG\), inclusion monotonicity gives
\[
\gamma_\FAG\subseteq\gamma_\FCC=\emptyset.
\]
Hence \(\gamma_\FAG=\emptyset\). By the membership condition for power
profiles,
\[
\emptyset=\gamma_\FAG\in\Facnei_\FAG(s),
\]
contrary to \(2\). Therefore
\[
\emptyset\notin\Facnei_\FCC(s).
\]
Since \(s\) and \(\FCC\) were arbitrary, \(1\) follows.

\end{proof}

%%%%%%%%%%%%%%%
%%%%%%%%%%%%%%%
\begin{definition}[Seriality and determinism of \(\PAC\)-representative actual neighborhood frames]
\label{def:31-seriality-and-determinism-of-pac-representative-actual-neighborhood-frames}

Let
\[
\FACNF=(\FST,\{\Facnei_\FCC \mid \FCC \subseteq \FAG\})
\]
be an \(\PAC\)-representative actual neighborhood frame.

\begin{itemize}

\item
\(\FACNF\) is \Fdefs{\(\PAC\)-serial} if, for every \(s \in \FST\) and every
\(\FCC \subseteq \FAG\),
\[
\Facnei_\FCC(s) \neq \emptyset.
\]

\item
\(\FACNF\) is \Fdefs{\(\PAC\)-deterministic} if, for every \(s \in \FST\) and
every \(X \in \Facnei_\FAG(s)\), the set \(X\) is a singleton.

\end{itemize}

\end{definition}

As before, we let the symbols \(\mathtt{S}\) and \(\mathtt{D}\) denote
seriality and determinism, respectively. We write
\[
\epsilon,\mathtt{S},\mathtt{D},\mathtt{SD}
\]
for the four possible combinations of these two properties, where \(\epsilon\)
denotes the case in which neither property is imposed.

For
\[
\FXX \in \{\epsilon,\mathtt{S},\mathtt{D},\mathtt{SD}\},
\]
an \(\PAC\)-representative actual neighborhood frame is a
\Fdefs{\(\PAC\)-representative actual neighborhood \(\FXX\)-frame} if it
satisfies every \(\PAC\)-property whose symbol occurs in \(\FXX\).

The following result gives an alternative definition for seriality that will be implicitly used later.

%%%%%%%%%%%%%%%
%%%%%%%%%%%%%%%
\begin{lemma}[Seriality reduces to the empty-coalition neighborhood]
\label{lem:11-seriality-reduces-to-the-empty-coalition-neighborhood}

Let \(\FACNF\) be an \(\PAC\)-representative actual neighborhood frame. Then
\(\FACNF\) is \(\PAC\)-serial iff, for every \(s \in \FST\),
\[
\Facnei_\emptyset(s) \neq \emptyset.
\]

\end{lemma}

%%%%%%%%%%%%%%%
%%%%%%%%%%%%%%%
\begin{proof}
The left-to-right direction is immediate, since \(\emptyset\subseteq\FAG\).

Conversely, suppose that
\[
\Facnei_\emptyset(s)\neq\emptyset
\]
for every \(s\in\FST\). Fix \(s\in\FST\) and \(\FCC\subseteq\FAG\). Since
\(\FACNF\) is \(\PAC\)-representative, let
\(\Omega_s\subseteq\Gamma_s\) be a coherent profile cover at \(s\). Choose
\(X\in\Facnei_\emptyset(s)\). By coverage, there is some
\(\gamma\in\Omega_s\) such that
\[
\gamma_\emptyset=X.
\]
Since \(\gamma\in\Gamma_s\), each component of \(\gamma\) is an actual power
for the corresponding coalition. In particular,
\[
\gamma_\FCC\in\Facnei_\FCC(s).
\]
Thus \(\Facnei_\FCC(s)\neq\emptyset\). Since \(s\) and \(\FCC\) were arbitrary,
\(\FACNF\) is \(\PAC\)-serial.
\end{proof}

%%%%%%%%%%%%%%%
%%%%%%%%%%%%%%%
\subsection{Actual representation without independence}
\label{subsec:07-02-actual-representation-without-independence}

%%%%%%%%%%%%%%%
%%%%%%%%%%%%%%%
\begin{theorem}[Equivalence of finite-agent \texorpdfstring{$\PAC$}{PAC}-representativeness and \texorpdfstring{$\PAC$}{PAC}-abstract-representativeness]
\label{thm:07-equivalence-of-finite-agent-pac-representativeness-and-pac-abstract-representativeness}

Let
\[
\FACNF=(\FST,\{\Facnei_\FCC\mid \FCC\subseteq\FAG\})
\]
be an actual neighborhood frame, where \(\FAG\) is finite and nonempty. Then
\(\FACNF\) is \(\PAC\)-representative in the finite-agent sense iff
\(\FACNF\) is \(\PAC\)-abstract-representative.

\end{theorem}

%%%%%%%%%%%%%%%
%%%%%%%%%%%%%%%
\begin{proof}

We prove the two directions separately.

\medskip
\noindent
\((\Rightarrow)\)
Assume that \(\FACNF\) is \(\PAC\)-representative in the finite-agent sense.
Then \(\FACNF\) satisfies actual triviality of the empty coalition, liveness,
and coherent coverability.

It remains to show abstract presentability. Fix \(s\in\FST\). By
Proposition~\ref{prop:02-coherent-coverability-via-the-coherent-core},
the coherent core \(\operatorname{Core}_s\) is a coherent profile cover at
\(s\).

There are two cases.

First suppose that
\[
\operatorname{Core}_s=\emptyset.
\]
Since \(\operatorname{Core}_s\) is a coherent profile cover at \(s\), its
covering condition gives, for every coalition \(\FCC\subseteq\FAG\),
\[
\Facnei_\FCC(s)
=
\{\gamma_\FCC\mid \gamma\in\operatorname{Core}_s\}
=
\emptyset.
\]
Define
\[
I_s=\emptyset,
\]
let \(o_s:I_s\to\mathcal P(\FST)\) be the empty map, and, for each
\(a\in\FAG\), let \(\equiv^s_a\) be the empty relation on \(I_s\). Then each
\(\equiv^s_a\) is an equivalence relation on \(I_s\). Grand compatibility is
vacuous, and for every coalition \(\FCC\subseteq\FAG\),
\[
\{\rho^s_\FCC(i)\mid i\in I_s\}
=
\emptyset
=
\Facnei_\FCC(s).
\]
Thus
\[
\mathbb I_s=(I_s,o_s,\{\equiv^s_a\}_{a\in\FAG})
\]
is an abstract actual presentation of \(\FACNF\) at \(s\).

Now suppose that
\[
\operatorname{Core}_s\neq\emptyset.
\]
Since actual triviality of the empty coalition holds at \(s\), and
\(\operatorname{Core}_s\) is a nonempty coherent profile cover at \(s\),
Lemma~\ref{lem:10-coherent-profile-covers-unfold-into-abstract-presentations}
yields an abstract actual presentation
\[
\mathbb I_s
\]
of \(\FACNF\) at \(s\).

Since \(s\in\FST\) was arbitrary, we obtain a family
\[
\mathbb I=(\mathbb I_s)_{s\in\FST}
\]
which is an abstract actual presentation of \(\FACNF\). Therefore \(\FACNF\)
satisfies abstract presentability. Together with liveness, this shows that \(\FACNF\) is
\(\PAC\)-abstract-representative.

\medskip
\noindent
\((\Leftarrow)\)
Assume that \(\FACNF\) is \(\PAC\)-abstract-representative. Then \(\FACNF\)
satisfies liveness. Let
\[
\mathbb I=(\mathbb I_s)_{s\in\FST}
\]
be an abstract actual presentation of \(\FACNF\). By
Lemma~\ref{lem:02-empty-coalition-triviality-from-abstract-presentability},
\(\FACNF\) is trivial for the empty coalition.

For each \(s\in\FST\), Lemma~\ref{lem:09-abstract-presentations-induce-coherent-profile-covers}
implies that the profile set induced by \(\mathbb I_s\) is a coherent profile
cover at \(s\). Hence \(\FACNF\) is coherently coverable.

Thus \(\FACNF\) satisfies actual triviality of the empty coalition, liveness,
and coherent coverability. Therefore \(\FACNF\) is \(\PAC\)-representative in
the finite-agent sense.

\end{proof}

The equivalence theorem lets us convert the abstract actual representation theorem
into a purely neighborhood-level representation theorem for the four classes
that do not impose independence of agents.

%%%%%%%%%%%%%%%
%%%%%%%%%%%%%%%
\begin{corollary}[Actual representation theorem for finite-agent frames without independence]
\label{cor:01-actual-representation-theorem-for-finite-agent-frames-without-independence}

Let \(\FAG\) be finite and nonempty, and let
\[
\FXX\in\{\epsilon,\mathtt{S},\mathtt{D},\mathtt{SD}\}.
\]
Over \(\FAG\), the general concurrent game \(\FXX\)-frames are characterized,
up to \(\PAC\)-representability, by the \(\PAC\)-representative actual
neighborhood \(\FXX\)-frames.

\end{corollary}

%%%%%%%%%%%%%%%
%%%%%%%%%%%%%%%
\begin{proof}

We prove the two directions.

\medskip
\noindent
\textbf{Havingness.}
Let \(\FGCGF\) be a general concurrent game \(\FXX\)-frame, and let
\(\FACNF\) be the actual neighborhood frame induced by \(\FGCGF\). By
Theorem~\ref{thm:05-abstract-actual-havingness-theorem}, \(\FACNF\) is a
\(\PAC\)-abstract-representative actual neighborhood \(\FXX\)-frame. In
particular, \(\FACNF\) is \(\PAC\)-abstract-representative. Hence, by
Theorem~\ref{thm:07-equivalence-of-finite-agent-pac-representativeness-and-pac-abstract-representativeness},
\(\FACNF\) is \(\PAC\)-representative.

Since \(\FXX\) contains no independence requirement, its only possible
additional requirements are seriality and determinism. These are the same
frame-level conditions in the abstract and coherent-cover formulations.
Therefore \(\FACNF\) is an \(\PAC\)-representative actual neighborhood
\(\FXX\)-frame.

\medskip
\noindent
\textbf{Enoughness.}
Let \(\FACNF\) be an \(\PAC\)-representative actual neighborhood
\(\FXX\)-frame. By
Theorem~\ref{thm:07-equivalence-of-finite-agent-pac-representativeness-and-pac-abstract-representativeness},
\(\FACNF\) is \(\PAC\)-abstract-representative. Again, since \(\FXX\) contains
no independence requirement, the only possible additional requirements are
seriality and determinism, and these transfer unchanged. Hence \(\FACNF\) is a
\(\PAC\)-abstract-representative actual neighborhood \(\FXX\)-frame.

By Theorem~\ref{thm:06-abstract-actual-enoughness-theorem}, there exists a general
concurrent game \(\FXX\)-frame \(\FGCGF\) such that \(\FGCGF\) is
\(\PAC\)-representable by \(\FACNF\).

\end{proof}

%%%%%%%%%%%%%%%
%%%%%%%%%%%%%%%
\subsubsection*{On finite actions}

Even when the state space is finite and the chosen coherent profile covers are
finite, the construction above does not in general yield a finite action set.
Indeed, the unfolding forest contains all finite labelled paths generated from
a coherent profile cover and may therefore contain infinitely many nodes.
More importantly, these nodes may determine infinitely many equivalence
classes for an agent. For example, repeatedly extending a path along edges
whose labels do not contain an agent \(a\) can produce nodes that are pairwise
inequivalent under \(\approx^s_a\).

Since the action set in the abstract enoughness construction is defined by
\[
\FAC^{\mathbb I}
=
\{*\}
\cup
\{\,(s,a,[i]^s_a)
\mid
s\in\FST,\ a\in\FAG,\ i\in I_s\,\},
\]
infinitely many agent-equivalence classes may give rise to infinitely many
actions. Thus, the representation theorem does not guarantee that the
representing action set can be chosen finite.

%%%%%%%%%%%%%%%
%%%%%%%%%%%%%%%
\subsubsection*{Comparison with weak and two-agent \texorpdfstring{$\PAC$}{PAC}-representativeness}

We briefly record the relationship between the finite-agent notion introduced
above and the weak finite-agent conditions of
Definition~\ref{def:14-weak-finite-agent-pac-representativeness}.

First, finite-agent \(\PAC\)-representativeness implies weak finite-agent
\(\PAC\)-representativeness. Indeed, by
Corollary~\ref{cor:01-actual-representation-theorem-for-finite-agent-frames-without-independence}
with \(\FXX=\epsilon\), every finite-agent \(\PAC\)-representative actual
neighborhood frame \(\PAC\)-represents some general concurrent game frame.
By Proposition~\ref{prop:01-necessity-of-the-weak-finite-agent-conditions}, every actual
neighborhood frame induced in this way satisfies the weak finite-agent
conditions.

Second, the converse fails already for three agents. The frame in
Example~\ref{ex:02-example} is weak finite-agent
\(\PAC\)-representative but does not \(\PAC\)-represent any general concurrent
game frame. Hence, again by
Corollary~\ref{cor:01-actual-representation-theorem-for-finite-agent-frames-without-independence},
it cannot be finite-agent \(\PAC\)-representative.

Third, for two agents the two notions coincide. When \(|\FAG|=2\), the weak
finite-agent conditions are exactly the two-agent \(\PAC\)-representative
conditions of
Definition~\ref{def:10-two-agent-pac-representative-actual-neighborhood-frames}.
By the two-agent actual enoughness theorem,
Theorem~\ref{thm:02-two-agent-actual-enoughness-theorem}, every such frame
\(\PAC\)-represents some general concurrent game frame. By the havingness
direction of
Corollary~\ref{cor:01-actual-representation-theorem-for-finite-agent-frames-without-independence},
it is finite-agent \(\PAC\)-representative. The converse is the first point.

%%%%%%%%%%%%%%%
%%%%%%%%%%%%%%%
\section{Alpha representation without assuming independence via actual representation}
\label{sec:08-alpha-representation-without-assuming-independence-via-actual-representation}

In this section, we show that, for alpha powers, the four classes of general
concurrent game frames that do not involve independence are represented by the
corresponding classes of \(\PAL\)-representative alpha neighborhood frames.
Here, \(\PAL\)-representativeness is understood as in
Definition~\ref{def:12-finite-agent-pal-representative-alpha-neighborhood-frames}.
The frame conditions used below, namely \(\PAL\)-seriality and
\(\PAL\)-determinism, are understood as in
Definition~\ref{def:13-seriality-independence-and-determinism-of-pal-representative-alpha-neighborhood-frames}.

The havingness direction for alpha powers over finite-agent general concurrent
game frames has already been recalled in
Theorem~\ref{thm:03-finite-agent-alpha-havingness-theorem}. Hence, it remains
to prove the enoughness direction.
The proof proceeds by reduction to the actual representation theorem obtained
in Subsection~\ref{subsec:07-02-actual-representation-without-independence}.
Given an \(\PAL\)-representative alpha neighborhood frame, we first extract an
actual basis whose upward closure recovers the alpha neighborhoods. We then show
that this actual basis is \(\PAC\)-representative, and apply the enoughness
direction of
Corollary~\ref{cor:01-actual-representation-theorem-for-finite-agent-frames-without-independence}.

We first record a simple consequence of the basic \(\PAL\)-representativeness
conditions. It shows that, at a fixed state, the emptiness of alpha
neighborhoods is independent of the coalition.

%%%%%%%%%%%%%%%
%%%%%%%%%%%%%%%
\begin{fact}[Emptiness transfer for alpha neighborhoods]
\label{fact:02-emptiness-transfer-for-alpha-neighborhoods}

Let
\[
\FALNF=(\FST,\{\Falnei_\FCC\mid \FCC\subseteq\FAG\})
\]
be an \(\PAL\)-representative alpha neighborhood frame. For all \(s\in\FST\)
and \(\FCC\subseteq\FAG\), we have
\[
\Falnei_\FCC(s)=\emptyset
\quad\text{if and only if}\quad
\Falnei_\emptyset(s)=\emptyset.
\]

\end{fact}

%%%%%%%%%%%%%%%
%%%%%%%%%%%%%%%
\begin{proof}

Let \(s\in\FST\) and \(\FCC\subseteq\FAG\).

First suppose that \(\Falnei_\FCC(s)=\emptyset\). Since
\(\emptyset\subseteq\FCC\), coalition monotonicity of alpha powers gives
\[
\Falnei_\emptyset(s)\subseteq\Falnei_\FCC(s),
\]
and hence \(\Falnei_\emptyset(s)=\emptyset\).

Conversely, suppose that \(\Falnei_\emptyset(s)=\emptyset\). Then the
nonmonotonic core \(\Falneinc_\emptyset(s)\) is empty, and hence
\[
\bigcup\Falneinc_\emptyset(s)=\emptyset.
\]
If, toward a contradiction, \(\Falnei_\FCC(s)\neq\emptyset\), choose
\(X\in\Falnei_\FCC(s)\). By groundedness of alpha powers, there exists
\(Y\in\Falnei_\FCC(s)\) such that
\[
Y\subseteq\bigcup\Falneinc_\emptyset(s)
\quad\text{and}\quad
Y\subseteq X.
\]
Thus \(Y=\emptyset\), contradicting liveness. Therefore
\(\Falnei_\FCC(s)=\emptyset\).

\end{proof}

The key step is to extract an actual basis from a given alpha neighborhood
frame.

\begin{lemma}[Actual-basis lemma for finite-agent alpha frames without assuming independence]
\label{lem:12-actual-basis-lemma-for-finite-agent-alpha-frames-without-assuming-independence}

Let \(\FAG\) be a finite nonempty set of agents, and let
\[
\FXX\in\{\epsilon,\mathtt{S},\mathtt{D},\mathtt{SD}\}.
\]
Let
\[
\FALNF=(\FST,\{\Falnei_\FCC\mid \FCC\subseteq\FAG\})
\]
be an \(\PAL\)-representative alpha neighborhood \(\FXX\)-frame. Then there
exists an \(\PAC\)-representative actual neighborhood \(\FXX\)-frame
\[
\FACNF=(\FST,\{\Facnei_\FCC\mid \FCC\subseteq\FAG\})
\]
such that, for every \(s\in\FST\) and every \(\FCC\subseteq\FAG\),
\[
\Falnei_\FCC(s)
=
\{Y\subseteq\FST
\mid
\text{there exists }X\in\Facnei_\FCC(s)\text{ such that }X\subseteq Y
\}.
\tag{U}
\label{eq:12-alpha-upward-closure-of-basis}
\]

\end{lemma}

%%%%%%%%%%%%%%%
%%%%%%%%%%%%%%%
\begin{proof}

We construct the actual basis state by state. Fix \(s\in\FST\) temporarily.

\medskip
\noindent
\textbf{1. Definition of the actual basis.}

Suppose first that
\[
\Falnei_\emptyset(s)=\emptyset.
\]
By Fact~\ref{fact:02-emptiness-transfer-for-alpha-neighborhoods},
\[
\Falnei_\FCC(s)=\emptyset
\qquad
\text{for every }\FCC\subseteq\FAG.
\]
In this case define
\[
\Facnei_\FCC(s)=\emptyset
\qquad
\text{for every }\FCC\subseteq\FAG.
\]

Now suppose that
\[
\Falnei_\emptyset(s)\neq\emptyset.
\]
By alpha triviality of the empty coalition, the nonmonotonic core
\(\Falneinc_\emptyset(s)\) is a singleton. Write its unique element as
\(T_s\):
\[
\Falneinc_\emptyset(s)=\{T_s\}.
\]
Since \(T_s\in\Falnei_\emptyset(s)\), liveness gives
\[
T_s\neq\emptyset.
\]

If \(\mathtt{D}\notin\FXX\), define, for every \(\FCC\subseteq\FAG\),
\[
\Facnei_\FCC(s)
=
\{X\in\Falnei_\FCC(s)\mid X\subseteq T_s\}.
\tag{B$_0$}
\label{eq:13-alpha-basis-nondet}
\]

If \(\mathtt{D}\in\FXX\), define
\[
\Facnei_\FCC(s)
=
\begin{cases}
\Falneinc_\FAG(s),
&
\text{if }\FCC=\FAG,\\[1mm]
\{X\in\Falnei_\FCC(s)\mid X\subseteq T_s\},
&
\text{if }\FCC\subsetneq\FAG.
\end{cases}
\tag{B$_1$}
\label{eq:14-alpha-basis-det}
\]

In the deterministic case we shall use the following claim.

\smallskip
\noindent
\emph{Claim.}
Suppose that \(\mathtt{D}\in\FXX\) and
\(\Falnei_\emptyset(s)\neq\emptyset\). Then
\[
T_s=\bigcup\Falneinc_\FAG(s),
\tag{D}
\label{eq:15-alpha-det-successors}
\]
and consequently, for every \(x\in T_s\),
\[
\{x\}\in\Falneinc_\FAG(s).
\tag{D$'$}
\label{eq:16-alpha-det-singletons}
\]
Moreover, every member of \(\Falneinc_\FAG(s)\) is contained in \(T_s\).

\smallskip
\noindent
\emph{Proof of the claim.}
Since \(\FALNF\) is \(\PAL\)-deterministic and
\(\Falneinc_\emptyset(s)=\{T_s\}\), determinism gives
\[
T_s
=
\bigcup\Falneinc_\emptyset(s)
\subseteq
\bigcup\Falneinc_\FAG(s).
\]
For the converse inclusion, let \(Z\in\Falneinc_\FAG(s)\). Then
\(Z\in\Falnei_\FAG(s)\). By groundedness, there is
\(Y\in\Falnei_\FAG(s)\) such that
\[
Y\subseteq T_s
\text{ and }
Y\subseteq Z.
\]
Since \(Z\) is \(\subseteq\)-minimal in \(\Falnei_\FAG(s)\), the inclusion
\(Y\subseteq Z\) forces \(Y=Z\). Hence \(Z\subseteq T_s\). As \(Z\) was
arbitrary,
\[
\bigcup\Falneinc_\FAG(s)\subseteq T_s.
\]
This proves \eqref{eq:15-alpha-det-successors}, and also shows that every
member of \(\Falneinc_\FAG(s)\) is contained in \(T_s\). Finally, if
\(x\in T_s\), then by \eqref{eq:15-alpha-det-successors} there is
\(Z\in\Falneinc_\FAG(s)\) with \(x\in Z\). By \(\PAL\)-determinism, \(Z\) is a
singleton, so \(Z=\{x\}\). This proves \eqref{eq:16-alpha-det-singletons}.
\hfill\(\triangleleft\)

\medskip
\noindent
\textbf{2. The basis generates exactly the given alpha neighborhoods.}

We prove \eqref{eq:12-alpha-upward-closure-of-basis}. Let \(s\in\FST\) and
\(\FCC\subseteq\FAG\).

If \(\Falnei_\emptyset(s)=\emptyset\), then both sides of
\eqref{eq:12-alpha-upward-closure-of-basis} are empty by construction and
Fact~\ref{fact:02-emptiness-transfer-for-alpha-neighborhoods}.

Assume, then, that \(\Falnei_\emptyset(s)\neq\emptyset\). Every member of
\(\Facnei_\FCC(s)\) belongs to \(\Falnei_\FCC(s)\): this is immediate from
\eqref{eq:13-alpha-basis-nondet} and from the proper-coalition part of
\eqref{eq:14-alpha-basis-det}, while for \(\FCC=\FAG\) in the deterministic
case it follows from
\(\Falneinc_\FAG(s)\subseteq\Falnei_\FAG(s)\). Since each
\(\Falnei_\FCC(s)\) is closed under supersets, the right-hand side of
\eqref{eq:12-alpha-upward-closure-of-basis} is contained in
\(\Falnei_\FCC(s)\).

For the converse inclusion, let \(Y\in\Falnei_\FCC(s)\).
First suppose either \(\mathtt{D}\notin\FXX\), or
\(\FCC\subsetneq\FAG\). By groundedness, there exists
\(X\in\Falnei_\FCC(s)\) such that
\[
X\subseteq T_s
\text{ and }
X\subseteq Y.
\]
By the relevant clause of \eqref{eq:13-alpha-basis-nondet} or
\eqref{eq:14-alpha-basis-det}, this means \(X\in\Facnei_\FCC(s)\). Hence
\(Y\) belongs to the right-hand side of
\eqref{eq:12-alpha-upward-closure-of-basis}.

It remains to consider the case
\[
\mathtt{D}\in\FXX
\text{ and }
\FCC=\FAG.
\]
Again by groundedness, choose \(X\in\Falnei_\FAG(s)\) such that
\[
X\subseteq T_s
\text{ and }
X\subseteq Y.
\]
Liveness gives \(X\neq\emptyset\). Pick \(x\in X\). Then \(x\in T_s\), so by
\eqref{eq:16-alpha-det-singletons},
\[
\{x\}\in\Falneinc_\FAG(s)=\Facnei_\FAG(s).
\]
Since \(\{x\}\subseteq Y\), the set \(Y\) again belongs to the right-hand
side of \eqref{eq:12-alpha-upward-closure-of-basis}. Thus
\eqref{eq:12-alpha-upward-closure-of-basis} holds for all coalitions.

\medskip
\noindent
\textbf{3. The basis is \(\PAC\)-representative.}

We verify actual triviality of the empty coalition, liveness, and coherent
coverability.

\smallskip
\noindent
\emph{Actual triviality of the empty coalition.}
Fix \(s\in\FST\). If \(\Facnei_\emptyset(s)=\emptyset\), there is nothing to
prove. Otherwise \(\Falnei_\emptyset(s)\neq\emptyset\), so
\(\Falneinc_\emptyset(s)=\{T_s\}\). Since \(\FAG\neq\emptyset\), we have
\(\emptyset\subsetneq\FAG\); hence even in the deterministic case,
\[
\Facnei_\emptyset(s)
=
\{X\in\Falnei_\emptyset(s)\mid X\subseteq T_s\}.
\]
The set \(T_s\) itself belongs to the right-hand side. Conversely, if
\(X\in\Falnei_\emptyset(s)\) and \(X\subseteq T_s\), then the minimality of
\(T_s\) in \(\Falnei_\emptyset(s)\) rules out \(X\subset T_s\). Hence
\(X=T_s\), and therefore
\[
\Facnei_\emptyset(s)=\{T_s\}.
\]

\smallskip
\noindent
\emph{Liveness.}
Every member of every \(\Facnei_\FCC(s)\) is a member of the corresponding
alpha neighborhood \(\Falnei_\FCC(s)\). This is clear from the construction,
including the grand-coalition deterministic clause because
\(\Falneinc_\FAG(s)\subseteq\Falnei_\FAG(s)\). Since \(\FALNF\) satisfies
liveness, \(\emptyset\notin\Facnei_\FCC(s)\) for every \(s\) and \(\FCC\).

\smallskip
\noindent
\emph{Coherent coverability.}
Fix \(s\in\FST\). If \(\Falnei_\emptyset(s)=\emptyset\), then
\(\Facnei_\FCC(s)=\emptyset\) for every \(\FCC\subseteq\FAG\). Hence
\(\Gamma_s=\emptyset\), since no power profile can choose an
empty-coalition component. Taking \(\Omega_s=\emptyset\), we have
\(\Omega_s\subseteq\Gamma_s\), coherence holds trivially, and for every
\(\FCC\subseteq\FAG\),
\[
\{\gamma_\FCC\mid \gamma\in\Omega_s\}
=
\emptyset
=
\Facnei_\FCC(s).
\]
Thus \(\Omega_s\) is a coherent profile cover at \(s\).

Assume from now on that \(\Falnei_\emptyset(s)\neq\emptyset\), so that
\(T_s\) is defined.

\smallskip
\noindent
\emph{Case 1: \(\mathtt{D}\notin\FXX\).}
In this case, two elementary membership facts will be used repeatedly. First,
for every \(\FDD\subseteq\FAG\),
\[
T_s\in\Facnei_\FDD(s),
\]
because \(T_s\in\Falnei_\emptyset(s)\), coalition monotonicity gives
\(T_s\in\Falnei_\FDD(s)\), and \(T_s\subseteq T_s\). Second, if
\(B\in\Facnei_{\FCC_0}(s)\) and \(\FCC_0\subseteq\FDD\), then
\[
B\in\Facnei_\FDD(s),
\]
because \(B\in\Falnei_{\FCC_0}(s)\), coalition monotonicity gives
\(B\in\Falnei_\FDD(s)\), and the definition of the basis gives
\(B\subseteq T_s\).

For every \(\FCC_0\subseteq\FAG\) and every
\(B\in\Facnei_{\FCC_0}(s)\), define a family
\[
\gamma^{\FCC_0,B}
=
(\gamma^{\FCC_0,B}_\FDD)_{\FDD\subseteq\FAG}
\]
by
\[
\gamma^{\FCC_0,B}_\FDD
=
\begin{cases}
B,
&
\text{if }\FCC_0\subseteq\FDD,\\[1mm]
T_s,
&
\text{if }\FCC_0\nsubseteq\FDD.
\end{cases}
\tag{ND}
\label{eq:17-nondet-profile}
\]
Let
\[
\Omega_s
=
\{
\gamma^{\FCC_0,B}
\mid
\FCC_0\subseteq\FAG,
B\in\Facnei_{\FCC_0}(s)
\}.
\]

We first show that each \(\gamma^{\FCC_0,B}\) is a power profile at \(s\).
The membership condition follows from the two membership facts just proved:
for a coalition \(\FDD\), the \(\FDD\)-component is \(B\) if
\(\FCC_0\subseteq\FDD\), and is \(T_s\) otherwise. For inclusion monotonicity,
let \(\FDD\subseteq\FEE\). If \(\FCC_0\subseteq\FDD\), then
\(\FCC_0\subseteq\FEE\), and both components are \(B\). If
\(\FCC_0\nsubseteq\FEE\), then also \(\FCC_0\nsubseteq\FDD\), and both
components are \(T_s\). The only remaining possibility is
\(\FCC_0\nsubseteq\FDD\) and \(\FCC_0\subseteq\FEE\), in which case
\[
\gamma^{\FCC_0,B}_\FEE=B\subseteq T_s=
\gamma^{\FCC_0,B}_\FDD.
\]
Thus \(\gamma^{\FCC_0,B}\in\Gamma_s\), and so \(\Omega_s\subseteq\Gamma_s\).

Coverage holds. Fix \(\FDD\subseteq\FAG\). If
\(\gamma=\gamma^{\FCC_0,B}\in\Omega_s\), then the explicit definition
\eqref{eq:17-nondet-profile} and the membership facts above give
\(\gamma_\FDD\in\Facnei_\FDD(s)\). Hence
\[
\{\gamma_\FDD\mid \gamma\in\Omega_s\}
\subseteq
\Facnei_\FDD(s).
\]
Conversely, if \(B\in\Facnei_\FDD(s)\), then
\(\gamma^{\FDD,B}\in\Omega_s\) and
\[
\gamma^{\FDD,B}_\FDD=B.
\]
Therefore
\[
\Facnei_\FDD(s)
=
\{\gamma_\FDD\mid \gamma\in\Omega_s\}
\qquad
\text{for every }\FDD\subseteq\FAG.
\]

It remains to prove coherence. Let
\(\gamma=\gamma^{\FCC_0,B}\in\Omega_s\) and let \(\FDD\subseteq\FAG\). We show
\[
\gamma_\FDD
=
\bigcup
\{
\delta_\FAG
\mid
\delta\in\Omega_s
\text{ and }
\delta\equiv^s_\FDD\gamma
\}.
\tag{ND-Coh}
\label{eq:18-nondet-coherent}
\]
For the inclusion from right to left, suppose that
\(\delta\in\Omega_s\) and \(\delta\equiv^s_\FDD\gamma\). Since
\(\delta\) is a power profile and \(\FDD\subseteq\FAG\), inclusion
monotonicity gives
\[
\delta_\FAG\subseteq\delta_\FDD=\gamma_\FDD.
\]
Thus the right-hand side of \eqref{eq:18-nondet-coherent} is contained in
\(\gamma_\FDD\).

For the reverse inclusion, first suppose \(\FCC_0\subseteq\FDD\). Then
\(\gamma_\FDD=B\), and also \(\gamma_\FAG=B\). Since
\(\gamma\equiv^s_\FDD\gamma\), the profile \(\gamma\) itself witnesses that
all elements of \(B\) lie in the union on the right-hand side of
\eqref{eq:18-nondet-coherent}.

Now suppose \(\FCC_0\nsubseteq\FDD\). Then \(\gamma_\FDD=T_s\). For every
\(\FEE\subseteq\FDD\), we also have \(\FCC_0\nsubseteq\FEE\), and hence
\(\gamma_\FEE=T_s\). The profile \(\gamma^{\emptyset,T_s}\) belongs to
\(\Omega_s\), because \(T_s\in\Facnei_\emptyset(s)\). Moreover, for every
\(\FEE\subseteq\FDD\),
\[
\gamma^{\emptyset,T_s}_\FEE=T_s=\gamma_\FEE,
\]
so \(\gamma^{\emptyset,T_s}\equiv^s_\FDD\gamma\). Since
\[
\gamma^{\emptyset,T_s}_\FAG=T_s=\gamma_\FDD,
\]
the whole set \(T_s\) is contained in the union on the right-hand side of
\eqref{eq:18-nondet-coherent}. This proves
\eqref{eq:18-nondet-coherent}. Hence \(\Omega_s\) is a coherent profile cover
at \(s\) in the non-deterministic case.

\smallskip
\noindent
\emph{Case 2: \(\mathtt{D}\in\FXX\).}
We again record the membership facts that will be used below. For every proper
coalition \(\FDD\subsetneq\FAG\),
\[
T_s\in\Facnei_\FDD(s),
\]
by the same empty-coalition and coalition-monotonicity argument as above. If
\(\FDD\subsetneq\FAG\), \(\FCC_0\subseteq\FDD\), and
\(B\in\Facnei_{\FCC_0}(s)\), then necessarily
\(\FCC_0\subsetneq\FAG\), and
\[
B\in\Facnei_\FDD(s)
\]
by coalition monotonicity and \(B\subseteq T_s\). Finally,
\eqref{eq:16-alpha-det-singletons} says that, for every \(y\in T_s\),
\[
\{y\}\in\Facnei_\FAG(s).
\]
Also, whenever \(B\in\Facnei_\FAG(s)=\Falneinc_\FAG(s)\), the claim above
shows \(B\subseteq T_s\), and \(\PAL\)-determinism makes \(B\) a singleton.

For every \(\FCC_0\subseteq\FAG\), every
\(B\in\Facnei_{\FCC_0}(s)\), and every \(x\in B\), define
\[
\gamma^{\FCC_0,B,x}
=
(\gamma^{\FCC_0,B,x}_\FDD)_{\FDD\subseteq\FAG}
\]
by
\[
\gamma^{\FCC_0,B,x}_\FDD
=
\begin{cases}
\{x\},
&
\text{if }\FDD=\FAG,\\[1mm]
B,
&
\text{if }\FDD\subsetneq\FAG
\text{ and }
\FCC_0\subseteq\FDD,\\[1mm]
T_s,
&
\text{if }\FDD\subsetneq\FAG
\text{ and }
\FCC_0\nsubseteq\FDD.
\end{cases}
\tag{Det}
\label{eq:19-det-profile}
\]
Let
\[
\Omega_s
=
\{
\gamma^{\FCC_0,B,x}
\mid
\FCC_0\subseteq\FAG,
B\in\Facnei_{\FCC_0}(s),
x\in B
\}.
\]

We show that each \(\gamma^{\FCC_0,B,x}\) is a power profile at \(s\). For
proper coalitions, membership follows from the first two membership facts just
stated. For the grand coalition, if \(\FCC_0=\FAG\), then
\(B\in\Facnei_\FAG(s)=\Falneinc_\FAG(s)\), so \(B\) is a singleton; since
\(x\in B\), we have \(\{x\}=B\in\Facnei_\FAG(s)\). If
\(\FCC_0\subsetneq\FAG\), then \(B\subseteq T_s\), so \(x\in T_s\); hence
\(\{x\}\in\Facnei_\FAG(s)\) by \eqref{eq:16-alpha-det-singletons}.

For inclusion monotonicity, let \(\FDD\subseteq\FEE\). If
\(\FEE\subsetneq\FAG\), the same three-case argument used in Case~1 gives
\[
\gamma^{\FCC_0,B,x}_\FEE
\subseteq
\gamma^{\FCC_0,B,x}_\FDD.
\]
It remains to consider \(\FEE=\FAG\). If \(\FDD=\FAG\), the inclusion is
trivial. If \(\FDD\subsetneq\FAG\), then
\(\gamma^{\FCC_0,B,x}_\FAG=\{x\}\). The \(\FDD\)-component is either \(B\),
which contains \(x\), or \(T_s\). In the latter case \(x\in T_s\): if
\(\FCC_0\subsetneq\FAG\), this follows from \(B\subseteq T_s\) and
\(x\in B\); if \(\FCC_0=\FAG\), it follows from
\(B\in\Falneinc_\FAG(s)\) and the claim that every such \(B\) is contained in
\(T_s\). Hence \(\{x\}\subseteq\gamma^{\FCC_0,B,x}_\FDD\). Thus
\(\gamma^{\FCC_0,B,x}\in\Gamma_s\), and so \(\Omega_s\subseteq\Gamma_s\).

Coverage holds. Fix \(\FDD\subseteq\FAG\). The inclusion
\[
\{\gamma_\FDD\mid \gamma\in\Omega_s\}
\subseteq
\Facnei_\FDD(s)
\]
follows from the membership verification above: if \(\FDD=\FAG\), the
component is a singleton \(\{x\}\in\Facnei_\FAG(s)\); if
\(\FDD\subsetneq\FAG\), the component is either \(B\in\Facnei_\FDD(s)\) or
\(T_s\in\Facnei_\FDD(s)\).

Conversely, let \(B\in\Facnei_\FDD(s)\). If \(\FDD\subsetneq\FAG\), then
liveness gives \(B\neq\emptyset\). Choose \(x\in B\). Then
\(\gamma^{\FDD,B,x}\in\Omega_s\) and
\[
\gamma^{\FDD,B,x}_\FDD=B.
\]
If \(\FDD=\FAG\), then
\(B\in\Facnei_\FAG(s)=\Falneinc_\FAG(s)\), so by \(\PAL\)-determinism
\(B=\{x\}\) for some \(x\). Hence
\(\gamma^{\FAG,B,x}\in\Omega_s\) and
\[
\gamma^{\FAG,B,x}_\FAG=\{x\}=B.
\]
Therefore
\[
\Facnei_\FDD(s)
=
\{\gamma_\FDD\mid \gamma\in\Omega_s\}
\qquad
\text{for every }\FDD\subseteq\FAG.
\]

It remains to prove coherence. Let
\(\gamma=\gamma^{\FCC_0,B,x}\in\Omega_s\) and let
\(\FDD\subseteq\FAG\). We prove
\[
\gamma_\FDD
=
\bigcup
\{
\delta_\FAG
\mid
\delta\in\Omega_s
\text{ and }
\delta\equiv^s_\FDD\gamma
\}.
\tag{Det-Coh}
\label{eq:20-det-coherent}
\]
As in Case~1, if \(\delta\equiv^s_\FDD\gamma\), then
\[
\delta_\FAG\subseteq\delta_\FDD=\gamma_\FDD,
\]
because \(\delta\) is a power profile. Hence the right-hand side of
\eqref{eq:20-det-coherent} is contained in \(\gamma_\FDD\).

For the reverse inclusion, distinguish three cases.

First suppose \(\FDD=\FAG\). Then \(\gamma_\FAG=\{x\}\). Since
\(\gamma\equiv^s_\FAG\gamma\), the profile \(\gamma\) itself places
\(x\) in the union on the right-hand side of
\eqref{eq:20-det-coherent}. Hence \(\gamma_\FAG\) is contained in that union.

Second suppose \(\FDD\subsetneq\FAG\) and \(\FCC_0\subseteq\FDD\). Then
\(\gamma_\FDD=B\). Let \(y\in B\). The profile
\(\gamma^{\FCC_0,B,y}\) belongs to \(\Omega_s\). Moreover, for every
\(\FEE\subseteq\FDD\), the \(\FEE\)-component of
\(\gamma^{\FCC_0,B,y}\) is determined only by \(\FCC_0\), \(B\), and whether
\(\FCC_0\subseteq\FEE\); it is independent of the chosen element \(x\) or
\(y\). Therefore
\[
\gamma^{\FCC_0,B,y}\equiv^s_\FDD\gamma.
\]
Its grand-coalition component is \(\{y\}\), so every \(y\in B\) belongs to the
union on the right-hand side of \eqref{eq:20-det-coherent}. Thus
\(B\) is contained in that union.

Finally suppose \(\FDD\subsetneq\FAG\) and \(\FCC_0\nsubseteq\FDD\). Then
\(\gamma_\FDD=T_s\). Let \(y\in T_s\). By
\eqref{eq:16-alpha-det-singletons},
\[
\{y\}\in\Facnei_\FAG(s).
\]
Hence \(\gamma^{\FAG,\{y\},y}\in\Omega_s\). For every
\(\FEE\subseteq\FDD\), we have \(\FAG\nsubseteq\FEE\), and also
\(\FCC_0\nsubseteq\FEE\) because \(\FEE\subseteq\FDD\). Therefore
\[
\gamma^{\FAG,\{y\},y}_\FEE=T_s=\gamma_\FEE.
\]
Thus \(\gamma^{\FAG,\{y\},y}\equiv^s_\FDD\gamma\). Its grand-coalition
component is \(\{y\}\), so \(y\) belongs to the union on the right-hand side of
\eqref{eq:20-det-coherent}. Since \(y\in T_s\) was arbitrary, \(T_s\) is
contained in that union.

This proves \eqref{eq:20-det-coherent}. Hence \(\Omega_s\) is a coherent
profile cover at \(s\) in the deterministic case.

Since \(s\) was arbitrary, \(\FACNF\) is coherently coverable. Together with
actual triviality of the empty coalition and liveness, this proves that
\(\FACNF\) is \(\PAC\)-representative.

\medskip
\noindent
\textbf{4. Preservation of seriality and determinism.}

Assume first that \(\mathtt{S}\in\FXX\). Since \(\FALNF\) is
\(\PAL\)-serial, \(\Falnei_\emptyset(s)\neq\emptyset\) for every
\(s\in\FST\). Hence \(T_s\) is defined at every state.

If \(\mathtt{D}\notin\FXX\), then for every \(s\in\FST\) and every
\(\FCC\subseteq\FAG\), the witness \(T_s\) belongs to \(\Facnei_\FCC(s)\), as
shown in Case~1 above. Thus every \(\Facnei_\FCC(s)\) is nonempty.

If \(\mathtt{D}\in\FXX\), the same witness \(T_s\) shows that
\(\Facnei_\FCC(s)\neq\emptyset\) for every proper coalition
\(\FCC\subsetneq\FAG\). For \(\FCC=\FAG\), equation
\eqref{eq:15-alpha-det-successors} and \(T_s\neq\emptyset\) imply
\(\Falneinc_\FAG(s)\neq\emptyset\). Since
\(\Facnei_\FAG(s)=\Falneinc_\FAG(s)\), the grand-coalition actual
neighborhood is nonempty as well. Therefore \(\FACNF\) is \(\PAC\)-serial
whenever \(\mathtt{S}\in\FXX\).

Assume next that \(\mathtt{D}\in\FXX\). If
\(\Falnei_\emptyset(s)=\emptyset\), then \(\Facnei_\FAG(s)=\emptyset\). If
\(\Falnei_\emptyset(s)\neq\emptyset\), then by construction
\[
\Facnei_\FAG(s)=\Falneinc_\FAG(s).
\]
Every member of \(\Falneinc_\FAG(s)\) is a singleton by
\(\PAL\)-determinism. Hence every member of \(\Facnei_\FAG(s)\) is a
singleton, and \(\FACNF\) is \(\PAC\)-deterministic.

No other frame condition occurs in
\(\{\epsilon,\mathtt{S},\mathtt{D},\mathtt{SD}\}\). Hence \(\FACNF\) is a
\(\PAC\)-representative actual neighborhood \(\FXX\)-frame, and the proof is
complete.

\end{proof}

We can now prove the enoughness direction for alpha powers. The argument applies
the actual representation theorem to the actual basis just constructed, and
then uses upward closure to recover the original alpha neighborhoods.

%%%%%%%%%%%%%%%
%%%%%%%%%%%%%%%
\begin{theorem}[Alpha enoughness theorem for finite-agent frames without assuming independence]
\label{thm:08-alpha-enoughness-theorem-for-finite-agent-frames-without-assuming-independence}

Let \(\FAG\) be a finite nonempty set of agents, and let
\[
\FXX\in\{\epsilon,\mathtt{S},\mathtt{D},\mathtt{SD}\}.
\]
Let
\[
\FALNF=(\FST,\{\Falnei_\FCC\mid \FCC\subseteq\FAG\})
\]
be an \(\PAL\)-representative alpha neighborhood \(\FXX\)-frame. Then there
exists a general concurrent game \(\FXX\)-frame
\[
\FGCGF=
(\FST,\FAC,\{\Fav_\FCC\mid\FCC\subseteq\FAG\},
\{\Fout_\FCC\mid\FCC\subseteq\FAG\})
\]
such that \(\FGCGF\) is \(\PAL\)-representable by \(\FALNF\).

\end{theorem}

%%%%%%%%%%%%%%%
%%%%%%%%%%%%%%%
\begin{proof}

By Lemma~\ref{lem:12-actual-basis-lemma-for-finite-agent-alpha-frames-without-assuming-independence},
there exists an \(\PAC\)-representative actual neighborhood \(\FXX\)-frame
\[
\FACNF=(\FST,\{\Facnei_\FCC\mid \FCC\subseteq\FAG\})
\]
such that, for every \(s\in\FST\) and every \(\FCC\subseteq\FAG\),
\[
\Falnei_\FCC(s)
=
\{Y\subseteq\FST
\mid
\text{there exists }X\in\Facnei_\FCC(s)\text{ such that }X\subseteq Y
\}.
\tag{*}
\label{eq:21-alpha-recovery-from-actual-basis}
\]

By the enoughness direction of
Corollary~\ref{cor:01-actual-representation-theorem-for-finite-agent-frames-without-independence},
there exists a general concurrent game \(\FXX\)-frame
\[
\FGCGF=
(\FST,\FAC,\{\Fav_\FCC\mid \FCC\subseteq\FAG\},
\{\Fout_\FCC\mid \FCC\subseteq\FAG\})
\]
such that \(\FGCGF\) is \(\PAC\)-representable by \(\FACNF\). Thus the actual
effectivity function \(\FACEF_\FCC\) induced by \(\FGCGF\) satisfies
\[
\FACEF_\FCC(s)=\Facnei_\FCC(s)
\]
for every \(s\in\FST\) and every \(\FCC\subseteq\FAG\).

Let \(\FALEF_\FCC\) be the alpha effectivity function induced by the same game.
By definition of alpha effectivity,
\[
\FALEF_\FCC(s)
=
\{Y\subseteq\FST
\mid
\text{there exists }X\in\FACEF_\FCC(s)
\text{ such that }X\subseteq Y
\}.
\]
Since \(\FACEF_\FCC(s)=\Facnei_\FCC(s)\), equation
\eqref{eq:21-alpha-recovery-from-actual-basis} gives
\[
\FALEF_\FCC(s)=\Falnei_\FCC(s)
\]
for every \(s\in\FST\) and every \(\FCC\subseteq\FAG\).

Hence \(\FGCGF\) is \(\PAL\)-representable by \(\FALNF\). Since \(\FGCGF\) is
also a general concurrent game \(\FXX\)-frame, the theorem is proved.

\end{proof}

%%%%%%%%%%%%%%%
%%%%%%%%%%%%%%%
\subsubsection*{On finite actions}

The alpha representation theorem likewise does not guarantee that the
representing action set can be chosen finite, even when the state space is
finite.

The proof of alpha enoughness proceeds by first extracting a
\(\PAC\)-representative actual basis from the given alpha neighborhood frame
and then applying the actual representation theorem to that basis. When
\(\FST\) is finite, the extracted actual basis, as well as the coherent profile
covers constructed for it, is finite. The possible infinitude arises only in
the subsequent actual representation step: unfolding a finite coherent profile
cover may still produce infinitely many agent-equivalence classes, which become
actions in the abstract enoughness construction.

Thus, the proof does not establish a finite-action representation theorem.
This concerns the construction used here and does not show that a different
finite-action representation is impossible.

%%%%%%%%%%%%%%%
%%%%%%%%%%%%%%%
\section{Amalgamation of power profiles is insufficient for independence}
\label{sec:09-amalgamation-of-power-profiles-is-insufficient-for-independence}

The abstract representation theorem for the classes with independence uses an index-level condition.  At that level, independence is
expressed by abstract amalgamation: two abstract witnesses for disjoint
coalitions must have a common extension.

It is natural to get a similar amalgamation condition at the level of power profiles.

%%%%%%%%%%%%%%%
%%%%%%%%%%%%%%%
\begin{definition}[Coherent covers with power-profile amalgamation]
\label{def:32-coherent-covers-with-power-profile-amalgamation}

Let
\[
\FACNF=(\FST,\{\Facnei_\FCC\mid \FCC\subseteq\FAG\})
\]
be an actual neighborhood frame, where \(\FAG\) is finite and nonempty, and fix
\(s\in\FST\). Let \(\Gamma_s\) be the set of power profiles at \(s\).

A coherent profile cover
\[
\Omega_s\subseteq\Gamma_s
\]
at \(s\) has \Fdefs{power-profile amalgamation} if, for all disjoint coalitions
\(\FCC,\FDD\subseteq\FAG\) and all profiles
\(\gamma,\delta\in\Omega_s\), there exists a profile
\(\lambda\in\Omega_s\) such that
\[
\gamma\equiv^s_\FCC\lambda
\text{ and }
\delta\equiv^s_\FDD\lambda.
\]

We say that \(s\) is \Fdefs{coherently coverable with power-profile
amalgamation} if there exists a coherent profile cover at \(s\) with
power-profile amalgamation.  We say that \(\FACNF\) is \Fdefs{coherently
coverable with power-profile amalgamation} if every state is coherently
coverable with power-profile amalgamation.

\end{definition}

The natural question is whether this profile-level strengthening can replace
abstract amalgamation in the representation theorem for the independence
classes
\[
\mathtt I,\qquad \mathtt{SI},\qquad \mathtt{ID},\qquad \mathtt{SID}.
\]
More precisely, suppose one defines a candidate \(\PAC\)-representative actual
neighborhood \(\mathtt I\)-frame by requiring actual triviality of the empty
coalition, liveness, and coherent coverability with power-profile amalgamation,
and then adds the usual \(\PAC\)-seriality and/or \(\PAC\)-determinism for the
other three classes. Does this yield the corresponding actual representation
theorems for independent game frames? The answer is no. Power-profile
amalgamation is necessary, but not sufficient.

We first show that this condition is necessary.

%%%%%%%%%%%%%%%
%%%%%%%%%%%%%%%
\begin{lemma}[Abstract amalgamation induces power-profile amalgamation]
\label{lem:13-abstract-amalgamation-induces-power-profile-amalgamation}

Let
\[
\FACNF=(\FST,\{\Facnei_\FCC\mid \FCC\subseteq\FAG\})
\]
be an actual neighborhood frame, where \(\FAG\) is finite and nonempty.  Let
\[
\mathbb I=(\mathbb I_s)_{s\in\FST}
\]
be an amalgamating abstract actual presentation of \(\FACNF\).  For each
\(s\in\FST\) and each \(i\in I_s\), let \(\widehat{i}\) be the power profile at
\(s\) defined by
\[
(\widehat{i})_\FCC=\rho^s_\FCC(i)
\qquad
\text{for every }\FCC\subseteq\FAG.
\]
Put
\[
\Omega_s=\{\widehat{i}\mid i\in I_s\}.
\]
Then \(\Omega_s\) is a coherent profile cover at \(s\) with power-profile
amalgamation.  Hence \(\FACNF\) is coherently coverable with power-profile
amalgamation.

\end{lemma}

%%%%%%%%%%%%%%%
%%%%%%%%%%%%%%%
\begin{proof}

By Lemma~\ref{lem:09-abstract-presentations-induce-coherent-profile-covers},
\(\Omega_s\) is a coherent profile cover at \(s\).  It remains to verify
power-profile amalgamation.

Let \(\gamma,\delta\in\Omega_s\), and let
\(\FCC,\FDD\subseteq\FAG\) be disjoint.  Choose \(i,j\in I_s\) such that
\[
\gamma=\widehat{i}
\text{ and }
\delta=\widehat{j}.
\]
Since \(\mathbb I\) is amalgamating, there exists \(k\in I_s\) such that
\[
i\equiv^s_\FCC k
\text{ and }
j\equiv^s_\FDD k.
\]
Let
\[
\lambda=\widehat{k}.
\]
If \(\FEE\subseteq\FCC\), then
Lemma~\ref{lem:01-basic-properties-of-presented-powers} gives
\[
\gamma_\FEE
=
\rho^s_\FEE(i)
=
\rho^s_\FEE(k)
=
\lambda_\FEE.
\]
Thus \(\gamma\equiv^s_\FCC\lambda\).  The same argument, using
\(j\equiv^s_\FDD k\), gives
\[
\delta\equiv^s_\FDD\lambda.
\]
Hence \(\Omega_s\) has power-profile amalgamation.

\end{proof}

We now show that the condition is not sufficient. The following example gives
an actual neighborhood frame which is coherently coverable with power-profile
amalgamation, live, trivial for the empty coalition, \(\PAC\)-serial, and
\(\PAC\)-deterministic, but which represents no independent general concurrent
game frame.

%%%%%%%%%%%%%%%
%%%%%%%%%%%%%%%
\begin{example}[Power-profile amalgamation does not suffice for game independence]
\label{ex:04-power-profile-amalgamation-does-not-suffice-for-game-independence}

Let
\[
\FAG=\{a,b,c\},
\qquad
\FST=\{s,u_0,u_1\},
\qquad
O=\{u_0,u_1\}.
\]
We define an actual neighborhood frame whose neighborhoods are the same at
every state \(t\in\FST\).  They are given by
\[
\begin{array}{c|c}
\FCC & \Facnei_\FCC(t)\\
\hline
\emptyset,\{a\},\{b\},\{c\}
&
\{O\}
\\[1mm]
\{a,b\},\{a,c\},\{b,c\},\{a,b,c\}
&
\{\{u_0\},\{u_1\}\}.
\end{array}
\]
Thus the empty coalition and each singleton have only the coarse actual power
\(O\), while every two-agent coalition and the grand coalition have exactly the
two singleton powers \(\{u_0\}\) and \(\{u_1\}\).

The frame is trivial for the empty coalition, since
\(\Facnei_\emptyset(t)=\{O\}\) for every \(t\). It is live, since all listed
powers are nonempty. It is also \(\PAC\)-serial, since every neighborhood is nonempty, and
\(\PAC\)-deterministic, since every grand-coalition power is a singleton.

\medskip

\noindent
\emph{Step 1: Coherent coverability with power-profile amalgamation.}
Fix \(t\in\FST\).  Since the neighborhoods do not depend on \(t\), the same
construction works at every state.  Define two power profiles
\(\gamma,\delta\in\Gamma_t\) as follows:
\[
\begin{array}{c|c|c|c|c}
&
\emptyset
&
\{a\},\{b\},\{c\}
&
\{a,b\},\{a,c\},\{b,c\}
&
\{a,b,c\}
\\
\hline
\gamma
&
O
&
O
&
\{u_0\}
&
\{u_0\}
\\[1mm]
\delta
&
O
&
O
&
\{u_1\}
&
\{u_1\}.
\end{array}
\]
Put
\[
\Omega_t=\{\gamma,\delta\}.
\]

We first check that \(\Omega_t\) is a coherent profile cover at \(t\).
Membership and inclusion monotonicity are immediate from the table of
neighborhoods.  Coverage is also immediate: the empty coalition and the
singletons have the unique power \(O\), while every two-agent coalition and the
grand coalition have exactly the two powers \(\{u_0\}\) and \(\{u_1\}\).

It remains to check coherence.  Let \(\lambda\in\Omega_t\) and
\(\FCC\subseteq\FAG\).

If \(|\FCC|\leq 1\), then \(\gamma\) and \(\delta\) agree below \(\FCC\).  Hence
the grand-coalition components of all profiles in \(\Omega_t\) agreeing with
\(\lambda\) below \(\FCC\) are \(\{u_0\}\) and \(\{u_1\}\).  Their union is
\(O\), and therefore
\[
\lambda_\FCC
=
O
=
\{u_0\}\cup\{u_1\}.
\]

If \(|\FCC|\geq 2\), then the \(\FCC\)-component distinguishes \(\gamma\) from
\(\delta\).  Hence the only profile in \(\Omega_t\) agreeing with \(\lambda\)
below \(\FCC\) is \(\lambda\) itself.  Since, in both \(\gamma\) and \(\delta\),
the \(\FCC\)-component and the grand-coalition component coincide, we have
\[
\lambda_\FCC=\lambda_{\{a,b,c\}}.
\]
Thus every profile in \(\Omega_t\) is supported by \(\Omega_t\), and
\(\Omega_t\) is coherent.

We next check that \(\Omega_t\) has power-profile amalgamation.  Let
\(\eta,\theta\in\Omega_t\), and let \(\FCC,\FDD\subseteq\FAG\) be disjoint.  We
need some \(\lambda\in\Omega_t\) such that \(\lambda\) agrees with \(\eta\)
below \(\FCC\) and agrees with \(\theta\) below \(\FDD\).

If both \(\FCC\) and \(\FDD\) have size at most one, then all profiles in
\(\Omega_t\) agree below both coalitions.  We may take \(\lambda=\eta\).

If \(|\FCC|\geq 2\), then \(|\FDD|\leq 1\), since \(\FCC\) and \(\FDD\) are
disjoint subsets of the three-agent set \(\FAG\).  Take \(\lambda=\eta\).  Then
\(\lambda\) agrees with \(\eta\) below \(\FCC\) trivially, and it agrees with
\(\theta\) below \(\FDD\), because all profiles in \(\Omega_t\) agree below
every coalition of size at most one.

The case \(|\FDD|\geq 2\) is symmetric.  Then \(|\FCC|\leq 1\), and we take
\(\lambda=\theta\).

Hence \(\Omega_t\) has power-profile amalgamation.  Since \(t\in\FST\) was
arbitrary, the frame is coherently coverable with power-profile amalgamation at
every state.  Moreover, it satisfies actual triviality of the empty coalition
and liveness.

\medskip

\noindent
\emph{Step 2: A constraint forced by independent game representation.}
Suppose, toward a contradiction, that there exists an independent general
concurrent game frame \(\FGCGF\) such that \(\FGCGF\) is
\(\PAC\)-representable by the actual neighborhood frame above.

Since
\[
\Facnei_{\{a\}}(s)
=
\Facnei_{\{b\}}(s)
=
\Facnei_{\{c\}}(s)
=
\{O\},
\]
every available individual action of \(a\), \(b\), or \(c\) has exact outcome
\(O\).  Since
\[
\Facnei_{\{a,b\}}(s)
=
\Facnei_{\{a,c\}}(s)
=
\Facnei_{\{b,c\}}(s)
=
\{\{u_0\},\{u_1\}\},
\]
every available two-agent action has exact outcome either \(\{u_0\}\) or
\(\{u_1\}\).

We record the key consequence.  For every available full action
\[
(\alpha,\beta,\chi)
\]
of \(a,b,c\) at \(s\), the three two-agent restrictions must have the same exact
outcome:
\begin{equation}
\tag{*}
\label{eq:22-three-pair-outcomes-must-agree}
\Fout_{\{a,b\}}(s,\alpha,\beta)
=
\Fout_{\{a,c\}}(s,\alpha,\chi)
=
\Fout_{\{b,c\}}(s,\beta,\chi).
\end{equation}
Indeed, by the ODA-condition,
\[
\Fout_{\{a,b,c\}}(s,\alpha,\beta,\chi)\neq\emptyset.
\]
By outcome monotonicity, the grand-coalition outcome is contained in the outcome
of each two-agent restriction.  Hence each two-agent restriction has nonempty
outcome and is available, again by the ODA-condition.  By
\(\PAC\)-representability, each of the three two-agent outcomes is one of
\(\{u_0\}\) and \(\{u_1\}\).  Since a nonempty set cannot be contained in both
\(\{u_0\}\) and \(\{u_1\}\), the three two-agent outcomes must coincide.

\medskip

\noindent
\emph{Step 3: The independence contradiction.}
By representability of the \(\{a,b\}\)-neighborhood, choose available
\(\{a,b\}\)-actions
\[
(\alpha_0,\beta_0)
\text{ and }
(\alpha_1,\beta_1)
\]
such that
\[
\Fout_{\{a,b\}}(s,\alpha_0,\beta_0)=\{u_0\},
\qquad
\Fout_{\{a,b\}}(s,\alpha_1,\beta_1)=\{u_1\}.
\]
Also choose an available action \(\chi\) of \(c\), which exists because
\(O\in\Facnei_{\{c\}}(s)\).

By independence, both full actions
\[
(\alpha_0,\beta_0,\chi)
\text{ and }
(\alpha_1,\beta_1,\chi)
\]
are available at \(s\).  Applying
\eqref{eq:22-three-pair-outcomes-must-agree} to the first full action gives
\[
\Fout_{\{a,c\}}(s,\alpha_0,\chi)=\{u_0\},
\qquad
\Fout_{\{b,c\}}(s,\beta_0,\chi)=\{u_0\}.
\]
Applying \eqref{eq:22-three-pair-outcomes-must-agree} to the second full action
gives
\[
\Fout_{\{a,c\}}(s,\alpha_1,\chi)=\{u_1\},
\qquad
\Fout_{\{b,c\}}(s,\beta_1,\chi)=\{u_1\}.
\]

Now consider the mixed tuple
\[
(\alpha_1,\beta_0,\chi).
\]
We first show that it is available.  Since
\[
\Fout_{\{a,b\}}(s,\alpha_1,\beta_1)=\{u_1\},
\]
outcome monotonicity gives
\[
\Fout_{\{a\}}(s,\alpha_1)\neq\emptyset.
\]
Hence \(\alpha_1\) is available for \(a\) by the ODA-condition.  Also,
\[
\Fout_{\{b,c\}}(s,\beta_0,\chi)=\{u_0\},
\]
so \((\beta_0,\chi)\) is available for \(\{b,c\}\), again by the
ODA-condition.  Since \(\{a\}\) and \(\{b,c\}\) are disjoint, independence
yields that
\[
(\alpha_1,\beta_0,\chi)
\]
is an available full action.

But the \(\{a,c\}\)-restriction of this full action has outcome
\[
\Fout_{\{a,c\}}(s,\alpha_1,\chi)=\{u_1\},
\]
whereas its \(\{b,c\}\)-restriction has outcome
\[
\Fout_{\{b,c\}}(s,\beta_0,\chi)=\{u_0\}.
\]
This contradicts \eqref{eq:22-three-pair-outcomes-must-agree}.  Therefore no
independent general concurrent game frame can be \(\PAC\)-representable by the
actual neighborhood frame constructed above.

\medskip

\noindent
\emph{Step 4: Separation from index-level abstract amalgamation.}
The same example also separates the profile-level condition from the index-level
condition.  The frame satisfies actual triviality of the empty coalition and
liveness, and, as shown above, it is coherently coverable with power-profile
amalgamation.

If it had an amalgamating abstract actual presentation, then it would be a
\(\PAC\)-abstract-representative actual neighborhood \(\mathtt{I}\)-frame.  By
Theorem~\ref{thm:06-abstract-actual-enoughness-theorem}, there would then exist an
independent general concurrent game frame \(\PAC\)-representable by this actual
neighborhood frame.  This contradicts the non-representability result proved
above.

Hence power-profile amalgamation does not imply index-level abstract
amalgamation.

\end{example}

The example shows that independence requires more crucial information than power-profile amalgamation can provide. The reason is that power profiles record only extensional coalition powers.
They do not remember the identity of the underlying individual action
components across different profiles.

This leaves a genuine open problem.  Does a purely neighborhood-theoretic
characterization of the independence classes exist, or must any complete
characterization reconstruct an index-like witness structure?

%%%%%%%%%%%%%%%
%%%%%%%%%%%%%%%
\section{Concluding remarks}
\label{sec:10-concluding-remarks}

This paper has studied the representation of actual and alpha powers in general
concurrent game frames by neighborhood frames. The general aim was to identify
conditions on neighborhood frames that are both necessary and sufficient for
representing powers of general concurrent game frames.

For actual powers, we considered representation conditions at two main levels.
At the index level, we introduced abstract actual presentations and the
corresponding notions of \(\PAC\)-abstract-representativeness. We showed that, for each of the eight classes of general
concurrent game frames determined by seriality, independence of agents, and
determinism, the corresponding notion of
\(\PAC\)-abstract-representativeness is necessary and sufficient for actual
representation.

At the power-profile level, we introduced coherent profile covers and the corresponding
notions of \(\PAC\)-representativeness not involving independence. For the four classes of general concurrent game
frames that do not impose independence, we showed that the corresponding
notions of \(\PAC\)-representativeness are equivalent to the corresponding
notions of \(\PAC\)-abstract-representativeness. Consequently, they are necessary and sufficient for actual
representation.

For alpha powers, we used the actual representation results to obtain
representation theorems for the corresponding four classes of general concurrent game frames not involving independence. In
particular, the previously introduced notions of \(\PAL\)-representativeness at
the level of alpha neighborhoods are necessary and sufficient.

Finally, we examined independence from the power-profile perspective. We
introduced a natural power-profile amalgamation condition and showed that it is
necessary but not sufficient.

Several questions remain open.
First, it remains to find an intrinsic
condition on actual neighborhood frames to represent actual powers arising from independent general concurrent game frames.
A second direction concerns finiteness. It would be useful to characterize, at
the neighborhood-frame level, the powers arising from general concurrent game
frames with finite state sets and finite action sets.
Finally, coherent coverability gives an exact index-free condition for actual
representability without independence, but it is still not a simple power-level condition in the two-agent style. It
remains open whether coherent coverability admits an equivalent
power-level reformulation.

%%%%%%%%%%%%%%%
%%%%%%%%%%%%%%%
\subsection*{Declaration of AI assistance in the research and writing process}

During the preparation of this work, the authors used several AI-assisted tools as part of an iterative research and writing
process. These tools were used to discuss definitions, examples, proof strategies, and presentations, and to assist with language
polishing and revision of \LaTeX{} code. The authors critically
reviewed, edited, and verified the resulting content, including all
definitions, examples, theorems, proofs, and references.
The authors take full responsibility for the whole paper.

%\subsection*{Acknowledgements}

\bibliographystyle{alpha}
\bibliography{Strategic-reasoning}

\end{document}